\documentclass[a4paper,sn-mathphys,pdflatex,twoside]{sn-jnl}
\pdfoutput=1

\makeatletter
\input{size12.clo}
\makeatother
\normalbaroutside
\makeatletter
{\catcode`\;=\active \gdef;{\ifmmode\semicolon\else\@semicolon\fi}}
\makeatother

\theoremstyle{thmstyleone}%
\newtheorem{theorem}{Theorem}[section]
\newtheorem{proposition}[theorem]{Proposition}%
\newtheorem{corollary}[theorem]{Corollary}%
\newtheorem{lemma}[theorem]{Lemma}%
\newtheorem{definition}[theorem]{Definition}%
\newtheorem{assumption}[theorem]{Assumption}%

\theoremstyle{thmstyletwo}%
\newtheorem{example}[theorem]{Example}%
\newtheorem{remark}[theorem]{Remark}%

\theoremstyle{thmstylethree}%

\renewcommand\theequation{\thesection.\arabic{equation}}

\newcommand{\narrowunderline}[1]{\mkern1mu\underline{\mkern-1mu#1\mkern-2mu}\mkern2mu }
\newcommand{\G}{\mathcal{G}}
\DeclareMathOperator{\Res}{Res}

\makeatletter
\@addtoreset{theorem}{section}
\@addtoreset{equation}{section}
\makeatother

\makeatletter
\def\ps@titlepage{%
      \def\@oddhead{%
      \vbox to 0pt{\vspace*{-38pt}%
    }}%
     \let\@evenhead\@oddhead%
     \def\@oddfoot{\vbox to 18pt{\vfill\reset@font\rmfamily\hfil\thepage\hfil}}
     \def\@evenfoot{}}%
\def\ps@headings{%
    \def\@oddfoot{\hfill}%
    \let\@evenfoot\@oddfoot%
      \def\@evenhead{%
      \vbox to 0pt{\vspace*{-39pt}%
         \hbox to \hsize{\hfill \hfill}}\par
      \hspace*{-\textwidth}\hbox to \hsize{\headerfont\thepage\qquad\rightmark\hfill}}%
      \def\@oddhead{%
      \vbox to 0pt{\vspace*{-39pt}%
         \hbox to \hsize{\hfill \hfill}}\par
      \hspace*{-\textwidth}\hbox to \hsize{\headerfont\hfill\leftmark\qquad\thepage}}%
      \let\@mkboth\markboth%
      }%
\makeatother

\allowdisplaybreaks[3]
\begin{document}

\title[A Laplacian to compute intersection numbers on 
$\overline{\mathcal{M}}_{g,n}$ and \dots]{A Laplacian to compute
  intersection numbers on 
  $\overline{\mathcal{M}}_{g,n}$ and correlation functions in NCQFT}

\author[1]{\fnm{Harald} \sur{Grosse}}\email{harald.grosse@univie.ac.at}
\author*[2]{\fnm{Alexander} \sur{Hock}\href{https://orcid.org/0000-0002-8404-4056}{\includegraphics[width=8pt]{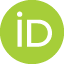}}}
\email{alexander.hock@maths.ox.ac.uk}
\author[3]{\fnm{Raimar} \sur{Wulkenhaar}\href{https://orcid.org/0000-0002-3371-9826}{\includegraphics[width=8pt]{ORCID-iD_icon-64x64.png}}}\email{raimar@math.uni-muenster.de}

\affil[1]{\orgdiv{Fakult\"at f\"ur Physik},
  \orgname{Universit\"at Wien},
  \orgaddress{\street{Boltzmanngasse 5},
  \postcode{1090} \city{Vienna},
  \country{Austria}}}

\affil[2]{\orgdiv{Mathematical Institute},
  \orgname{University of Oxford},
  \orgaddress{\street{Andrew Wiles Building, \\Woodstock Road},
    \postcode{OX2 6GG} \city{Oxford},
    \country{United Kingdom}}}

\affil[3]{\orgdiv{Mathematisches Institut},
  \orgname{Westf\"alische Wilhelms-Universit\"at},
  \orgaddress{\street{Einsteinstra\ss{}e 62},
  \postcode{48149} \city{M\"unster},
  \country{Germany}}}

\abstract{\unboldmath Let $F_g(t)$ be the generating function of
  intersection numbers of $\psi$-classes on the moduli spaces
  $\overline{\mathcal{M}}_{g,n}$ of stable complex curves of genus
  $g$. As by-product of a complete solution of all non-planar
  correlation functions of the renormalised $\Phi^3$-matrical QFT
  model, we explicitly construct a Laplacian $\Delta_t$ on a space of
  formal parameters $t_i$ which satisfies
  $\exp(\sum_{g\geq 2} N^{2-2g}F_g(t))=\exp((-\Delta_t+F_2(t))/N^2)1$
  as formal power series in $1/N^2$. The result is achieved via
  Dyson-Schwinger equations from noncommutative quantum field theory
  combined with residue techniques from topological recursion.  The
  genus-$g$ correlation functions of the $\Phi^3$-matricial QFT model
  are obtained by repeated application of another differential
  operator to $F_g(t)$ and taking for $t_i$ the renormalised moments
  of a measure constructed from the covariance of the model.}

\keywords{intersection numbers, matrix models, topological recursion, 
  Dyson-Schwinger equations, noncommutative geometry, quantum field theory}

\pacs[MSC Classification]{14C17, 32G15, 32G81, 81R60}

\maketitle
\markboth{\textit{A Laplacian to compute intersection numbers on 
    $\overline{\mathcal{M}}_{g,n}$ and \dots\hspace*{1.3cm}}}{%
  \textsc{H. Grosse, A. Hock \& R. Wulkenhaar}}

\section{Advertisement}

This paper completes the reverse engineering of a special quantum
field theory on noncommutative geometries. The final step could be of
interest in other areas of mathematics:
\begin{theorem}
\label{Thm:Zstable}
Let  
\begin{align*} 
F_g(t_0,t_2,t_3,\dots,t_{3g-2}) :=\sum_{(k)}
\frac{\langle\tau_2^{k_2}\tau_3^{k_3}\dots\tau_{3g-2}^{k_{3g-2}}
\rangle}{(1-t_0)^{2g-2+\sum_i k_i}}
\prod_{i=2}^{3g-2} \frac{t_i^{k_i}}{k_i!},\quad \sum_{i\geq 2} (i-1) k_i=3g-3,
\end{align*}
be the generating function of intersection numbers of
$\psi$-classes\footnote{Strictly speaking, $F_g$ generates
    intersection numbers of $\kappa$-classes on
    $\overline{\mathcal{M}}_{g,0}$, but these are in one-to-one
    correspondence with a subset of intersection numbers of
    $\psi$-classes. See Remark~\ref{rmk:intersection}.} on the
moduli spaces $\overline{\mathcal{M}}_{g,n}$ of stable complex curves
of genus $g$ \cite{Witten:1990hr, Kontsevich:1992ti}.  
The stable partition function satisfies, as a formal power
  series in $N^{-2}$,
\begin{align}
\fbox{$\displaystyle \exp\Big(\sum_{g=2}^\infty N^{2-2g} F_g(t)\Big)
=\exp\Big(-\frac{1}{N^2} \Delta_t +\frac{F_2(t)}{N^2}\Big) 1 $}
\end{align}
where $\displaystyle F_2(t)= 
\frac{7}{240} \cdot \frac{t_2^3}{3! T_0^5} + \frac{29}{5760} \frac{t_2
  t_3}{T_0^4} 
+ \frac{1}{1152} \frac{t_4}{T_0^3}$ 
generates the intersection numbers of genus 2, 
$T_0:=(1-t_0)$
and $\Delta_t=-\sum_{i,j}\hat{g}^{ij}\partial_i\partial_j-
\sum_{i}\hat{\Gamma}^{i}\partial_i$  is a Laplacian on 
the formal parameters $t_0,t_2,t_3,\dots$ given by
\begin{align*}
\Delta_t&:=-\Big(\frac{2t_2^3}{45 T_0^3} + \frac{37 t_2t_3}{1050 T_0^2} +
\frac{t_4}{210 T_0}\Big)\frac{\partial^2}{\partial t_0^2}
-\Big(
\frac{2t_2^3}{27 T_0^4} + \frac{1097 t_2t_3}{12600 T_0^3} 
+\frac{41 t_4}{2520 T_0^2}\Big)
\frac{\partial}{\partial t_0}
\nonumber
\\
&
-\sum_{k=2}^\infty \Big(
\Big(\frac{2t_2^2}{45 T_0^3} 
+ \frac{2 t_3}{105 T_0^2}\Big) t_{k+1}
+\frac{t_2 \mathcal{R}_{k+1}(t) }{2T_0}
+\frac{3 \mathcal{R}_{k+2}(t)}{2(3+2k)}
\Big)
\frac{\partial^2}{\partial t_k \partial t_0}
\nonumber
\\
&
-\sum_{k,l=2}^\infty \Big(
\frac{ t_2 t_{k+1}t_{l+1}}{90 T_0^2}
+ \frac{t_{k+1} \mathcal{R}_{l+1}(t)}{4T_0}
+ \frac{t_{l+1} \mathcal{R}_{k+1}(t)}{4T_0}
\nonumber
\\
&\hspace*{5cm}
+\frac{(1{+}2k)!!(1{+}2l)!!\mathcal{R}_{k+l+1}(t)}{4(1{+}2k{+}2l)!!}
\Big)
\frac{\partial^2}{\partial t_k\partial t_l}
\nonumber
\\
&
-\sum_{k=2}^\infty \Big(
\Big(\frac{19 t_2^2}{540 T_0^4} 
+ \frac{5 t_3}{252 T_0^3}\Big) t_{k+1}
+\frac{t_2 \mathcal{R}_{k+1}(t)}{48T_0^2}
+\frac{\mathcal{R}_{k+2}(t)}{16(3+2k) T_0}
+ \frac{t_2t_{k+2}}{90T_0^3}
\nonumber
\\
&\hspace*{5cm}
+\frac{\mathcal{R}_{k+2}(t)}{2T_0}
\Big)\frac{\partial}{\partial t_{k}}
\end{align*}
with  $\mathcal{R}_m(t):=  
\displaystyle \frac{2}{3} 
\sum_{k=1}^m \frac{(2m{-}1)!! \,kt_{k+1}}{(2k{+}3)!! T_0}  
\sum_{l=0}^{m-k} \frac{l!}{(m{-}k)!} B_{m-k,l}
\Big(\Big\{\frac{j!t_{j+1}}{(2j{+}1)!!T_0}\Big\}_{j=1}^{m-l+1}\Big)$.
\end{theorem}
\noindent
The $F_g(t)$ are recursively extracted from 
$\mathcal{Z}_g(t):=\frac{1}{(g-1)!} (-\Delta_t+F_2(t))^{g-1} 1$ via
\begin{align*}
F_g(t)&=\textstyle \mathcal{Z}_g(t)
- \frac{1}{(g-1)!}\sum_{k=2}^{g-1} B_{g-1,k}\big(
\big\{h! F_{h+1}(t)\}_{h=1}^{g-k}\big)
\\*
&= \textstyle \mathcal{Z}_g(t)- \frac{1}{(g-1)!}
\sum_{k=2}^{g-1} (-1)^{k-1}(k-1)!B_{g-1,k}\big(
\big\{h! \mathcal{Z}_{h+1}(t)\}_{h=1}^{g-k}\big).
\end{align*}
Here and in Theorem \ref{Thm:Zstable}, $B_{m,k}(\{x\})$ are the 
Bell polynomials (see Definition~\ref{def:Bell}). These equations are
easily implemented in any computer algebra system.

Theorem~\ref{Thm:Zstable} seems to be closely related with 
$\exp(\sum_{g\geq 0} F_g)= \exp(\hat{W})1$ proved by Alexandrov
\cite{Alexandrov:2010bn}\footnote{We thank Ga\"etan Borot for
bringing this reference to our attention.}, where
$\hat{W}:=\frac{2}{3}\sum_{k=1}^\infty (k+\frac{1}{2}) t_k
\hat{L}_{k-1}$ involves the generators $\hat{L}_n$ of the Virasoro
algebra. Including  $N$ and moving
$\exp(N^2 F_0+F_1)$ to the other side, our $\Delta_t$ is in principle
obtained via Baker-Campbell-Hausdorff formula from Alexandrov's
equation. Of course, evaluating the necessary commutators is not viable.

Theorem~\ref{Thm:Zstable} suggests several fascinating questions which we
haven't studied yet:
\begin{itemize}\itemsep 0pt
\item Is $\hat{\Gamma}^i$ a Levi-Civita connection for $\hat{g}^{ij}$, i.e.
$\hat{\Gamma}^i= \sum_j \hat{g}^{ij}\sqrt{\det  \hat{g}^{-1}}
\partial_j (\frac{1}{\sqrt{\det  \hat{g}^{-1}}})$? Here
$\det \hat{g}^{-1}$ would be the determinant of $( \hat{g}^{ij})$, whatever
this means. 

\item Is there a reasonable definition of a volume $\int dt 
  \frac{1}{\sqrt{\det  \hat{g}^{-1}(t)}}$?

\item Is it possible to prove that $\sum_{g=2}^\infty N^{2-2g} F_g(t)$
  is Borel summable for $t_l<0$?

\end{itemize}

Theorem~\ref{Thm:Zstable} is a by-product of our effort to 
construct the non-planar sector of
the renormalised $\Phi^3_D$-matricial quantum field theory in any dimension 
$D\in \{0,2,4,6\}$. These models are closely related to the Kontsevich model 
\cite{Kontsevich:1992ti} so that a link to intersection numbers 
is not surprising.

\section{Introduction}

Matrix models have a huge scope of applications
\cite{Brezin:1977sv,DiFrancesco:2004qj}, ranging from combinatorics
over 2D quantum gravity \cite{Gross:1989vs,DiFrancesco:1993cyw} up to
quantum field theory on noncommutative spaces
\cite{Langmann:2002cc,Langmann:2003if,Grosse:2004yu,
  Grosse:2005ig,Grosse:2006qv,Grosse:2006tc,Disertori:2006nq,
  Grosse:2012uv}. Of particular interest is the Kontsevich model
\cite{Kontsevich:1992ti}, which was designed to prove \mbox{Witten's}
conjecture \cite{Witten:1990hr} that the generating function of
intersection numbers of $\psi$-classes on the moduli space
$\overline{\mathcal{M}}_{g,n}$ of stable complex curves satisfies the
KdV equations. See also \cite{Witten:1991mn,
  Lando:2004??,Eynard:2016yaa}.

More recent investigations of matrix models led to the discovery of a 
universal structure called \textit{topological recursion}
\cite{Chekhov:2006vd,Eynard:2007kz}. Topological recursion was 
subsequently identified in 
many different areas of mathematics and theoretical physics
\cite{Eynard:2014zxa,Eynard:2016yaa}. The
Kontsevich model is, next to the Hermitian 1-matrix model
(to which it is related via  
Miwa transformation; see e.g.\ \cite{Ambjorn:1992gw}), the most
basic example for topological recursion.

On the other hand, renormalisation of quantum field theories on
noncommutative geometries generically leads to matrix models similar
to the Kontsevich model. The crucial difference is that convergence of
the (usually) formal series in the coupling constant is addressed, and
achieved by renormalisation
\cite{Grosse:2005ig,Grosse:2006qv,Grosse:2006tc}.  Renormalisation is
sensitive to the dimension encoded in the covariance of the matrix
model. For historical reasons, namely the perturbative renormalisation
\cite{Grosse:2004yu} of the $\Phi^4$-model on Moyal space and its
vanishing $\beta$-function \cite{Disertori:2006nq}, also the quartic
analogue of the Kontsevich model was intensely studied.  In
\cite{Grosse:2012uv} the simplest topological sector was reduced to a
closed equation for the 2-point function. This equation was recently
solved in \cite{Grosse:2019jnv} after understanding the pattern behind
the solution for the 2D-Moyal space \cite{Panzer:2018tvy}. On the
4D-Moyal space, the solution for the 2-point function was derived in
\cite{Grosse:2019qps}, which has resolved the triviality problem for
this specific noncommutative $\Phi^4_4$-model.  All correlation
functions with simplest topology (genus $g=0$ and one boundary) can be
explicitly described by a nested Catalan problem
\cite{deJong:2019oez}.

In \cite{Grosse:2016pob,Grosse:2016qmk} these methods developed for
the $\Phi^4$-model were reapplied to the cubic (Kontsevich-type)
model.  The new tools, together with the Makeenko-Semenoff solution
\cite{Makeenko:1991ec} of a non-linear integral equation, permitted an
exact solution of all planar (i.e.\ genus-$0$) renormalised
correlation functions in dimension $D\in\{2,4,6\}$. In particular,
exact (and surprisingly compact) formulae for planar correlation
functions with $B\geq 2$ boundary components were obtained. The
simplicity of the formulae \cite{Grosse:2016pob} for $B\geq 2$
suggests an underlying pattern.  It is traced back to the universality
phenomena captured by topological recursion\footnote{We thank Roland
  Speicher for the hint that there might be a relation between our
  work and topological recursion.}.  We refer to the book
\cite{Eynard:2016yaa}.

In this article we give the complete description of the non-planar 
sector of the renormalised $\Phi_D^3$-model. The notation 
defined in \cite{Grosse:2016qmk} will
be used and recalled in section~\ref{sec:setup}. We borrow from 
topological recursion the notational simplification to complex variables 
$z$ for the previous $\sqrt{X+c}$ and the vision that the 
correlation functions are holomorphic 
in $z \in \mathbb{C}\setminus \{0\}$
(see section \ref{sec:changeofvariables}). Knowing this, we proceed 
however in a different way. In section~\ref{sec:loop} we introduce
our main tool, a boundary creation 
operator $\hat{\mathrm{A}}^{\dag g}_{z_1,\dots,z_B}$ which, when applied to 
a genus-$g$ correlation function $\G_g(z_1|\dots|z_{B-1})$ 
with $B-1$ boundaries 
labelled $z_1,\dots,z_{B-1}$, creates a $B^{\text{th}}$ boundary 
labelled $z_B$. The existence of such an operator is suggested by the 
`loop insertion operator' in topological recursion \cite{Eynard:2016yaa}.

We rely on the sequence $\{\varrho_l\}_{l\in \mathbb{N}}$ of moments 
of a measure arising from the renormalised planar $1$-point function 
\cite{Grosse:2016pob,Grosse:2016qmk}. This sequence 
is uniquely defined by the renormalised covariance of the model, 
the renormalised coupling constant and the dimension $D\in \{0,2,4,6\}$. 
The boundary creation operator acts on Laurent polynomials in the $z_i$ 
with coefficients in rational functions of the $\varrho_l$. The heart 
of this paper is a combinatorial proof, independent of topological 
recursion,  that the boundary creation operator does what it should
(Theorem~\ref{finaltheorem}, portioned into Lemmas proved 
in an appendix). It is then (to our taste) considerably easier compared with 
topological recursion to derive in section~\ref{sec:non-planar-1pt}
structural results about the $\G_g(z_1|\dots|z_{B})$ such as 
the degree of the Laurent polynomials, the maximal number of occurring 
$\{\varrho_l\}$ and the weight of the rational function.

The main part of section \ref{sec:non-planar-1pt} is devoted to
the solution of $\G_g(z)$ for $g\geq 1$.
Starting point a Dyson-Schwinger type equation (\ref{1punktneuV}),
\begin{align}
    \hat{K}_z\G_g(z)+ \lambda\sum_{h=1}^{g-1} \G_{h}(z)\G_{g-h}(z)
+8\lambda^4 \hat{\mathrm{A}}^{\dag \,g{-}1}_{z,z}
\G_{g-1}(z)=0,
\label{DSE-G}
\end{align}
where $\hat{K}_z$ an integral operator. Thus, all $\G_g(z)$ can be
recursively evaluated if $\hat{K}_z$ can be inverted. Topological
recursion tells us that the inverse is a residue combined with a
special kernel operator. We give a direct combinatorial proof that the 
same method works in our case. 

We show in section~\ref{sec:freeenergy} that the $\G_g(z)$
arise for $g\geq 1$ by application of the boundary creation operator
to a uniquely defined `free energy' $F_g(\varrho)$. These
$F_g(\varrho)$ are characterised by 'only' $p(3g-3)$ rational numbers,
where $p(n)$ is the number of partitions of $n$. We show in
  section~\ref{sec:Laplacian} that \eqref{DSE-G} 
can be written as a second-order differential operator acting on
$\exp(\sum_{g\geq 1} N^{2-2g} F_g)$ in which it is convenient to
eliminate $F_1$. The result is Theorem~\ref{Thm:Zstable} expressed in
terms of $\varrho_0=1-t_0$ and $\varrho_l
=-\frac{t_l+1}{(2l+1)!!}$. In other words, to construct the non-planar
sector of the $\Phi^3_D$-matricial QFT model one has to replace the
formal parameters $t_l$ in the generating function $F_g(t)$ of
intersection numbers by precisely determined moments $\{\varrho_l\}$
resulting from the renormalisation of the planar sector of the model.
We finish with a discussion of the (deformed) Virasoro algebra in
section~\ref{sec:Virasoro} and a short summary (section~\ref{sec:summary}).

In the meantime an analogue of the boundary creation operator
  $\hat{\mathrm{A}}^{\dag g}_{z_J,z}$ was found for the matricial $\Phi^4$-model
  \cite{Hock:2020rje}. In this setting, two of us with J. Branahl
  introduced generalised correlation functions which satisfy an
  interwoven system of Dyson-Schwinger equations
  \cite{Branahl:2020yru}.  The solution for low $|\chi|$ provided
  strong evidence for the conjecture that the matricial $\Phi^4$-model
  (in that paper called quartic analogue of the Kontsevich model)
  satisfies blobbed topological recursion \cite{Borot:2015hna}, a
  generalisation of topological recursion. The proof that the genus
  $g=0$ sector satisfies blobbed topologcal recusion was achieved by
  two of us in \cite{Hock:2021tbl}, based on a functional relation
  related to the $x$-$y$-symmetry in topological recursion
  \cite{Eynard:2013csa,Hock:2022wer}. Some geometrical aspects of
  these generalised correlation functions were analysed in
  \cite{Branahl:2020uxs}.  Furthermore, the free energies $F_0$ and
  $F_1$ are computed in \cite{Branahl:2021uea}. All these results are
  reviewed in \cite{Branahl:2021slr}. It is also worth to mention that
  the complex $\Phi^4$-model with complex instead of Hermitian matrix
  (also known as the LSZ-model \cite{Langmann:2003if} from
  noncommutative geometry
  perspective) was proved by one of us with J. Branahl to be governed
  by topological recusion \cite{Branahl:2022uge}.

\section{Summary of previous results}

\label{sec:setup}
\subsection{Setup}

We follow \cite{Grosse:2016qmk} to introduce tuples
$\narrowunderline{n}=(n_1,n_2,..,n_{\frac{D}{2}})$
of non-negative integers. 
The number of tuples $\narrowunderline{n}$ of given
$|\narrowunderline{n}|=n_1+n_2+..+n_{\frac{D}{2}}$ 
is $\binom{|\narrowunderline{n}|+\frac{D}{2}-1}{\frac{D}{2}-1}$.
Let  $\mathbb{N}_\mathcal{N}^{D/2}:=\{
\narrowunderline{m}\in \mathbb{N}^{D/2}: |\narrowunderline{m}|\leq
\mathcal{N}\}$. The action of the $\Phi^3_D$-model is then defined by
\begin{align}
 S[\Phi]&=V\bigg(\sum_{\narrowunderline{n},\narrowunderline{m}
\in \mathbb{N}_\mathcal{N}^{D/2}} \!\!\!\!\!\!
Z\frac{H_{\narrowunderline{n}\narrowunderline{m}}}{2}
\Phi_{\narrowunderline{n}\narrowunderline{m}}\Phi_{\narrowunderline{m}\narrowunderline{n}}
 + \!\!\!\! \sum_{\narrowunderline{n}\in \mathbb{N}_\mathcal{N}^{D/2}}
\!\!\!\!
(\kappa
{+}\nu E_{\narrowunderline{n}}{+}\zeta E_{\narrowunderline{n}}^2) 
\Phi_{\narrowunderline{n}\narrowunderline{n}}
\nonumber
\\
&\qquad 
+\frac{\lambda_{bare}Z^{3/2}}{3} \!\!\!\!\!\!\!
\sum_{\narrowunderline{n},\narrowunderline{m},\narrowunderline{k}
\in \mathbb{N}_\mathcal{N}^{D/2}}
\!\!\!\!\!\!\!\!
\Phi_{\narrowunderline{n}\narrowunderline{m}}
 \Phi_{\narrowunderline{m}\narrowunderline{k}}
\Phi_{\narrowunderline{k}\narrowunderline{n}}\bigg),
\label{action3}
\raisetag{2mm}
\end{align}
where $H_{\narrowunderline{n}\narrowunderline{m}}:=
E_{\narrowunderline{n}}+E_{\narrowunderline{m}}$.  The constant $V$ is
first of all a formal parameter; for a noncommutative quantum field
theory model, $V=(\frac{\theta}{4})^{D/2}$ will be related to the
deformation parameter of the Moyal plane.  The parameters
$\lambda_{bare},\kappa,\nu,\zeta,Z$ and soon $\mu_{bare}$ are
$\mathcal{N}$-dependent renormalisation parameters. They will be
determined by normalisation conditions parametrised by physical mass
$\mu$ and coupling constant $\lambda$. The matrices
$(\Phi_{\narrowunderline{n}\narrowunderline{m}})$ are multi-indexed
Hermitian matrices, $\Phi_{\narrowunderline{n}\narrowunderline{m}}
=\overline{\Phi_{\narrowunderline{m}\narrowunderline{n}}}$. The
external matrix
$E=(E_{\narrowunderline{m}}\delta_{\narrowunderline{n},\narrowunderline{m}})$
can be assumed to be diagonal and has the eigenvalues
$E_{\narrowunderline{n}}=\frac{\mu_{bare}^2}{2}
+e\big(\frac{|\narrowunderline{n}|}{\mu^2 V^{2/D}}\big)$, where $e(x)$
is a monotonously increasing differentiable function with $e(0)=0$ (on
the noncommutative Moyal plane, $e(x)=x$).

The next step is to define (and rearrange) the partition function 
\begin{align}\label{Zustandssumme}
  \mathcal{Z}[J]&=\int \mathcal{D} \Phi
  \exp\left(-S[\Phi]+V\mathrm{Tr}(J\Phi) \right)
  \\
 &=
 \int \mathcal{D}\Phi\exp\Big(-VZ
 \sum_{\narrowunderline{n},\narrowunderline{m}\in \mathbb{N}_\mathcal{N}^{D/2}}
 \frac{H_{\narrowunderline{n}\narrowunderline{m}}}{2}
 \Phi_{\narrowunderline{n}\narrowunderline{m}}
 \Phi_{\narrowunderline{m}\narrowunderline{n}}\Big)
\nonumber
\\
&
\times \exp\Big(-\frac{\lambda_{bare}Z^{3/2} }{3V^2}
\sum_{\narrowunderline{n},\narrowunderline{m},\narrowunderline{k}
  \in \mathbb{N}_\mathcal{N}^{D/2}}
\frac{\partial^3}{\partial J_{\narrowunderline{n}\narrowunderline{m}}
  \partial J_{\narrowunderline{m}\narrowunderline{k}}
  \partial J_{\narrowunderline{k}\narrowunderline{n}}}
 \Big)\bigg[
 \nonumber
 \\
&\quad \exp\bigg(V\hspace*{-4mm}
\sum_{\narrowunderline{n},\narrowunderline{m}\in \mathbb{N}_\mathcal{N}^{D/2}}
\hspace*{-5mm}
\frac{(J_{\narrowunderline{n}\narrowunderline{m}}
  {-}(\kappa{+}\nu E_{\narrowunderline{n}}{+}\zeta E_{\narrowunderline{n}}^2)
  \delta_{\narrowunderline{m},\narrowunderline{n}})
  (J_{\narrowunderline{m}\narrowunderline{n}}{-}(\kappa{+}\nu E_{\narrowunderline{n}}
  {+}\zeta E_{\narrowunderline{n}}^2)
  \delta_{\narrowunderline{m},\narrowunderline{n}})}{
  2ZH_{\narrowunderline{n}\narrowunderline{m}}}\bigg)\bigg],
\nonumber
\end{align}
where source $(J_{\narrowunderline{n}\narrowunderline{m}})$ is 
a multi-indexed Hermitian matrix of rapidly decaying entries.

The correlation functions are defined as moments of the 
partition function. It turns out by earlier work \cite{Grosse:2012uv} 
that the correlation
functions expand into multi-cyclic contributions. It is therefore 
convenient to work with 
$\mathbb{J}_{\narrowunderline{p}_1...\narrowunderline{p}_{N_\beta}}
:=\prod_{j=1}^{N_\beta} J_{\narrowunderline{p}_j\narrowunderline{p}_{j+1}}$ 
with $\narrowunderline{p}_{N_\beta+1}\equiv \narrowunderline{p}_{1}$. 
Taking into account that genus-$g$ correlation functions 
scale with 
$V^{-2g}$ \cite{Brezin:1977sv,Grosse:2012uv}, the following expansion of 
the partition function is obtained:
\begin{align}\label{Entwicklungskoeffizienten}
  &\log\frac{\mathcal{Z}[J]}{\mathcal{Z}[0]}
  \\*[-2ex]
  &=:\sum_{B=1}^\infty\sum_{1\leq N_1 \leq ...\leq N_B}^\infty
\sum_{\narrowunderline{p}_i^\beta \in \mathbb{N}_\mathcal{N}^{D/2}}
\sum_{g=0}^{\infty}V^{2-B-2g}\frac{G^{(g)}_{|\narrowunderline{p}_1^1...\narrowunderline{p}^1_{N_1}|...|\narrowunderline{p}_1^B...\narrowunderline{p}^B_{N_B}|}}{
\prod_{i=1}^{N_B} \nu_i!}
\prod_{\beta=1}^B\frac{\mathbb{J}_{\narrowunderline{p}^\beta_1...\narrowunderline{p}^\beta_{N_\beta}}}{N_\beta}, \nonumber
\end{align}
where $\nu_i=\# \{\beta\in \{1,...,B\}\;:~N_\beta=i\}$.
We call the cumulant $G^{(g)}_{|\narrowunderline{p}_1^1...
\narrowunderline{p}^1_{N_1}|...|\narrowunderline{p}_1^B...
\narrowunderline{p}^B_{N_B}|}$ an $(N_1+...+N_B)$-point function 
of genus $g$; when the $N_\beta$ do not matter, a correlation function 
of genus $g$ with $B$ boundary components. hese
$(N_1+...+N_B)$-point functions have a perturbative expansion
into ribbon graphs drawn on genus-$g$ Riemann surfaces with $B$
boundary components. These ribbon graphs are dual to maps and as
such can also be studied from a point of view of
enumerative geometry.

Finally, we recall from \cite{Grosse:2016qmk} the Ward-Takahashi 
identity for $|\narrowunderline{q}|\neq |\narrowunderline{p}|$
\begin{align}
  \sum_{\narrowunderline{m}\in \mathbb{N}_\mathcal{N}^{D/2}}
  \frac{\partial^2}{\partial J_{\narrowunderline{q}\narrowunderline{m}}
    \partial J_{\narrowunderline{m}\narrowunderline{p}}}\mathcal{Z}[J]
  &=
  \sum_{\narrowunderline{m}\in \mathbb{N}_\mathcal{N}^{D/2}}
  \frac{V}{(E_{\narrowunderline{q}}-E_{\narrowunderline{p}})Z}
  \Big(J_{\narrowunderline{m}\narrowunderline{q}}
  \frac{\partial}{\partial J_{\narrowunderline{m}\narrowunderline{p}}}
  -J_{\narrowunderline{p}\narrowunderline{m}}
  \frac{\partial}{\partial J_{\narrowunderline{q}\narrowunderline{m}}}\Big)
 \mathcal{Z}[J]
 \nonumber
\\ 
&-\frac{V}{Z}(\nu+\zeta H_{\narrowunderline{p}\narrowunderline{q}})
\frac{\partial\mathcal{Z}[J]}{\partial J_{\narrowunderline{q}\narrowunderline{p}}}.
\label{Wardaa}
\end{align}
It arises from invariance of the partition function
under unitary transformation $\Phi \mapsto U^\dagger \Phi U$ 
of the integration variable \cite{Disertori:2006nq}, or 
directly from the structure of $\mathcal{Z}[J]$ \cite{Hock:2018wup}.

\subsection{Dyson-Schwinger equation for $B=1$}

The Dyson-Schwinger equations are determined in \cite{Grosse:2016pob,
  Grosse:2016qmk}
for $g=0$ and solved for all planar correlation functions.
Inserting there the formal genus expansion
$  G_{|\narrowunderline{p}_1^1...\narrowunderline{p}^1_{N_1}|...
  |\narrowunderline{p}_1^B...\narrowunderline{p}^B_{N_B}|}
:=\sum_{g=0}^\infty V^{-2g}
  G^{(g)}_{|\narrowunderline{p}_1^1...\narrowunderline{p}^1_{N_1}|...|
    \narrowunderline{p}_1^B...\narrowunderline{p}^B_{N_B}|}$
it is straightforward to extract from \cite{Grosse:2016pob,
  Grosse:2016qmk} Dyson-Schwinger equations for the 
cumulants
$G^{(g)}_{|\narrowunderline{p}_1^1...\narrowunderline{p}^1_{N_1}|...|
    \narrowunderline{p}_1^B...\narrowunderline{p}^B_{N_B}|}$.
It turns out that the planar 1-point function plays a special r\^ole; we
shift it to
\begin{align}
 \frac{W^{(0)}_{|\narrowunderline{p}|}}{2\lambda}
&:=G^{(0)}_{|\narrowunderline{p}|}
+\frac{F_{\narrowunderline{p}}}{\lambda},
&
\label{EF}
F_{\narrowunderline{p}}&:=E_{\narrowunderline{p}}-\frac{\lambda \nu}{2}
=\frac{\mu^2}{2}+e\Big(\frac{|\narrowunderline{p}|}{\mu^2 V^{2/D}}\Big).
\end{align}%
Three relations between renormalisation parameters are 
immediate \cite{Grosse:2016qmk}:
$\lambda=Z^{1/2}\lambda_{bare}$, $\frac{\lambda\zeta }{Z}=1-\frac{1}{Z}$
and $\mu_{bare}^2=\mu^2 + \lambda \nu$. 
Now we can use all the Dyson-Schwinger equations evaluated in
\cite{Grosse:2016qmk}. The planar 1-point function satisfies
\cite[eq.~(3.12)]{Grosse:2016qmk}
\begin{align}\label{1PunktFkt}
  (W^{(0)}_{|\narrowunderline{p}|})^2+2\lambda \nu W^{(0)}_{|\narrowunderline{p}|}
  +\frac{2\lambda^2}{V}\sum_{\narrowunderline{n}\in \mathbb{N}_\mathcal{N}^{D/2}}
  \frac{W^{(0)}_{|\narrowunderline{p}|}-W^{(0)}_{|\narrowunderline{n}|}}{
    F^2_{\narrowunderline{p}}-F^2_{\narrowunderline{n}}}
&=\frac{4F_{\narrowunderline{p}}^2}{Z}+C,
\end{align}%
where $C:=-\frac{\lambda^2\nu^2(1+Z)+4 \kappa \lambda}{Z}$,
whereas for $g\geq 1$ we get
\begin{align}\label{g11}
  W^{(0)}_{|\narrowunderline{p}|} G^{(g)}_{|\narrowunderline{p}|}
  +\lambda\nu G^{(g)}_{|\narrowunderline{p}|}
+\lambda  \sum_{h=1}^{g-1}
G^{(h)}_{|\narrowunderline{p}|} G^{(g-h)}_{|\narrowunderline{p}|}
+\frac{\lambda^2}{V} \hspace*{-3mm}
\sum_{\narrowunderline{n}\in \mathbb{N}_\mathcal{N}^{D/2}} \hspace*{-2mm}
\frac{G^{(g)}_{|\narrowunderline{p}|}-G^{(g)}_{|\narrowunderline{n}|}}{
  F^2_{\narrowunderline{p}}-F^2_{\narrowunderline{n}}}
+\lambda G^{(g-1)}_{|\narrowunderline{p}|\narrowunderline{p}|}=0.
\raisetag{1ex}
\end{align}

\subsection{Integral equations}

Introducing the measure
\begin{align}
\varrho(X):= \frac{2(2\lambda)^2}{V}
\sum_{\narrowunderline{n}\in \mathbb{N}_\mathcal{N}^{D/2}}
\delta(X-4F_{\narrowunderline{n}}^2),
\label{measure0}
\end{align}
we can rewrite \eqref{g11} as an integral equation.
The measure has
support in $[4 F_{\narrowunderline{0}}^2,\Lambda_{\mathcal{N}}^2]$
where $\Lambda_{\mathcal{N}}^2=\max(4F_{\narrowunderline{n}}^2\;:~
|\narrowunderline{n}|=\mathcal{N})$.  For quantum field theory it is
necessary to take a large-$\mathcal{N}$ limit. In general this
produces divergences which need renormalisation. Optionally the
large-$\mathcal{N}$ limit can be entangled with a limit $V\to \infty$
which, supposing the $F_{\narrowunderline{n}}$ scale down with $V$ (as
e.g.\ in \eqref{EF}), can be designed to let $\varrho(X)$ converge to
a continuous function. We also pass
to mass-dimensionless quantities via multiplication by specified 
powers of $\mu$ \cite{Grosse:2016qmk}. This amounts to choose the mass scale 
as $\mu=1$. 

To keep maximal flexibility we consider a measure $\varrho$ with
support in $[1,\Lambda^2]$ of which a limit $\Lambda\to \infty$ has to
be taken for quantum field theory.  As already observed in
\cite{Makeenko:1991ec}, equation \eqref{1PunktFkt} \emph{extends to
  the closed equation}
\begin{align}\label{1punktcomplex}
(W_0(X))^2 +2\lambda\nu W_0(X)
+\int_1^{\Lambda^2} dY\,\varrho(Y)\frac{W_0(X)-W_0(Y)}{X-Y}
&=\frac{X}{Z}+C
\end{align}
\emph{for  a sectionally holomorphic function $W_0(X)$} from
which one extracts 
$W^{(0)}_{|\narrowunderline{p}|}=W_0(4F_{\narrowunderline{p}}^2)$. 
The corresponding relation
\begin{align}
G^{(g)}_{|\narrowunderline{p}_1^1...\narrowunderline{p}^1_{N_1}|...|
\narrowunderline{p}_1^B...\narrowunderline{p}^B_{N_B}|} 
&= G_g\big(4F_{\narrowunderline{p}_1^1}^2,...,4F_{\narrowunderline{p}^1_{N_1}}^2|...|
4F_{\narrowunderline{p}_1^B}^2,...,4F_{\narrowunderline{p}^B_{N_B}}^2\big),
\label{npunktcomplex}
\end{align}
for $2g+B>1$ extends \eqref{g11} to
\begin{align}
0&=W_0(X) G_g(X)
  +\lambda\nu G_g(X)
+\frac{1}{2}
\int_1^{\Lambda^2} dY\,\varrho(Y)\frac{G_g(X)-G_g(Y)}{X-Y}
\nonumber
\\
&+\lambda  \sum_{h=1}^{g-1}
G_h(X) G_{g-h}(X)
+\lambda G_{g-1}(X|X).\label{gpunkt}
\end{align}

Using techniques for boundary values of sectionally holomorphic 
functions \cite{Makeenko:1991ec}, easily adapted to 
include $Z-1,\nu,C\neq 0$ \cite{Grosse:2016qmk}, 
one obtains the following solution of \eqref{1punktcomplex}:
\begin{align}\label{1g0}
 W_0(X)=&\frac{\sqrt{X+c}}{\sqrt{Z}}-\lambda\nu 
+\frac{1}{2}\int_1^{\Lambda^2} dY\frac{\varrho(Y)}{
(\sqrt{X+c}+\sqrt{Y+c})\sqrt{Y+c}}.
\end{align}
Here, the finite parameter $c$ and the (for $\Lambda^2\to \infty$) possibly
divergent $Z,\nu$ are determined by 
renormalisation conditions depending on the dimension:
\begin{align*}
\underbrace{ W_0(1)=1}_{D\geq 2},\qquad 
\underbrace{W_0'(1)=\frac{1}{2}}_{D\geq 4},\qquad
\underbrace{W_0''(1)=-\frac{1}{4}}_{D=6},
\end{align*}
together with the convention $Z=1$ for $D\in \{2,4\}$ and $\nu=0$ for
$D=2$. For given coupling constant $\lambda$ as the only remaining
parameter, these equations can be solved for $c$:
\begin{align}\label{implicit}
  (1-\sqrt{1{+}c})
\Big(\frac{1+\sqrt{1{+}c}}{2}\Big)^{\!\!\delta_{D,0}{+}\delta_{D,6}}
=\frac{1}{2}\int_{1}^{\Lambda^2}dY\frac{\varrho(Y)}{
(\sqrt{1{+}c}+\sqrt{Y{+}c})^{D/2}\sqrt{Y{+}c}}.
\end{align} 
By the implicit function theorem, \eqref{implicit} has a smooth
solution in an inverval $-\lambda_c<\lambda<\lambda_c$, 
in any dimension $D\in \{0,2,4,6\}$. 
The Lagrange inversion theorem gives the 
expansion of $c$ in $\lambda^2$: 
\begin{align*}
c=\sum_{n=1}^\infty \frac{1}{n!}\frac{d^{n-1}}{dw^{n-1}}
\Bigg\vert_{w=0}\bigg(\frac{\frac{w}{2}\int_{1}^{\Lambda^2}dY
\frac{\varrho(Y)}{(\sqrt{1+w}+\sqrt{Y+w})^{D/2}\sqrt{Y+w}}}{
(1-\sqrt{1+w})\left(\frac{1+\sqrt{1+w}}{2}\right)^{\delta_{D,0}+\delta_{D,6}}}\Bigg)^n.
\end{align*}
After that renormalisation procedure the limit $\Lambda^2\to\infty$ is
safe in all correlation function and any dimension $D\in \{2,4,6\}$.

\subsection{Dyson-Schwinger equation for $B>1$}

An $(N_1+...+N_B)$-point function of genus $g$ can be obtained
  from the $(1+1+...+1)$-point function of genus $g$ with $B$
boundary components through the explicit formula
\cite[Prop.~4.1]{Grosse:2016qmk}
\begin{align}
  G_g&(X_1^1,...,X_{N_1}^1|...|X_1^B,...,X_{N_B}^B)
  \label{recursiveGBg}
  \\
  &=\lambda^{N_1+...+N_B-B}\sum_{k_1=1}^{N_1}...\sum_{k_B=1}^{N_B}
  G_g(X_{k_1}^1|...|X_{k_B}^{N_B})
  \prod_{\beta=1}^B\prod_{\substack{l_\beta=1\\
      l_\beta\neq k_\beta}}^{N_\beta}
  \frac{4}{X^\beta_{k_\beta}-X^\beta_{l_\beta}}.\nonumber
\end{align}
For $g=0$ one has to write $\frac{W_0(X_{k_1}^1)}{2\lambda}$ instead of
$G_0(X_{k_1}^1)$. Furthermore, a $(1+1+\dots+1)$-point function
$G_{g}(X_1|X\triangleleft_J)$ with $B>1$ boundary components and genus
$g$ fulfils the linear integral
equation \cite[eq.\ (4.5)]{Grosse:2016qmk}
\begin{align}\nonumber\label{gBpunkt}
 0&=\big(W_0(X_1)+\lambda\nu\big) G_{g}(X_1|X\triangleleft_J)
+\frac{1}{2}\int_{1}^{\Lambda^2}dY\varrho(Y)\frac{G_g(X_1|X\triangleleft_{J})-
  G_g(Y|X\triangleleft_{J})}{X_1-Y}
\nonumber
\\
&
+\lambda \hspace*{-0.7cm}
\sum_{\substack{h+h'=g \\ I\uplus I'= J\\ (h,I),(h',I) \neq (0,\emptyset)}}
\hspace*{-0.7cm}
G_h(X_1|X\triangleleft_{I})G_{h'} (X_1|X\triangleleft_{I'})
+\lambda\sum_{\beta\in J}G_g(X_1,X_\beta,X_\beta|X
\triangleleft_{J\backslash\{\beta\}})
\nonumber
\\
&
+  \lambda G_{g-1}(X_1|X_1|X\triangleleft_J),
\end{align}
where $J=\{2,3,..,B\}$. Here and throughout the paper
  (for $z$ instead of $X$) we abbreviate
$G_g(X_0|X\triangleleft_I):=G_g(X_0|X_{i_1}|X_{i_2}|...|X_{i_p})$
and 
$G_g(X\triangleleft_I):=G_g(X_{i_1}|X_{i_2}|...|X_{i_p})$
if $I=\{i_1,...,i_p\}$.
We let $G_g(X_0|X\triangleleft_\emptyset)=G_g(X_0)$.
In the sum, $I\uplus I'=J$ means summation over
all possibly empty subsets $I\subset J$, with $I':=J\setminus I$.
The difference to the planar sector $(g=0)$ is the last term indexed 
$g-1$ which only contributes if $g\geq 1$. Furthermore, the entire sector 
of genus $h<g$ contributes to the genus-$g$ sector.

The equations \eqref{gBpunkt} for $g=0$ have been solved in
\cite{Grosse:2016pob}: 
\begin{align}
 G_0(X|Y)=&\frac{4\lambda^2}{\sqrt{X+c}\sqrt{Y+c}
(\sqrt{X+c}+\sqrt{Y+c})^2},
\label{G1B-thm}
\\
G_0(X^1|\dots|X^B)&=
\frac{d^{B-3}}{dt^{B-3}}
\Big(\frac{(-2\lambda)^{3B-4}}{(R(t))^{B-2} \sqrt{X^1{+}c{-}2t}^3
\cdots \sqrt{X^B{+}c {-}2t}^3}\Bigg)\Bigg|_{t=0}
\nonumber
\end{align}
for $B\geq 3$, where
\[
R(t) 
:= \lim_{\Lambda^2\to \infty}\Big(\frac{1}{\sqrt{Z}}-
\int_1^{\Lambda^2} \frac{dT
  \varrho(T)}{\sqrt{T{+}c}} \frac{1}{(\sqrt{T{+}c}+\sqrt{T{+}c{-}2t})
\sqrt{T{+}c{-}2t}}\Big).
\]
Note that multiple $t$-derivatives of $R(t)$ at $t=0$ produce 
renormalised moments of the measure \eqref{measure0}:
\begin{align}
\varrho_l := \lim_{\Lambda^2\to \infty}
\Big( \frac{\delta_{l,0}} {\sqrt{Z}}-
\frac{1}{2}
\int_1^{\Lambda^2} \frac{dT\,  \varrho(T)}{(\sqrt{T{+}c})^{3+2l}}\Big).
\label{moments}
\end{align}
In fact the proof of \eqref{G1B-thm} consists in a resummation of an ansatz 
which involves Bell polynomials (see 
Definition~\ref{def:Bell}) in the $\{\varrho_l\}$. 

The next goal is to find solutions for \eqref{gpunkt} 
and \eqref{gBpunkt} at any genus by employing techniques 
of complex analysis. The moments \eqref{moments} will be of 
paramount importance for that. We will find that all solutions 
are universal in terms of $\{\varrho_l\}$. The concrete model characterised 
by the sequence $E_{\narrowunderline{n}}$, coupling constant $\lambda$ 
and the dimension $D$ only affects the values of $\{\varrho_l\}$ 
via the measure \eqref{measure0} and the $D$-dependent solution $c$ of 
(\ref{implicit}).

\section{Solution of the non-planar sector}

\subsection{Change of variables}
\label{sec:changeofvariables}

As already mentioned, the equation for $W_0$ and its solution
holomorphically extend to (certain parts of) the complex plane. The
corresponding techniques have been brought to perfection by Eynard. We
draw a lot of inspiration from the exposition given in \cite{Eynard:2016yaa}. 
Starting point is another change of variables:
\begin{align}
    & z:=\sqrt{X+c},\qquad \mathcal{W}_0(z(X)):=W_0(X)\qquad \text{and}
    \label{changeofvariables}
    \\
  &\mathcal{G}_g(z^1_1(X^1_1),...,z_{N_1}^1(X^1_{N_1})|...|
  z_1^B(X_1^B),...,z_{N_B}^B(X^B_{N_B}))
  \nonumber
  \\
  &\hspace*{5cm}:=G_g(X^1_1,...,X^1_{N_1}|...|X^B_1,...,X^B_{N_B})
  \nonumber
\end{align}
for $2g+B>1$.  In the beginning, $z$ is defined to be positive;
nevertheless all correlation functions have an analytic
continuation. We define them by the complexification of the equations
\eqref{gpunkt} and \eqref{gBpunkt}, where we assume that the complex
variables fulfil the equations if they lie on the interval
$[\sqrt{1+c},\sqrt{1+\Lambda^2}]$. By recursion hypothesis each
correlation function is analytic for non-vanishing imaginary part of
the complex variables $z_i$, possibly with the exception of diagonals
$z_i=\pm z_j$.

We rephrase some of the earlier results in this setup. 
The solutions \eqref{gpunkt}, \eqref{G1B-thm} and the formula 
for the ($1+1+1$)-point function 
given in \cite{Grosse:2016pob} are easily translated into 
\begin{align}
\mathcal{W}_0(z)&=\frac{z}{\sqrt{Z}}-\lambda\nu
+\frac{1}{2}\int_{\sqrt{1+c}}^{\sqrt{\Lambda^2+c}} dy 
\frac{\tilde{\varrho}(y)}{(z+y)y},\qquad
\tilde{\varrho}(y):=2y\varrho(\sqrt{y^2-c}),
\nonumber
\\
\G_0(z_1|z_2)&= \frac{4\lambda^2}{z_1z_2(z_1+z_2)^2}, \qquad\qquad
\G_0(z_1|z_2|z_3)=-\frac{32 \lambda^5}{\rho_0z_1^3z_2^3z_3^3}.
\label{G0z}
\end{align}
Note that $\tilde{\varrho}(y)$ has support in 
$[\sqrt{1+c},\sqrt{\Lambda^2+c}]\subset \mathbb{R}_+$ because of $c>-1$
\cite{Grosse:2016qmk}. Furthermore, $\mathcal{W}_0(z)$ extends to a
sectionally holomorphic function with 
branch cut along $[-\sqrt{1+\Lambda^2}, 
-\sqrt{1+c}]$, the $(1+1)$-point function of genus zero is holomorphic
outside $z_i=0$ and the diagonals $z_1=-z_2$, whereas  
the $(1+1+1)$-point function (and all higher-$B$ functions) at genus 
0 are meromorphic with only pole at $z_i=0$. 

\begin{definition}\label{defint}
Let $\hat{K}_z$ be the integral operator of the 
linear integral equation,
\begin{align*}
\hat{K}_z f(z):=\mathcal{W}_0(z) f(z)
+\lambda\nu f(z)
+\frac{1}{2}\int_{\sqrt{1+c}}^{\sqrt{\Lambda^2+c}} dy\,
\tilde{\varrho}(y)\frac{f(z)-f(y)}{z^2-y^2},
\end{align*}
where $\mathcal{W}_0(z)$ is given by \eqref{G0z}.
\end{definition}
\noindent In this notation, \eqref{gpunkt} takes the form 
\begin{align}\label{1punktneuV}
0&=\hat{K}_z \G_g(z) 
+\lambda  \sum_{h=1}^{g-1}\G_h(z) 
\G_{g-h}(z) +\lambda \G_{g-1}(z|z).
\end{align}
We will heavily rely on:
\begin{lemma}\label{lemmaopK}
The operator $\hat{K}_z$ defined in Definition~\ref{defint} satisfies
\begin{align*}
\hat{K}_z\Big(\frac{1}{z}\Big) &=\frac{1}{\sqrt{Z}} , & 
\hat{K}_z\Big(\frac{1}{z^{3+2n}}\Big) &= 
\sum_{k=0}^{n}\frac{\varrho_k}{z^{2n+2-2k}}.
\end{align*}
\end{lemma}
\begin{proof}
This is a reformulation of \cite[Lemma 5.5]{Grosse:2016pob}.
\end{proof}

The first step beyond \cite{Grosse:2016pob} is to determine the 
1-point function at genus 1:
\begin{proposition}\label{proposition1}
 The solution of \eqref{1punktneuV} for $g=1$ is 
\begin{align*}
\G_1(z)=\frac{\lambda^3\varrho_1}{\varrho_0^2z^3} 
-\frac{\lambda^3}{\varrho_0z^5},
\end{align*}
where the $\varrho_l$ are given in \eqref{moments}.
\end{proposition}
\begin{proof}
From \eqref{G0z} we have
$\lambda \G_0(z|z)=\frac{\lambda^3}{z^4}$. Lemma~\ref{lemmaopK}
suggests the ansatz 
$\G_1(z)=\frac{\beta}{z^3}+\frac{\gamma}{z^5}$ with 
$\hat{K}_z \G_1(z)=\frac{\beta\varrho_0 }{z^2}
+\frac{\gamma \varrho_0}{z^4}+\frac{\gamma \varrho_1}{z^2}$. 
Comparison of coefficients yields the assertion.
 \end{proof}

\noindent
From \eqref{recursiveGBg} we get the $N$-point function of genus 1 which
in complex variables reads
\begin{align*}
\G_1(z_1,...,z_N)=\sum_{k=1}^N \G_1(z_k)
\prod_{l=1,l\neq k}^N\frac{4\lambda}{z_k^2-z_l^2}.
\end{align*}

Next we express equation \eqref{gBpunkt} in the new variables. To find
more convenient results we use \eqref{recursiveGBg} to write with
$J=\{2,..,B\}$
\begin{align}
&\G_g(z_1,z_\beta,z_\beta|z
\triangleleft_{J\backslash\{\beta\}})=\lim_{z_\alpha\to z_\beta}\G_g(z_1,z_\alpha,z_\beta|z
\triangleleft_{J\backslash\{\beta\}})
\nonumber
\\
=&16\lambda^2\left[\frac{\G_g(z_1|z
	\triangleleft_{J\backslash\{\beta\}})}{(z_1^2-z_\beta^2)^2}-\lim_{z_\alpha\to z_\beta} \frac{
	\frac{\G_g(z_\alpha|z
		\triangleleft_{J\backslash\{\beta\}})}{z_1^2-z_\alpha^2}-\frac{\G_g(z_\beta|z
		\triangleleft_{J\backslash\{\beta\}})}{z_1^2-z_\beta^2}}{(z_\alpha+z_\beta)
	(z_\alpha-z_\beta)}\right]
\nonumber
\\
=&16\lambda^2\frac{\partial}{2z_\beta \partial z_\beta}\frac{\G_g(z_1|z
	\triangleleft_{J\backslash\{\beta\}})-\G_g(z_\beta|z
	\triangleleft_{J\backslash\{\beta\}})}{z_1^2-z_\beta^2}
\label{G1bb}
\end{align}
and $\G_0(z_1,z_2,z_2)=8\lambda\frac{\partial}{2z_2\partial z_2
}\frac{\mathcal{W}_0(z_1)-\mathcal{W}_0(z_2)}{z_1^2-z_2^2}$.

Inserting \eqref{G1bb} into \eqref{gBpunkt} gives with 
Definition~\ref{defint} the following formula for 
$\G_g(z_1|z\triangleleft_J)$, for
$2g+|J| \geq 2$:
\begin{align}\nonumber
 0&=\hat{K}_{z_1} \G_g(z_1|z\triangleleft_J)
 + \lambda\G_{g-1}(z_1|z_1|z\triangleleft_{J})
 +\lambda \hspace*{-8mm}
 \sum_{\substack{h+h'=g \\ I\uplus I'= J\\
     (h,I),(h',I')\neq (0,\emptyset)}}\hspace*{-5mm}
 \G_h(z_1|z\triangleleft_{I})\G_{g-h}
 (z_1|z\triangleleft_{J\backslash I})
\\
\label{B>2}
&
+(2\lambda)^3 \sum_{\beta\in J}
 \frac{\partial}{z_\beta \partial z_\beta }\frac{\G_g(z_1|z
 	\triangleleft_{J\backslash\{\beta\}})-\G_g(z_\beta|z
 	\triangleleft_{J\backslash\{\beta\}})}{z_1^2-z_\beta^2}.
\end{align}
Note that this covers also (\ref{1punktneuV}) as $J=\emptyset$.

\subsection{Boundary creation operator}\label{sec:loop}

We are going to construct an operator which plays the r\^ole of the  
formal $T_{\narrowunderline{n}}:=
\frac{1}{E_{\narrowunderline{n}}}
\frac{\partial}{\partial E_{\narrowunderline{n}}}$ applied 
to the logarithm of the partition function $\mathcal{Z}[0]$ 
given in \eqref{Zustandssumme}. In dimension $D=0$ where 
$Z-1=\kappa=\nu=\zeta=0$ and $\mu_{bare}=\mu$, 
$\lambda_{bare}=\lambda$ we formally have 
\begin{align}
&T_{\narrowunderline{n}} \log \Big(\int d\Phi \;
e^{-\mathrm{tr}(E\Phi^2+\frac{\lambda}{3}\Phi^3)}\Big)
\nonumber
\\
&=-\frac{1}{\mathcal{Z}[0]}
\int d\Phi \;
\sum_{\narrowunderline{m}} 
\frac{\Phi_{\narrowunderline{m}\narrowunderline{n}}
\Phi_{\narrowunderline{n}\narrowunderline{m}}}{E_{\narrowunderline{n}}}
e^{-\mathrm{tr}(E\Phi^2+\frac{\lambda}{3}\Phi^3)}
\nonumber
\\
&=\frac{1}{\lambda\mathcal{Z}[0]}
\int d\Phi \;
\frac{1}{E_{\narrowunderline{n}}}
\Big(\frac{\partial}{\partial 
\Phi_{\narrowunderline{n}\narrowunderline{n}}}
+ 2 E_{\narrowunderline{n}}
\Phi_{\narrowunderline{n}\narrowunderline{n}}
\Big)
e^{-\mathrm{tr}(E\Phi^2+\frac{\lambda}{3}\Phi^3)}
\nonumber
\\
&=\frac{2}{\lambda\mathcal{Z}[0]}
\int d\Phi \;
\Phi_{\narrowunderline{n}\narrowunderline{n}}
e^{-\mathrm{tr}(E\Phi^2+\frac{\lambda}{3}\Phi^3)}
=\frac{2}{\lambda} G_{|\narrowunderline{n}|}.
\end{align}
By repeated application of $T_{\narrowunderline{n}_i}$ we formally
produce an $(1+\dots+1)$-point function. Of course, these operations 
are not legitimate:  in dimensions $D\in \{2,4,6\}$ we have to include for
renormalisation the $\Phi$-linear terms in \eqref{action3}, and the partition
function has no chance to exist for real $\lambda$. 

Nevertheless, we are able to show that 
$T_{\narrowunderline{n}_i}$ admits a rigorous replacement 
which we call the \emph{boundary creation operator}. It will be our
main device:
\begin{definition}\label{DefOp}
For $J=\{1,\dots,p\}$ let $|J|:=p$ and $z_J:=(z_1,\dots,z_{p})$. 
Then
\begin{align}
\hat{\mathrm{A}}^{\dag g}_{z_J,z}
:=
\sum_{l= 0}^{3g-3+|J|} \Big(-\frac{(3+2l) \varrho_{l+1}}{
\varrho_0 z^3}+\frac{3+2l}{z^{5+2l}}\Big)
\frac{\partial}{\partial \varrho_l}
+ \sum_{i\in J}\frac{1}{\varrho_0 z^3 z_i}
\frac{\partial}{\partial z_i}.
	\end{align}
\end{definition}
\noindent Note that the last variable $z$ in $\hat{\mathrm{A}}^{\dag
  g}_{z_J,z}$ plays a very different r\^ole than the $z_J$!

\begin{lemma}
  The differential operators $\hat{\mathrm{A}}^{\dag g}_{z_J,z}$
  commute,
\begin{align*}
\hat{\mathrm{A}}^{\dag g}_{z_J,z_p,z_q}
\hat{\mathrm{A}}^{\dag g}_{z_J,z_p}
=\hat{\mathrm{A}}^{\dag g}_{z_J,z_q,z_p}
\hat{\mathrm{A}}^{\dag g}_{z_J,z_q}.
\end{align*}
\end{lemma}
\begin{proof} Being a derivative, it is enough to verify 
$\hat{\mathrm{A}}^{\dag g}_{z_J,z_p,z_q}
\hat{\mathrm{A}}^{\dag g}_{z_J,z_p}(\varrho_k)
=\hat{\mathrm{A}}^{\dag g}_{z_J,z_q,z_p}
\hat{\mathrm{A}}^{\dag g}_{z_J,z_q}(\varrho_k)$ for any $k$ and 
$\hat{\mathrm{A}}^{\dag g}_{z_J,z_p,z_q}
\hat{\mathrm{A}}^{\dag g}_{z_J,z_p}(z_i)
=\hat{\mathrm{A}}^{\dag g}_{z_J,z_q,z_p}
\hat{\mathrm{A}}^{\dag g}_{z_J,z_q}(z_i)$ for any $i\in J$. 
This is guaranteed by 
\begin{align*}
\hat{\mathrm{A}}^{\dag g}_{z_J,z_p,z_q}
\hat{\mathrm{A}}^{\dag g}_{z_J,z_p}(\varrho_k)
&= 
\frac{(3+2k)(5+2k) \varrho_{k+2}}{
\varrho_0^2 z_q^3z_p^3}
-\frac{(3+2k)(5+2k)}{
\varrho_0 z_q^{7+2k} z_p^3}
- \frac{3(3+2k)\varrho_{k+1}\varrho_1}{\varrho_0^3 z_q^3z_p^3}
\\
&
+ \frac{3(3+2k)\varrho_{k+1}}{\varrho_0^2 z_q^5z_p^3}
+ \frac{3(3+2k) \varrho_{k+1}}{\varrho_0^2 z_q^3z_p^5}
-\frac{(3+2k)(5+2k)}{\varrho_0 z_q^3z_p^{7+2k}},
\\
\hat{\mathrm{A}}^{\dag g}_{z_J,z_p,z_q}
\hat{\mathrm{A}}^{\dag g}_{z_J,z_p}(z_i)
&= \frac{3\varrho_1}{\varrho_0^3 z_q^3 z_p^3 z_i}
-\frac{3}{\varrho_0^2 z_q^5 z_p^3 z_i}
-\frac{3}{\varrho_0^2 z_q^3z_p^5 z_i}
- \frac{1}{\varrho_0^2 z_q^3 z_p^3 z_i^3}.
\end{align*}
\end{proof}
\noindent
This shows that boundary components labelled by $z_i$ behave 
like bosonic particles at position $z_i$. The creation operator 
$(2\lambda)^3 \hat{\mathrm{A}}^{\dag g}_{z_J,z}$ adds to a $|J|$-particle state 
another particle at position $z$. The $|J|$-particle state is 
precisely given by $\G_g(z\triangleleft_J)$: 
\begin{theorem}\label{finaltheorem}
Assume that 
$\G_g(z)$ is, for $g\geq 1$, an odd function of $z\neq 0$ and a
rational function of 
$\varrho_0,\dots,\varrho_{3g-2}$ (true for $g=1$).
Then the  $(1+1+...+1)$-point function of genus $g\geq 1$ and $B$ 
boundary components of the renormalised $\Phi^3_D$-matricial QFT 
model in dimension $D\in\{0,2,4,6\}$ has the solution
\begin{align}
\G_g(z_1|...|z_B)=(2\lambda)^{3B-3}\hat{\mathrm{A}}^{\dag g}_{z_1,\dots,z_B}
\big(\hat{\mathrm{A}}^{\dag g}_{z_1,\dots,z_{B-1}}
\big(\cdots \hat{\mathrm{A}}^{\dag g}_{z_1,z_2}
\G_g(z_1)...\big)\big),\qquad
z_i\neq 0,
\label{eqfinalthm}
\end{align}
where $\G_g(z_1)$ is the 1-point function of genus $g\geq 1$ and 
the boundary creation operator $\hat{\mathrm{A}}^{\dag g}_{z_J}$ is defined in Definition \ref{DefOp}.
For $g=0$ the boundary creation operators act on the $(1+1)$-point function 
\begin{align*}
\G_0(z_1|...|z_B)=(2\lambda)^{3B-6}  \hat{\mathrm{A}}^{\dag 0}_{z_1,\dots,z_B}
\big( \hat{\mathrm{A}}^{\dag 0}_{z_1,\dots,z_{B-1}}
\big(\cdots  \hat{\mathrm{A}}^{\dag 0}_{z_1,z_2,z_3} 
\G_0(z_1|z_2)...\big)\big).
\end{align*} 
\end{theorem}
\begin{proof}
We rely on several Lemmas proved in Appendix \ref{appendix1}.
Regarding (\ref{eqfinalthm}) as a \emph{definition}, we prove in
Lemma \ref{lemma9} an equivalent formula for the 
linear integral equation \eqref{B>2}. This expression is satisfied because 
Lemma \ref{lemma6} and Lemma \ref{lemma8} add up to $0$. 
Consequently, the family of functions (\ref{eqfinalthm}) 
satisfies (\ref{B>2}). This solution is unique because 
of uniqueness of the perturbative expansion.
\end{proof}

\begin{remark}\label{rem:gsteps}
  The assumption that $\mathcal{G}_g(z)$ is an odd function of $z$ and
  rational in $\varrho_0,...,\varrho_{3g-2}$ will inductively follow
  from Theorem~\ref{thm:G-residue} and Proposition~\ref{structure-Gz}
  for genus $h\leq g$.  On the other hand, to prove
  Theorem~\ref{thm:G-residue} and Proposition~\ref{structure-Gz} for
  genus $h$, we need to apply Theorem~\ref{finaltheorem} for genus
  $h'<h$. In this way, Theorem~\ref{finaltheorem},
  Theorem~\ref{thm:G-residue} and Proposition~\ref{structure-Gz} are
  all inductively proved step by step when increasing the genus
  in each Theorem/Proposition.
\end{remark}

\begin{corollary}\label{coro2}
Let $J=\{2,...,B\}$. Assume that $z\mapsto \G_g(z)$ is holomorphic 
in $\mathbb{C}\setminus \{0\}$ with 
$\G_g(z)=-\G_g(-z)$ for all 
$z\in \mathbb{C}\setminus \{0\}$ and $g\geq 1$. Then
all $\G_g(z_1|z\triangleleft_J)$ with $2-2g-B<0$ 
\begin{enumerate}\itemsep 0pt
\item are holomorphic in every $z_i \in \mathbb{C}\setminus \{0\}$

\item are odd functions in every $z_i$, i.e.\
$ \G_g(-z_1|z\triangleleft_J)=-\G_g(z_1|z\triangleleft_J)$ for all
$z_1,z_i\in\mathbb{C}\setminus\{0\}$.
\end{enumerate}
\end{corollary}
\begin{proof}
The boundary creation operator $\hat{\mathrm{A}}^{\dag g}_{z_J,z}$ of 
Definition~\ref{DefOp} preserves holomorphicity in 
$\mathbb{C}\setminus\{0\}$ and 
maps odd functions into odd functions. Thus only the initial conditions 
need to be checked. They are fulfilled
for $\G_0(z_1|z_2|z_3)$ and $\G_1(z_1)$ according to (\ref{G0z}); for 
$g\geq 2$ by assumption.
\end{proof}
\noindent
The assumption will be verified later in Proposition~\ref{structure-Gz}.

\begin{corollary}\label{coro1}
The boundary creation operator $\hat{\mathrm{A}}^{\dag g}_{z_J,z_1}$ acting on an 
$(N_1+...+N_B)$-point function of genus $g$ gives the 
following $(1+N_1+...+N_B)$-point function of genus $g$
\begin{align*}
  &\G_g(z_1|z^1_1,..,z^1_{N_1}|..|z^B_1,..,z^B_{N_B})
  \\
  &=(2\lambda)^3 
\hat{\mathrm{A}}^{\dag g}_{z^1_1,..,z^1_{N_1},\dots,z^B_1,..,z^B_{N_B},z_1}
(\G_g(z^1_1,..,z^1_{N_1}|..|z^B_1,..,z^B_{N_B})).
\end{align*}
\end{corollary}
\begin{proof}
This follows from the change to complex variables in
equation \eqref{recursiveGBg} and $\hat{\mathrm{A}}^{\dag g}_{z_J,z_1}
(\frac{1}{z_i^2-z_j^2})=0$ for $1\neq i\neq j\neq 1$.
\end{proof}

\subsection{Solution of the $1$-point function for $g\geq 1$}
\label{sec:non-planar-1pt}

It remains to check that the 1-point function $\G_g(z)$ at genus 
$g\geq 1$ satisfies the assumptions of Theorem~\ref{finaltheorem} and 
Corollary~\ref{coro2}, namely: 
\begin{enumerate}\itemsep 0pt
\item  $\G_g(z)$ depends only on
the moments $\varrho_0,\dots,\varrho_{3g-2}$ of the measure, 

\item $z\mapsto \G_g(z)$ is holomorphic on $\mathbb{C}\setminus \{0\}$ and 
an odd function of $z$.
\end{enumerate}
We establish these properties by solving \eqref{1punktneuV}
via a formula for the inverse of $\hat{K}_z$. This formula is inspired by 
topological recursion, see e.g.\ \cite{Eynard:2016yaa}. We give a 
few details in section~\ref{sec:TR}.
\begin{definition}\label{def:Bell}
The {\sffamily Bell polynomials} are defined by
\begin{align*}
B_{n,k}(x_1,...,x_{n-k+1})=\sum \frac{n!}{j_1!j_2!...j_{n-k+1}!}\left(\frac{x_1}{1!}\right)^{j_1}\left(\frac{x_2}{2!}\right)^{j_2}...\left(\frac{x_{n-k+1}}{(n-k+1)!}\right)^{j_{n-k+1}}
	\end{align*}
	for $n\geq 1$, where the sum is over non-negative integers $j_1,...,j_{n-k+1}$ with $j_1+...+j_{n-k+1}=k$ and $1j_1+...+(n-k+1)j_{n-k+1}=n$. 
Moreover, one defines $B_{0,0}=1$ and $B_{n,0}=B_{0,k}=0$ for $n,k>0$.
\end{definition}
\noindent
An important application is Fa\`a di 
Bruno's formula, the $n$-th order chain rule:
\begin{align}\label{Faadi}
\frac{d^n}{dx^n}f(g(x))=\sum_{k=1}^{n}f^{(k)}(g(x)) \,
B_{n,k}(g'(x),g''(x),...,g^{(n-k+1)}(x)).
\end{align}
\begin{proposition}\label{Prop5}
Let $f(z)=\sum_{k=0}^\infty \frac{a_{2k}}{z^{2k}}$ be an even Laurent 
series about $z=0$ bounded at $\infty$.
Then the inverse of the integral operator $\hat{K}_z$ of 
Definition \ref{defint} is given by the residue formula
\begin{align*}
\Big[z^2\hat{K}_z\frac{1}{z}\Big]^{-1}f(z) 
&=- \Res\displaylimits_{z'\to 0}\left[ K(z,z')\, f(z')dz'\right],
\\
K(z,z')&:=\frac{2}{(\mathcal{W}_0(z')-\mathcal{W}_0(-z'))(z'^2-z^2)}.
\end{align*}
\end{proposition}
\begin{proof}
The formulae (\ref{G0z}) give rise to the series expansion 
\begin{align}
\label{Gzminusz}
\frac{1}{2}(\mathcal{W}_0(z')-\mathcal{W}_0(-z'))= 
\sum_{l=0}^\infty \varrho_l \cdot (z')^{2l+1},
\end{align}
where the $\varrho_l$ are given in \eqref{moments}
(either take $\lim_{\Lambda\to \infty}\mathcal{W}_0(z)$ in 
(\ref{G0z}) or leave out the limit in \eqref{moments})
The series of its reciprocal is found using \eqref{Faadi}:
\begin{align}
&\frac{2}{(\mathcal{W}_0(z')-\mathcal{W}_0(-z'))} = \frac{1}{z' \varrho_0}
\sum_{m=0}^\infty \frac{(z')^{2m}}{m!} S_m,
\label{S-Bell}
\\
&S_m:= \frac{d^m}{d \tau^m}\Big|_{\tau=0} \Big(\sum_{l=0}^\infty 
\frac{\varrho_l}{\varrho_0} \tau^l\Big)^{-1} 
= \sum_{i=0}^m\frac{(-1)^i i!}{\varrho_0^{i}} 
B_{m,i}(1!\varrho_1,2! \varrho_2,...,(m{-}i{+}1)!\varrho_{m-i+1}).
\nonumber
\end{align}
Multiplication by the geometric series gives 
\begin{align}
K(z,z')=-\frac{1}{z^2 z' \varrho_0} 
\sum_{n,m=0}^\infty \frac{(z')^{2m+2n}}{m! z^{2n}} S_m.
\label{Kzz}
\end{align}
The residue of a monomial in 
$f(z')=\sum_{k=0}^\infty \frac{a_{2k}}{(z')^{2k}}$ is then
\begin{align}\label{Res}
\Res\displaylimits_{z'\to 0}\left[ K(z,z')\frac{dz'}{(z')^{2k}}\right]
=-\frac{1}{\varrho_0}\sum_{j=0}^{k} \frac{S_j}{j!z^{2k-2j+2}}.
\end{align}
In the next step we apply the operator $z^2\hat{K}\frac{1}{z}$ 
to \eqref{Res}, where Lemma~\ref{lemmaopK} is used:
\begin{align}
&z^2 \hat{K}_z\Big(\frac{1}{z}
\frac{(-1)}{\varrho_0} 
\sum_{j=0}^{k} \frac{S_j}{j!z^{2k-2j+2}}\Big)
=-\frac{z^2}{\varrho_0} 
\sum_{j=0}^{k}\sum_{i=0}^{k-j}\frac{S_j\varrho_i}{j! z^{2k-2j-2i+2}}
\nonumber
\\
&=-\sum_{j=0}^{k}\frac{S_{k-j}}{(k-j)! z^{2j}} 
-\frac{1}{\varrho_0} 
\sum_{i=0}^{k-1} \sum_{j=i+1}^{k}
\frac{S_{k-j}\varrho_{j-i}}{(k-j)! z^{2i}}.
\label{Bell2}
\end{align}
The last sum over $j$ is treated as follows, where the 
Bell polynomials are inserted for $S_m$:
\begin{align*}
&\sum_{j=i+1}^{k}\frac{S_{k-j} \varrho_{j-i}}{(k-j)!}
=\sum_{j=1}^{k-i}\frac{S_{k-j-i}\varrho_{j}}{(k-j-i)! }
\\
&=\sum_{j=1}^{k-i}\sum_{l=0}^{k-j-i}\frac{(-1)^{l}l!}{(k-j-i)!
\varrho_0^{l}}\,\varrho_jB_{k-j-i,l}(1!\varrho_1,...,(k-j-i-l+1)!
\varrho_{k-j-i-l+1})
\\
&=\sum_{l=0}^{k-i-1}\frac{(-1)^{l}l!}{\varrho_0^{l}(k-i)!}
\sum_{j=1}^{k-i-l}\binom{k{-}i}{j}j!\,\varrho_j
B_{k-j-i,l}(1!\varrho_1,...,(k{-}j{-}i{-}l{+}1)!\varrho_{k-j-i-l+1})
\\
&=\sum_{l=0}^{k-i}\frac{(-1)^{l}(l+1)!}{\varrho_0^{l}(k-i)!} 
B_{k-i,l+1}(1!\varrho_1,...,(k-i-l)!\varrho_{k-i-l})
\\
&=-\varrho_0\frac{S_{k-i}}{(k-i)!}.
\end{align*}
We have used $B_{n,0}=0$ and $B_{0,n}=0$ for $n>0$ to eliminate some 
terms, changed the order of sums, and used the following identity 
for the Bell polynomials \cite[Lemma 5.9]{Grosse:2016pob}
\begin{align}
\sum_{j=1}^{n-k}\binom{n}{j}x_jB_{n-j,k}(x_1,...,x_{n-j-k+1})
=(k+1)B_{n,k+1}(x_1,...,x_{n-k}).
\label{xBell}
\end{align}
Inserted back we find that \eqref{Bell2} reduces to the ($j=k$)-term
of the first sum in the last line of \eqref{Bell2}, i.e.\
\begin{align*}
&z^2 \hat{K}_z \Big(
\frac{1}{z} \Res\displaylimits_{z'\to 0}\Big[ 
K(z,z')\frac{dz'}{(z')^{2k}}\Big]\Big)= -\frac{1}{z^{2k}}.
\end{align*}
This finishes the proof.
\end{proof}
\begin{theorem} \label{thm:G-residue}
For any $g\geq 1$ and $z\in \mathbb{C}\setminus \{0\}$ 
one has 
\begin{align}
\G_g(z) = \frac {\lambda}{z} 
\Res\displaylimits_{z'\to 0} 
\left[ K(z,z') \Big\{
\sum_{h=1}^{g-1}\G_h(z') 
\G_{g-h}(z') + \G_{g-1}(z'|z')\Big\} 
(z')^2 dz'\right].
\label{eq-G1-final}
\end{align}
\end{theorem}
\begin{proof} 
The formula arises when applying Proposition~\ref{Prop5} to
\eqref{1punktneuV} and holds if 
the function in $\{~\}$ is an even Laurent polynomial in $z'$ 
bounded in $\infty$.
This is the case for $g=1$ where only $\G_0(z'|z')
=\frac{\lambda^2}{(z')^4}$ contributes. Evaluation of the residue 
reconfirms Proposition~\ref{proposition1}. We proceed by induction in 
$g\geq 2$, assuming that all 
$\G_h(z')$ with $1\leq h<g$ on the rhs of \eqref{eq-G1-final} 
are odd Laurent polynomials bounded in $\infty$; 
their product is even. The induction hypothesis 
also verifies the assumption of 
Theorem~\ref{finaltheorem} so that 
$\G_{g-1}({-}z'|{-}z')=-\G_{g-1}(z'|{-}z')=\G_{g-1}(z'|z')$ is even and, because of 
$\G_{g-1}(z'|z'')=(2\lambda)^3 \hat{\mathrm{A}}^{\dag g-1}_{z'',z'} 
\G_{g-1}(z'')$, inductively a 
Laurent polynomial bounded in $\infty$. 
Thus, equation \eqref{eq-G1-final} holds for genus $g\geq 2$ and, 
as consequence of \eqref{Res}, $\G_g(z)$ is 
again an odd Laurent polynomial bounded in $\infty$. 
Equation~\eqref{eq-G1-final}  is thus proved 
for all $g\geq 1$, and the assumption of 
Theorem~\ref{finaltheorem} is verified.
\end{proof}
\noindent
A more precise characterisation can be given. It relies on
\begin{definition}
A polynomial $P(x_1,x_2,\dots)$ is called {\sffamily $n$-weighted} if 
$\sum_{k=1}^\infty k x_k \frac{\partial}{\partial x_k} P(x_1,x_2,\dots)
=nP(x_1,x_2,\dots)$.
\end{definition}
\noindent 
The Bell polynomials $B_{n,k}(x_1,\dots,x_{n-k+1})$ are $n$-weighted. 
The number of monomials in an $n$-weighted polynomial is $p(n)$, 
the  number of partitions of $n$. 
The product of an $n$-weighted 
by an $m$-weighted polynomial is $(m+n)$-weighted.

In the following we let
$P^{decoration}_j(\varrho)$ be some
$j$-weighted polynomial in 
$\{\frac{\varrho_1}{\varrho_0},\dots,\frac{\varrho_j}{\varrho_0}\}$
with rational coefficients. A \emph{decoration} (empty or primes)
distinguishes several such polynomials, but is of no relevance.
\begin{proposition}
\label{structure-Gz}
 For $g\geq 1$ one has
\[
\G_g(z)= (2 \lambda)^{4g-1} 
\sum_{k=0}^{3g-2} \frac{P_{3g-2-k}(\varrho)}{\varrho_0^{2g-1} z^{2k+3}},
\]
where $P_0\in \mathbb{Q}$ and the $P_j(\varrho)$ for $j\geq 1$ are 
some $j$-weighted polynomials in 
$\{\frac{\varrho_1}{\varrho_0},\dots,\frac{\varrho_j}{\varrho_0}\}$
with rational coefficients. 
\end{proposition}
\begin{proof}
  The case $g=1$ is directly checked.  We proceed by induction in $g$
  for both terms in $\{~\}$ in \eqref{eq-G1-final}.  The hypothesis
  gives
  $\G_h(z')\G_{g-h}(z')= (2\lambda)^{4g-2} \sum_{k=0}^{3g-4}
  \frac{P'_{3g-4-k}(\varrho)}{\varrho_0^{2g-2} (z')^{2k+6}}$. In the
  second term in $\{~\}$,
  $\G_{g-1}(z'|z')=(2\lambda)^3
  \hat{\mathrm{A}}^{\dag\,g{-}1}_{z',z'} \G_{g-1}(z')$, the three
  types of contributions in the boundary creation operator act as
  follows:
  \begin{align*} (2\lambda)^3&
      \sum_{k=0}^{3g-5}(2\lambda)^{4g-5}
      \frac{P_{3g-5-k}(\varrho)}{(z')^{2k+3}}
      \Big(-\frac{3\varrho_1}{\varrho_0 (z')^3}+\frac{3}{(z')^5}\Big)
      \frac{\partial}{\partial \varrho_0} \frac{1}{\varrho_0^{2g-3}}
      \\
      &=\frac{(2\lambda)^{4g-2}3(3-2g)}{\varrho_0^{2g-2}}
      \Big(\sum_{k=0}^{3g-5}\frac{\frac{\varrho_1}{\varrho_0}
        P_{3g-5-k}(\varrho)}{(z')^{2k+6}} +\sum_{k=0}^{3g-5}\frac{
        P_{3g-4-(k+1)}(\varrho)}{(z')^{2(k+1)+6}} \Big),
      \\
      (2\lambda)^3& \sum_{k=0}^{3g-5}(2\lambda)^{4g-5}
      \frac{P_{3g-5-k}(\varrho)}{\varrho^{2g-3}}
      \frac{1}{\varrho_0(z')^4} \frac{\partial}{\partial z'}
      \frac{1}{(z')^{2k+3}}
      \\
      &=-\frac{(2\lambda)^{4g-2}}{\varrho_0^{2g-2}}
      \sum_{k=0}^{3g-5}\frac{(2k+3)P_{3g-4-(k+1)}(\varrho)}{(z')^{2(k+1)+6}}
      \\
      (2\lambda)^3& \sum_{k=0}^{3g-5}\frac{(2\lambda)^{4g-5}
      }{\varrho^{2g-3}(z')^{2k+3}} \hat{A}^{\dag
        {g-1}}_{z',z'}P_{3g-5-k}(\varrho)
      \\
      &=
      \sum_{k=0}^{3g-5}\frac{(2\lambda)^{4g-2}}{\varrho_0^{2g-3}(z')^{2k+3}}
      \Big\{ \Big(-\frac{3\varrho_1}{\varrho_0
        (z')^3}+\frac{3}{(z')^5}\Big) \sum_{l=1}^{3g-5-k}
      \Big({-}\frac{\varrho_l}{\varrho_0^2}\Big)
      \frac{\partial}{\partial \frac{\varrho_l}{\varrho_0}}
      P_{3g-5-k}(\varrho)
      \\
      & \qquad\qquad +
      \sum_{l=1}^{3g-5-k}\Big(-\frac{(3+2l)\frac{\varrho_{l+1}}{\varrho_0}}{(z')^3}
      +\frac{3+2l}{(z')^{5+2l}} \Big)\frac{1}{\varrho_0}
      \frac{\partial}{\partial \frac{\varrho_l}{\varrho_0}}
      P_{3g-5-k}(\varrho) \Big \}
\end{align*}
which are all of the same structure
$(2\lambda)^{4g-2} \sum_{k=0}^{3g-4} 
\frac{P''_{3g-4-k}(\varrho)}{\varrho_0^{2g-2} (z')^{2k+6}}$
\begin{align*}
&\frac{\lambda}{z} \Res\displaylimits_{z'\to 0} 
\Big[K(z,z') dz' (z')^2 (2\lambda)^{4g-2}
\sum_{k=0}^{3g-4} 
\frac{P_{3g-4-k}(\varrho)}{\varrho_0^{2g-2} (z')^{2k+6}}\Big]
\\
&=(2\lambda)^{4g-1}\sum_{k=0}^{3g-4} \sum_{j=0}^{k+2} 
\frac{P_{3g-4-k}(\varrho)S_j(\varrho)}{2j!\varrho_0^{2g-1} z^{2k+7-2j}}
=(2\lambda)^{4g-1}\sum_{k=0}^{3g-2} 
\frac{P'_{3g-2-k}(\varrho)}{\varrho_0^{2g-1} z^{2k+3}},
\end{align*}
because $S_j(\varrho)$ is also a $j$-weighted polynomial by 
\eqref{S-Bell}.
\end{proof}
\noindent
In particular, this proves the assumption of 
Theorem~\ref{finaltheorem}, namely that 
$\G_g(z)$ depends only on $\{\varrho_0,\dots, \varrho_{3g-2}\}$.
To be precise, we reciprocally increase the genus in 
Theorem~\ref{finaltheorem} and Proposition~\ref{structure-Gz}
as described in Remark~\ref{rem:gsteps}.

\subsection{Remarks on topological recursion}

\label{sec:TR}

We return to equation~\eqref{B>2} for $g>0$ or $|J|\geq 3$ in case of $g=0$.
In the first term,
$\G_g(z_1|z \triangleleft_J)$ is given by
Theorem \ref{finaltheorem} so that
$z_1\G_g(z_1|z \triangleleft_J)$ is by 
Proposition~\ref{structure-Gz} and Definition \ref{DefOp} of
$\hat{\mathrm{A}}^{\dag g}_{z_J,z_1}$ an even
Laurent polynomial about $z_1=0$ bounded at $\infty$. Therefore,
Proposition~\ref{Prop5} applies and gives the representation
\begin{align}
&z_1 \G_g(z_1|z\triangleleft_J)
\label{Gg-TR}
\\
&=\Res\displaylimits_{z'\to 0} K(z_1,z')(z')^2 dz'
\Big[
\lambda \hspace*{-5mm}
\sum_{\substack{h+h'=g\\I\uplus I'= J\\ (h,I),(h',I')\neq (0,\emptyset)}}
\hspace*{-5mm}\G_h(z'|z\triangleleft_I) \G_{h'}(z'|z\triangleleft_{I'})
\nonumber
\\
&\hspace*{1cm}
+\lambda \G_{g-1}(z'|z'|z\triangleleft_J)
+(2\lambda)^3 \sum_{\beta \in J}
\frac{\partial}{z_\beta \partial z_\beta}
\frac{\G_g(z'|z\triangleleft_{J\setminus \beta})
  -\G_g(z_\beta|z\triangleleft_{J\setminus \beta})}{(z')^2-z_2^2}\Big].
\nonumber
\end{align}
The last  term $\G_g(z_\beta|z\triangleleft_{J\setminus \beta})$
does not contribute to the residue. The other term
$\G_g(z'|z\triangleleft_{J\setminus \beta})$ also arises for
$(h,I)=(g,\{\beta\})$ and
$(h',I')=(0,\{\beta\})$ in the sum, where it has the prefactor 
$\G_0(z'|z_\beta)$ given in the second line of (\ref{G0z}).
The total contribution of this term inside $[~~]$ is
\begin{align*}
&\frac{\lambda \G_g(z'|z\triangleleft_{J\setminus \beta})}{z'z_\beta}
\Big(\frac{4\lambda^2}{(z'-z_\beta)^2}
+\frac{4\lambda^2}{(z'+z_\beta)^2}\Big)\;.
\end{align*}
This suggests to introduce
 \begin{align}
  \omega_{g,B}(z_1,...,z_B)&:=\Big(\prod_{i=1}^B z_i\Big) 
  \G_g(z_1|...|z_B) \quad \text{for }
2g+B>2, 
\label{def-omega-tr}
\\
\omega_{0,2}(z_1,z_2)&:=\frac{4\lambda^2}{(z_1-z_2)^2}.
\nonumber
\end{align}
By Theorem \ref{finaltheorem} and Proposition \ref{structure-Gz} together with
Definition \ref{DefOp},
each $\omega_{g,B}(z_1,...,z_B)$ with $2g+B>2$ is an
even Laurent-polynomial in every $z_i$ so that we can replace
$z'\mapsto -z'$ in one of the arguments. The structure then matches
the combination 
$\omega_{0,2}(z',z_\beta)+\omega_{0,2}(z',-z_\beta)$ needed as prefactor of
$\omega_{g,B-1}(z',z_2,\stackrel{\beta}{\check{\dots}},z_B)$.
Multiplying (\ref{Gg-TR}) by $\prod_{i\in J}z_i$ and taking the
previous discussion into account, we have proved (after a shift $B\mapsto B+1$):
\begin{theorem}[cf.\ {\cite[Thm.\ 6.4.4]{Eynard:2016yaa}}]
  \label{theorem1}
  The functions $\omega_{g,B+1}(z_0,...,z_B)$ given via \eqref{def-omega-tr}
in terms of the complexification \eqref{npunktcomplex} and
\eqref {changeofvariables}
of the cumulants
$\G_g(z_0|z\triangleleft_{\{1,...,B\}})$ in
\eqref{Entwicklungskoeffizienten}
satisfy for $2g+B\geq 2$ the recursive equation
 \begin{align}
   &\omega_{g,B+1}(z_0,...,z_B)
   \label{omega-TR}
   \\
  &=\lambda \Res\displaylimits_{z'\to 0}
  \bigg[K(z_0,z')\,dz'\Big(\omega_{g-1,B+2}(z',{-}z',
  z\triangleleft_{\{1,...,B\}})
\nonumber  \\
  &\hspace*{3cm}+
\hspace*{-8mm}
\sum_{\substack{h+h'=g\\ I\uplus I' =\{1,...,B\}
    \\ (h,I),(h',I')\neq (0,\emptyset)}}
\hspace*{-6mm}
\omega_{h,|I|+1}(z',z\triangleleft_I)
  \omega_{h',|I'|+1}({-}z',z\triangleleft_{I'})\Big)\bigg],
  \nonumber
\end{align}
where $\omega_{g,|I|+1}(z_0,z\triangleleft_I)=
\omega_{g,n+1}(z_0,z_{i_1},....,z_{i_n})$ if
$I=\{{i_1},....,{i_n}\}$ and the
the kernel $K$ is given in Proposition \ref{Prop5}.
\end{theorem}
\noindent
The case ($g=0$, $B=3$) was excluded in the above discussion,
but can be checked directly. 

After rescaling $\omega_{g,B}(z_1,..,z_B)=
(2\lambda)^{4g+3B-4}\omega^{TR}_{g,B}(z_1,..,z_B)$ we 
have reproved the known result
\cite[Thm.\ 6.4.4]{Eynard:2016yaa}
of topological recursion when taking a factor $2$ out of $K$ and writing
\begin{align*}
  \frac{1}{(\mathcal{W}_0(z')-\mathcal{W}_0(-z'))(z'^2-z^2)}
  &=\frac{\frac{1}{2}(\frac{1}{z-z'}-\frac{1}{z+z'})}{
    x'(z') (y(z')-y(-z'))}
  \;,\qquad
  \\
  \text{where } \quad x(z)&:=\frac{z^2}{2},~y(z):=-\mathcal{W}_0(z).
\end{align*}
In fact, \cite[Thm.\ 6.4.4]{Eynard:2016yaa} 
motivated our ansatz 
for an inverse of $\hat{K}_z$ as the residue involving $K(z,z')$.
Our proof of Theorem~\ref{theorem1} is of comparable difficulty and 
length as \cite[Thm.\ 6.4.4]{Eynard:2016yaa}.

For the $\omega_{g,|J|}(z\triangleleft_J)$ with $J\geq 2$ we have
the alternative representation of Proposition \ref{finaltheorem}.
There is, however, one property where 
(\ref{omega-TR}) is really needed: for the proof of the symmetry
$\G_g(z_1|z_2)=\G_g(z_2|z_1)$ for $g\geq 1$. The complete symmetry
of $\omega_{g,B}(z_1,...,z_B)$ or 
equivalently $\G_g(z_1|...|z_B)$ in all arguments then follows
from Proposition \ref{finaltheorem}. For $g=0$ the symmetry of the
starting point $\G_0(z_1|z_2|z_3)$ is manifest in (\ref{G0z}).
The proof of $\G_g(z_1|z_2)=\G_g(z_2|z_1)$ is the same as
\cite[Thm.\ 4.6]{Eynard:2007kz} (which uses a slightly different notation).
There is no need to repeat it in this paper.

\section{A Laplacian to compute intersection  
numbers}

\subsection{Free energy and boundary 
annihilation operator}

\label{sec:freeenergy}

\begin{definition}\label{DefA}
We introduce the operators
\begin{align}
\hat{\mathrm{A}}^{\dag}_z &:=
\sum_{l= 0}^{\infty} \Big(-\frac{(3+2l) \varrho_{l+1}}{
\varrho_0 z^3}+\frac{3+2l}{z^{5+2l}}\Big)
\frac{\partial}{\partial \varrho_l}, &
\hat{\mathrm{N}} &= - \sum_{l=0}^\infty
\varrho_l\frac{\partial}{\partial \varrho_l} ,
\nonumber
\\
\hat{\mathrm{A}}_{\check{z}}  f(\bullet) &:=-
\sum_{l=0}^\infty \Res\displaylimits_{z\to 0}
\Big[\frac{z^{4+2l} \varrho_l}{3+2l} f(z) dz\Big]
\end{align}
and the free energies
\begin{align*}
F_1=-\frac{1}{24} \log \varrho_0,\qquad 
F_g&= 
\frac{1}{(2g-2)(2\lambda)^3}
\hat{\mathrm{A}}_{\check{z}}  \G_g(\bullet) 
\text{ for } g> 1.
\end{align*}
We call $\hat{\mathrm{A}}_{\check{z}}$ a {\sffamily boundary
annihilation operator} acting on Laurent polynomials $f$. 
\end{definition}
\begin{proposition}\label{prop:F}
The $F_g$ have for $g>1$ a presentation as
\begin{equation}\label{Fg-homogeneous}
F_g= (2 \lambda)^{4g-4} \frac{P_{3g-3}(\varrho) }{\varrho_0^{2g-2}},
\end{equation}
where $P_{3g-3}(\varrho)$ is some $(3g-3)$-weighted polynomial in 
$\{\frac{\varrho_1}{\varrho_0},\dots,
\frac{\varrho_{3g-3}}{\varrho_0}\}$. They satisfy
\begin{align}
  \G_g(z) = (2\lambda)^3 \hat{\mathrm{A}}^{\dag}_z F_g
  \qquad \text{for } g\geq 1.
\label{eq:GAF}
\end{align}
\end{proposition}
\begin{proof} 
From Proposition~\ref{structure-Gz}
we conclude for $g\geq 2$
\begin{align*}
F_g&:=\frac{1}{(2\lambda)^3(2g-2)} \hat{\mathrm{A}}_{\check{z}} \G_g(\bullet)
\\
&=-
\frac{(2 \lambda)^{4g-4}}{2g-2} \Res\displaylimits_{z\to 0}\Big[
\sum_{l=0}^\infty \frac{\varrho_lz^{4+2l}}{(3+2l)}
\sum_{k=0}^{3g-2} \frac{P_{3g-2-k}(\varrho)}{\varrho_0^{2g-1} z^{2k+3}}
dz\Big]
\\
&= \frac{(2 \lambda)^{4g-4}}{2-2g}
\sum_{k=1}^{3g-2} \frac{
\frac{\varrho_{k-1}}{(2k+1)\varrho_0}
\cdot P_{3g-2-k}(\varrho)}{\varrho_0^{2g-2}}
= (2 \lambda)^{4g-4} \frac{P'_{3g-3}(\varrho) }{\varrho_0^{2g-2}},
\end{align*}
which confirms the structure \eqref{Fg-homogeneous}.

Equation (\ref{eq:GAF}) is straightforward to check for $g=1$.  
Equation (\ref{eq:GAF}) for $g\geq 2$ is equivalent to
$(2g-2)\G_g(z)  = A^\dag_z A_{\check{w}} \G_g(\bullet)$.
Inserting the definitions we have
\begin{align*}
\hat{\mathrm{A}}^{\dag}_z  \hat{\mathrm{A}}_{\check{w}} \G_g(\bullet)
&=-\hat{\mathrm{A}}^{\dag}_z 
\Res\displaylimits_{w\to 0}
  \Big[\sum_{l=0}^\infty\frac{\varrho_l w^{4+2l}dw}{3+2l}
  \G_g(w)\Big]
  \\
  &=\Res\displaylimits_{w\to 0}
  \Big[\sum_{l=0}^\infty  w^{4+2l}dw
\Big(-\frac{1}{z^{5+2l}}
+\frac{\varrho_{l+1}}{\varrho_0 z^3}
\Big)
  \G_g(w)\Big]
  \\
  &-\Res\displaylimits_{w\to 0}
  \Big[\sum_{l=0}^\infty\frac{\varrho_l w^{4+2l}dw}{3+2l}
  \hat{\mathrm{A}}^{\dag}_z  \G_g(w)\Big]
  \\
  &=\Res\displaylimits_{w\to 0}
  \Big[
  dw\Big(-\frac{w^{2}}{z(z^2-w^2)}
  +\sum_{l=0}^\infty  w^{2+2l}\frac{\varrho_{l}}{\varrho_0 z^3}
\Big)
  \G_g(w)\Big]
  \\
  &-\Res\displaylimits_{w\to 0}
  \Big[\sum_{l=0}^\infty\frac{\varrho_l w^{4+2l}dw}{3+2l}
  \Big( \hat{\mathrm{A}}^{\dag g}_{w,z} \G_g(w)
-\frac{1}{\varrho_0z^3 w} \frac{\partial \G_g(w)}{\partial w}
\Big)  \Big]
\\
&=-\Res\displaylimits_{w\to 0}
  \Big[
  \frac{w^{2}dw}{z(z^2-w^2)}  \G_g(w)\Big]
  +\hat{\mathrm{A}}_{\check{w}}\hat{\mathrm{A}}^{\dag g}_{\bullet,z}
  \G_g(\bullet),
\end{align*}
where we have completed $\hat{\mathrm{A}}^{\dag}_{z}$ to
$\hat{\mathrm{A}}^{\dag g}_{w,z}$ and used that the residue of a total
differential vanishes. With Proposition~\ref{structure-Gz} it is easy
to see that
$\Res\displaylimits_{w\to 0} \Big[ \frac{w^{2}dw}{z(z^2-w^2)}
\G_g(w)\Big]=\G_g(z)$.  Since $\G_g(z|w)=\G_g(w|z)$ is
symmetric for $g\geq 2$ according to the discussion at the
end of section~\ref{sec:TR}, we have
$(2\lambda)^{-3}\G_g(z|w) =\hat{\mathrm{A}}^{\dag g}_{w,z} \G_g(w)
=\hat{\mathrm{A}}^{\dag g}_{z,w} \G_g(z)$
by Theorem \ref{finaltheorem}. We
have $\hat{\mathrm{A}}_{\check{w}} \hat{\mathrm{A}}^{\dag g}_{z,\bullet}
=\hat{\mathrm{N}}$, and since
$\hat{\mathrm{N}} P_j(\varrho)\equiv 0$ for any $j$-weighted polynomial
$P_j(\varrho)$ in
$\{\frac{\varrho_1}{\varrho_0},...  \frac{\varrho_j}{\varrho_0}\}$, we
have $\hat{\mathrm{N}} \G_g(z)=(2g-1)\G_g(z)$ by
Proposition~\ref{structure-Gz}. Altogether we have proved
$\hat{\mathrm{A}}^{\dag}_z  \hat{\mathrm{A}}_{\check{w}} \G_g(\bullet)
 =-\G_g(z) + \hat{\mathrm{N}} \G_g(z)
=(2g-2)\G_g(z)$.
\end{proof}
\begin{remark}
  Proposition \ref{prop:F} shows that the $F_g$ provide the most
  condensed way to describe the non-planar sector of the
  $\Phi^3$-matricial QFT model.  All information about the genus-$g$
  sector is encoded in the $p(3g-3)$ rational numbers which form the
  coefficients in the $(3g-3)$-weighted polynomial in
  $\{\frac{\varrho_1}{\varrho_0},
  \frac{\varrho_2}{\varrho_0},\dots\}$. From these polynomials we
  obtain the $(1+\dots+1)$-point function with $B$ boundary components
  via $\G_g(z)=(2\lambda)^3 \hat{\mathrm{A}}^{\dag}_z F_g$
  followed by Theorem~\ref{finaltheorem}.
\end{remark}

\begin{lemma}
Whenever $(2g+B-2)>0$, the operator $\hat{\mathrm{N}}$ measures the 
Euler characteristics, 
\begin{align*} 
\hat{\mathrm{N}} \G_g(z_1|\dots|z_B) = (2g+B-2)\G_g(z_1|\dots|z_B) .
\end{align*}
\end{lemma}
\begin{proof}
Both cases with $(2g+B-2)=1$ are directly checked. The general case follows 
by induction from $[\hat{\mathrm{N}},\hat{\mathrm{A}}^{\dag g}_{z_J,z}]=\hat{\mathrm{A}}^{\dag g}_{z_J,z}$ 
in combination with Theorem~\ref{finaltheorem} and 
$\hat{\mathrm{N}} F_g=(2g-2)F_g$ for $g\geq 2$.
\end{proof}
\begin{corollary}
\begin{align*}
\hat{\mathrm{A}}_{\check{z}} \G_g(\bullet|z_2|\dots|z_B)
= (2\lambda)^{3} (2g+B-3)\G_g(z_2|\dots|z_B)
\end{align*}
whenever $(2g+B-3)>0$.
\end{corollary}
\noindent
Hence, up to a rescaling, $\hat{\mathrm{A}}_{\check{z}}$
indeed removes the boundary component previously located at $z$.  We
also have $\hat{\mathrm{A}}_{\check{z}} F_g=0$ for all $g\geq
1$ so that the $F_g$ play the r\^ole of a vacuum.
Note that $\mathcal{W}_0(z)$ cannot be produced by whatever $F_0$.

\subsection{The Laplacian}
\label{sec:Laplacian}

Let $\mathcal{Z}_V^{np} := \exp\big(\sum_{g=1}^\infty V^{2-2g} F_g\big)$
be the non-planar part of the partition function, understood as
formal power series in $\frac{1}{V^2}$. Consider the following operation
with $\mathcal{Z}_V^{np}$:
\begin{align}
O_V:=&\frac{1}{\mathcal{Z}_V^{np}}
\Big[\frac{V^2}{8\lambda^4} \hat{K}_z \hat{\mathrm{A}}^\dag_z + 
\Big(\hat{\mathrm{A}}^\dag_z +\frac{1}{\varrho_0 z^4} 
\frac{\partial}{\partial z} \Big)
\hat{\mathrm{A}}^\dag_z +\frac{V^2}{64 \lambda^4 z^4} \Big]
\mathcal{Z}_V^{np}.
\label{ident-Z}
\end{align}
Its coefficients are 
\begin{align*}
  [V^2] O_V
  =\frac{1}{8\lambda^4} \hat{K}_z \hat{\mathrm{A}}^\dag_z F_1
  + \frac{1}{64 \lambda^4 z^4} =0
\end{align*}
because of 
$\hat{\mathrm{A}}^\dag_z F_1=\frac{\G_1(z)}{(2\lambda)^3}$ and
$\hat{K}_z\G_1(z)+\frac{\lambda^3}{z^4}=0$ by \eqref{1punktneuV} and
\eqref{G0z} and
\begin{align}
  [V^{4-2g}] O_V
  &= \frac{1}{8\lambda^4} \hat{K}_z \hat{\mathrm{A}}^{\dag}_z F_g  
  +  \sum_{h=1}^{g-1}
  \big(\hat{\mathrm{A}}^{\dag}_z F_h\big)
  \big(\hat{\mathrm{A}}^{\dag}_z F_ {g-h}\big)
  + \Big(\hat{\mathrm{A}}^{\dag}_z +\frac{1}{\varrho_0
  z^4} \frac{\partial}{\partial z}\Big)
\big(\hat{\mathrm{A}}^{\dag}_z F_ {g-1}\big) 
\end{align}
for $g\geq 2$. Inserting $\hat{\mathrm{A}}^{\dag}_z F_h
=\frac{1}{(2\lambda)^3} \G_h$ and
$\big(\hat{\mathrm{A}}^{\dag}_z +\frac{1}{\varrho_0
  z^4} \frac{\partial}{\partial z}\big) \G_{g-1}(z)=\frac{1}{(2\lambda)^3}
\G_{g-1}(z|z)$ we see that $[V^{4-2g}] O_V$ equals 
$\frac{1}{\lambda(2\lambda)^6}$ times
the rhs of \eqref{1punktneuV}, hence vanishes. This means that 
$O_V\equiv 0=\mathcal{Z}^{np}_V O_V$.

We invert $\hat{K}_z$ via Proposition~\ref{Prop5}, divide by $z$ and apply 
$\hat{\mathrm{A}}_{\check{z}}$ given by the residue
in Definition~\ref{DefA}:
\begin{align*}
\frac{V^2}{8\lambda^4} 
\hat{\mathrm{N}} \mathcal{Z}_V^{np}
&= {-} \!\sum_{\ell=0}^\infty 
\Res\displaylimits_{z\to 0}\!\Big[dz  \frac{z^{3+2\ell} \varrho_\ell}{(3+2\ell)}
\Res\displaylimits_{z'\to 0}\!
\Big[ dz'(z')^2 K(z,z')\Big(\!\Big(
\hat{\mathrm{A}}^\dag_{z'} {+}\frac{1}{\varrho_0 (z')^4} 
\frac{\partial}{\partial z'} \Big)
\hat{\mathrm{A}}^\dag_{z'} \\
& \hspace*{5cm} {+}\frac{V^2}{64 \lambda^4 (z')^4} \Big)
\Big]\Big]\mathcal{Z}_V^{np}.
\end{align*}
We insert $K(z,z')=\frac{2}{(\mathcal{W}_0(z')
      -\mathcal{W}_0(-z')) ({z'}^2-z^2)}$ from
Proposition~\ref{Prop5} and commute the integration contours
from $|z'|<|z|$ to $|z|<|z'|$. This procedure picks up the poles of
$K(z,z')$ at $z=z'$ and $z=-z'$, i.e.\ 
\[
\Res\displaylimits_{z\to 0}\Res\displaylimits_{z'\to 0}
=\Res\displaylimits_{z'\to 0}\Big(
\Res\displaylimits_{z\to 0}
+\Res\displaylimits_{z\to z'}
+\Res\displaylimits_{z\to -z'}
\Big)\;.
\]
There is no pole at $z=0$; the other two poles give the same contribution.
Abbreviating 
$f(z',\varrho)=
\big(\big(\hat{\mathrm{A}}^\dag_{z'} {+}\frac{1}{\varrho_0 (z')^4} 
\frac{\partial}{\partial z'} \big)
\hat{\mathrm{A}}^\dag_{z'} +\frac{V^2}{64 \lambda^4 (z')^4} \big)
\mathcal{Z}_V^{np}$, we thus have
\begin{align*}
&{-} \sum_{\ell=0}^\infty 
\Res\displaylimits_{z\to 0}\Big[dz  \frac{z^{3+2\ell} \varrho_\ell}{(3+2\ell)}
\Res\displaylimits_{z'\to 0}
\Big[ dz'(z')^2 \frac{2}{(\mathcal{W}_0(z')
      -\mathcal{W}_0(-z')) ({z'}^2-z^2)}
f(z',\varrho)\Big]\Big]
\\
&={-} \sum_{\ell=0}^\infty 
\Res\displaylimits_{z'\to 0}\Big[ dz'(z')^2
\Big(\Res\displaylimits_{z\to z'}+\Res\displaylimits_{z\to -z'}\Big)
\Big[
\frac{z^{3+2\ell} \varrho_\ell}{(3+2\ell)}
\frac{2f(z',\varrho) dz
}{(\mathcal{W}_0(z')-\mathcal{W}_0(-z')) ({z'}^2{-}z^2)}
\Big]\Big]
\\
&=
\Res\displaylimits_{z'\to 0}\Big[ 
\sum_{\ell=0}^\infty   \frac{dz' (z')^{4+2\ell} \varrho_\ell}{(3+2\ell)}
\frac{2}{(\mathcal{W}_0(z') -\mathcal{W}_0(-z'))}
f(z',\varrho)\Big]\;.
\end{align*}
Inserting (\ref{Gzminusz}) and renaming $z'\to z$ we get
\begin{align}
\frac{V^2}{8\lambda^4} 
\hat{\mathrm{N}} \mathcal{Z}_V^{np}
&=  
\Res\displaylimits_{z\to 0}
\Big[ 
dz \,z^3 \tilde{\mathcal{R}}(z)
\Big(\Big(
\hat{\mathrm{A}}^\dag_{z} +\frac{1}{\varrho_0 (z)^4} 
\frac{\partial}{\partial z} \Big)
\hat{\mathrm{A}}^\dag_{z} +\frac{V^2}{64 \lambda^4 z^4} \Big)
\Big]\mathcal{Z}_V^{np},
\nonumber
\end{align}
where 
\begin{align*}
\tilde{\mathcal{R}}(z)=\frac{\sum_{\ell=0}^\infty \frac{ \varrho_\ell z^{2\ell}}{
(3+2\ell)}}{\sum_{j=0}^\infty \varrho_j  z^{2j}}
=\sum_{m=0}^\infty \tilde{\mathcal{R}}_m(\varrho)\, z^{2m}.
\end{align*}
The inverse of the denominator is given by \eqref{S-Bell}, without the 
$\frac{1}{z'}$ 
prefactor and in variables $z'\mapsto z$.
Its product with the numerator is 
\begin{align}
\tilde{\mathcal{R}}_m(\varrho)&= 
\frac{S_m(\varrho)}{3\cdot m!} - \sum_{k=1}^m 
\frac{\varrho_k}{(3+2k)\varrho_0} \frac{S_{m-k}(\varrho)}{(m-k)!}
=-\frac{2}{3} \sum_{k=1}^m 
\frac{k \varrho_k}{(3+2k)\varrho_0} \frac{S_{m-k}(\varrho)}{(m-k)!},
\label{Rmrho}
\end{align}
where we have used \eqref{xBell} for the first $S_m(\varrho)$ 
to achieve better control of signs.

The residue of $\frac{V^2}{64 \lambda^4 z^4}$ is immediate and can 
be moved to the lhs:
\begin{align*}
&\frac{V^2}{8\lambda^4} 
\Big(\hat{\mathrm{N}}-\frac{1}{24}\Big) \mathcal{Z}_V^{np}
\\
&=  \sum_{m=0}^\infty \tilde{\mathcal{R}}_m(\varrho)
\Res\displaylimits_{z\to 0}
\Big[ z^{3+2m}\,dz 
\Big[\Big(
\hat{\mathrm{A}}^\dag_{z} +\frac{1}{\varrho_0 z^4} 
\frac{\partial}{\partial z} \Big)
\hat{\mathrm{A}}^\dag_{z} 
\Big]\Big]\mathcal{Z}_V^{np}
\nonumber
\\
&=\Big[  
\sum_{k=0}^\infty 
\Big(-\frac{3(3+2k)\varrho_1 \varrho_{k+1}\tilde{\mathcal{R}}_1(\varrho)}{\varrho_0^3}
+\frac{3(3+2k)\varrho_{k+1} \tilde{\mathcal{R}}_2(\varrho)}{\varrho_0^2}
\Big)\frac{\partial}{\partial \varrho_k} 
\\
&+
\sum_{k,l=0}^\infty \frac{(3+2k)(3+2l)\tilde{\mathcal{R}}_1(\varrho)}{\varrho_0^2}
\varrho_{l+1}\frac{\partial}{\partial \varrho_l}
\varrho_{k+1}\frac{\partial}{\partial \varrho_k}
\\
&-\sum_{k,l=0}^\infty \frac{(3+2k)(3+2l)\tilde{\mathcal{R}}_{l+2}(\varrho)}{\varrho_0}
\Big(
\varrho_{k+1}\frac{\partial}{\partial \varrho_k}
\frac{\partial}{\partial \varrho_l}
+\frac{\partial}{\partial \varrho_l}
\varrho_{k+1}\frac{\partial}{\partial \varrho_k}
\Big)
\\
&+\sum_{k,l=0}^\infty (3+2k)(3+2l)\tilde{\mathcal{R}}_{k+l+3}(\varrho)
\frac{\partial}{\partial \varrho_k}
\frac{\partial}{\partial \varrho_l}
\\
&+\sum_{k=0}^\infty \frac{3(3{+}2k)\varrho_{k+1}\tilde{\mathcal{R}}_2(\varrho)}{\varrho_0^2}
\frac{\partial}{\partial \varrho_k}
-\sum_{k=0}^\infty \frac{(3{+}2k)(5{+}2k)\tilde{\mathcal{R}}_{k+3}(\varrho)}{\varrho_0}
\frac{\partial}{\partial \varrho_k}
\Big]\mathcal{Z}_V^{np}.
\end{align*}
Next we separate the $\varrho_0$-derivatives:
\begin{align*}
&\frac{V^2}{8\lambda^4} 
\Big(\hat{\mathrm{N}}-\frac{1}{24}\Big) \mathcal{Z}_V^{np}
\nonumber
\\
&= 
\Big[ 
\Big(\frac{9 \tilde{\mathcal{R}}_1(\varrho) \varrho_1^2}{\varrho_0^2}
- \frac{18 \tilde{\mathcal{R}}_{2}(\varrho)\varrho_1}{\varrho_0}
+9\tilde{\mathcal{R}}_{3}(\varrho)\Big)
\frac{\partial^2}{\partial \varrho_0^2}
\\
& +\Big(
-\frac{9\varrho_1^2 \tilde{\mathcal{R}}_1(\varrho)}{\varrho_0^3}
+\frac{18\varrho_{1}\tilde{\mathcal{R}}_2(\varrho)}{\varrho_0^2}
+\frac{15 \tilde{\mathcal{R}}_1(\varrho) \varrho_2}{\varrho_0^2}
-\frac{30 \tilde{\mathcal{R}}_{3}(\varrho)}{\varrho_0}\Big)
\frac{\partial}{\partial \varrho_0}
\\
&
+\sum_{k=1}^\infty 6(3{+}2k)\Big(
\tilde{\mathcal{R}}_{k+3}(\varrho)
-\frac{\tilde{\mathcal{R}}_{2}(\varrho) \varrho_{k+1}}{\varrho_0}
-\frac{\tilde{\mathcal{R}}_{k+2}(\varrho)\varrho_{1}}{\varrho_0}
+\frac{\tilde{\mathcal{R}}_1(\varrho)\varrho_{k+1}\varrho_{1}}{\varrho_0^2}\Big)
\frac{\partial}{\partial \varrho_k}
\frac{\partial}{\partial \varrho_0}
\\
&
+
\sum_{k,l=1}^\infty (3{+}2k)(3{+}2l)\Big(
\frac{ \varrho_{l+1}\varrho_{k+1} \tilde{\mathcal{R}}_1(\varrho)}{\varrho_0^2}
+ \tilde{\mathcal{R}}_{k+l+3}(\varrho)
- \frac{2\varrho_{k+1}\tilde{\mathcal{R}}_{l+2}(\varrho)}{\varrho_0}
\Big)
\frac{\partial}{\partial \varrho_k}
\frac{\partial}{\partial \varrho_l}
\\
&
+\sum_{k=1}^\infty (3k+2)
\Big(-\frac{3\varrho_1 \varrho_{k+1}\tilde{\mathcal{R}}_1(\varrho)}{\varrho_0^3}
+\frac{6\varrho_{k+1} \tilde{\mathcal{R}}_2(\varrho)}{\varrho_0^2}
\\
&\qquad\qquad 
+ \frac{(5+2k)\varrho_{k+2} \tilde{\mathcal{R}}_1(\varrho)}{\varrho_0^2}
-\frac{2(5+2k)\tilde{\mathcal{R}}_{k+3}(\varrho)}{\varrho_0}
\Big)\frac{\partial}{\partial \varrho_k}
\Big]\mathcal{Z}_V^{np}.
\end{align*}
We isolate $F_1$, i.e.\ 
$\mathcal{Z}_V^{np}=\varrho_0^{-\frac{1}{24}} 
\mathcal{Z}_V^{stable}$, where 
$\mathcal{Z}_V^{stable} := \exp\big(\sum_{g=2}^\infty V^{2-2g} F_g\big)$
is the stable partition function.
We commute the factor 
$\varrho_0^{-\frac{1}{24}}$ in front of $[~~]$ and move it to the other side:
\begin{align*}
&\frac{V^2}{8\lambda^4} \hat{\mathrm{N}} \mathcal{Z}_V^{stable}
\\
&= 
\Big[ 
\Big(\frac{49 \varrho_1^2 \tilde{\mathcal{R}}_1(\varrho)}{64\varrho_0^4}
- \frac{49 \varrho_1 \tilde{\mathcal{R}}_{2}(\varrho)}{32\varrho_0^3}
-\frac{5 \tilde{\mathcal{R}}_1(\varrho) \varrho_2}{8\varrho_0^3}
+\frac{105 \tilde{\mathcal{R}}_{3}(\varrho)}{64\varrho_0^2}\Big)
\\
&
+
\Big(\frac{9 \tilde{\mathcal{R}}_1(\varrho) \varrho_1^2}{\varrho_0^2}
- \frac{18 \tilde{\mathcal{R}}_{2}(\varrho)\varrho_1}{\varrho_0}
+9\tilde{\mathcal{R}}_{3}(\varrho)\Big)
\frac{\partial^2}{\partial \varrho_0^2}
\\
& 
+\Big(
-\frac{39\varrho_1^2 \tilde{\mathcal{R}}_1(\varrho)}{4\varrho_0^3}
+\frac{39\varrho_{1}\tilde{\mathcal{R}}_2(\varrho)}{2\varrho_0^2}
+\frac{15 \tilde{\mathcal{R}}_1(\varrho) \varrho_2}{\varrho_0^2}
-\frac{123 \tilde{\mathcal{R}}_{3}(\varrho)}{4\varrho_0}\Big)
\frac{\partial}{\partial \varrho_0}
\\
&
+\sum_{k=1}^\infty 6(3{+}2k)\Big(
\tilde{\mathcal{R}}_{k+3}(\varrho)
-\frac{\tilde{\mathcal{R}}_{2}(\varrho) \varrho_{k+1}}{\varrho_0}
-\frac{\tilde{\mathcal{R}}_{k+2}(\varrho)\varrho_{1}}{\varrho_0}
+\frac{\tilde{\mathcal{R}}_1(\varrho)\varrho_{k+1}\varrho_{1}}{\varrho_0^2}\Big)
\frac{\partial}{\partial \varrho_k}
\frac{\partial}{\partial \varrho_0}
\\
&
+
\sum_{k,l=1}^\infty (3{+}2k)(3{+}2l)\Big(
\frac{ \varrho_{l+1}\varrho_{k+1} \tilde{\mathcal{R}}_1(\varrho)}{\varrho_0^2}
+ \tilde{\mathcal{R}}_{k+l+3}(\varrho)
- \frac{2\varrho_{k+1}\tilde{\mathcal{R}}_{l+2}(\varrho)}{\varrho_0}
\Big)
\frac{\partial}{\partial \varrho_k}
\frac{\partial}{\partial \varrho_l}
\\
&
+\sum_{k=1}^\infty (3{+}2k)\Big(
-\frac{\tilde{\mathcal{R}}_{k+3}(\varrho)}{4 \varrho_0}
+\frac{25\varrho_{k+1} \tilde{\mathcal{R}}_2(\varrho)}{4\varrho_0^2}
+\frac{\tilde{\mathcal{R}}_{k+2}(\varrho)\varrho_{1}}{4\varrho_0^2}
-\frac{13\varrho_1 \varrho_{k+1}\tilde{\mathcal{R}}_1(\varrho)}{4\varrho_0^3}
\\
&\qquad\qquad 
+ \frac{(5{+}2k)\varrho_{k+2} \tilde{\mathcal{R}}_1(\varrho)}{\varrho_0^2}
-\frac{2(5{+}2k)\tilde{\mathcal{R}}_{k+3}(\varrho)}{\varrho_0}
\Big)\frac{\partial}{\partial \varrho_k}
\Big]\mathcal{Z}_V^{stable}.
\end{align*}
Expanding 
$\mathcal{Z}_V^{stable}
=:1+\sum_{g=2}^\infty V^{2-2g}\mathcal{Z}_g$, we have
\begin{align*}
\hat{\mathrm{N}}\mathcal{Z}_V^{stable}
&
=\sum_{g=2}^\infty (V^{-2})^{g-1}(2g-2) \mathcal{Z}_g
=2V^{-2}\frac{d}{dV^{-2}} \sum_{g=2}^\infty (V^{-2})^{g-1}
\mathcal{Z}_g
\\
&=2V^{-2}\frac{d}{dV^{-2}} \mathcal{Z}_V^{stable}.
\end{align*}
Consequently, we obtain a parabolic differential equation in $V^{-2}$
which is easily solved. Inserting 
\[
\tilde{\mathcal{R}}_1(\varrho)={-}\frac{2}{15} \frac{\varrho_1}{\varrho_0},~~
\tilde{\mathcal{R}}_2(\varrho)=\frac{2}{15} \frac{\varrho_1^2}{\varrho_0^2}
-\frac{4}{21} \frac{\varrho_2}{\varrho_0},~~
\tilde{\mathcal{R}}_3(\varrho)=
{-}\frac{2}{15} \frac{\varrho_1^3}{\varrho_0^3}
+\frac{34}{105} \frac{\varrho_1\varrho_2}{\varrho_0^2}
-\frac{2}{9} \frac{\varrho_3}{\varrho_0},
\]
we have:
\begin{theorem}
\label{thm:Deltarho}
When expressed in terms of the moments of the measure $\varrho$, 
the stable partition function is given as a formal
  power series in $V^{-2}$ by
\begin{align*}
 \mathcal{Z}_V^{stable}
&:=\exp\Big(\sum_{g=2}^\infty V^{2-2g} F_g(\varrho)\Big)
= \exp\Big( -\frac{(2\lambda)^4}{V^2} \Delta_{\varrho}+\frac{F_2(\rho)}{V^2}\Big) 1,\qquad
\end{align*}
where 
\begin{align}
F_2&= (2\lambda)^4
\Big(
-\frac{21 \varrho_1^3 }{160\varrho_0^5}
+\frac{29}{128} \frac{\varrho_1\varrho_2}{\varrho_0^4}
-\frac{35}{384} \frac{\varrho_3}{\varrho_0^3}
\Big),
\label{F2}
\\
\Delta_\varrho &:=
-\Big(-\frac{6\varrho_1^3}{5\varrho_0^3} + \frac{111
  \varrho_1\varrho_2}{70 \varrho_0^2} -
\frac{\varrho_3}{2\varrho_0}\Big)
\frac{\partial^2}{\partial \varrho_0^2}
-\Big(
\frac{2\varrho_1^3}{\varrho_0^4} - \frac{1097 \varrho_1\varrho_2}{280
  \varrho_0^3} +\frac{41 \varrho_3}{24 \varrho_0^2}\Big)
\frac{\partial}{\partial \varrho_0}
\nonumber
\\
&
-\sum_{k=1}^\infty (3+2k)\Big(
\Big(-\frac{2\varrho_1^2}{5\varrho_0^3} 
+ \frac{2 \varrho_2}{7\varrho_0^2}\Big)\varrho_{k+1}
-\frac{3\tilde{\mathcal{R}}_{k+2}(\varrho)\varrho_{1}}{2\varrho_0}
+\frac{3\tilde{\mathcal{R}}_{k+3}(\varrho)}{2}
\Big)
\frac{\partial^2}{\partial \varrho_k \partial \varrho_0}
\nonumber
\\
&
-\sum_{k,l=1}^\infty (3{+}2k)(3{+}2l)\Big(
{-}\frac{ \varrho_1 \varrho_{l+1}\varrho_{k+1}}{30\varrho_0^2}
- \frac{\varrho_{k+1}\tilde{\mathcal{R}}_{l+2}(\varrho)}{4\varrho_0}
- \frac{\varrho_{l+1}\tilde{\mathcal{R}}_{k+2}(\varrho)}{4\varrho_0}
\nonumber
\\
&\qquad\qquad  +\frac{\tilde{\mathcal{R}}_{k+l+3}(\varrho)}{4}
\Big)
\frac{\partial^2}{\partial \varrho_k\partial \varrho_l}
\nonumber
\\
&
-\sum_{k=1}^\infty (3+2k)\Big(
\Big(\frac{19 \varrho_1^2}{60 \varrho_0^4} 
- \frac{25 \varrho_2}{84 \varrho_0^3}\Big)\varrho_{k+1}
+\frac{\varrho_{1}\tilde{\mathcal{R}}_{k+2}(\varrho)}{16\varrho_0^2}
-\frac{\tilde{\mathcal{R}}_{k+3}(\varrho)}{16 \varrho_0}
\nonumber
\\
&\qquad\qquad 
- \frac{(5+2k)\varrho_1\varrho_{k+2}}{30\varrho_0^3}
-\frac{(5+2k)\tilde{\mathcal{R}}_{k+3}(\varrho)}{2\varrho_0}
\Big)\frac{\partial}{\partial \varrho_k}
\end{align}
and $\tilde{\mathcal{R}}_m(\varrho)$ given by \eqref{Rmrho}.
\end{theorem}
\noindent

Because we are essentially treating the Kontsevich model
\cite{Kontsevich:1992ti}, our $F_g$ are nothing else than the
generators of intersection numbers of tautological characteristic
classes\footnote{We will point out in Remark \ref{rmk:intersection}
that $F_g$ should be understood as generating function of intersection
numbers of $\kappa$-classes, which however translate into a subset of
intersection numbers of $\psi$-classes.} on the moduli space of
stable complex curves
\cite{Witten:1990hr, Kontsevich:1992ti, Itzykson:1992ya,
  Eynard:2016yaa}.  
The free energies $F_g$ are
listed in different conventions in the literature. The translation to
e.g.\ \cite{Itzykson:1992ya, Eynard:2016yaa} is as follows:
\begin{align*}
\text{\cite{Itzykson:1992ya}}: && 
(1-I_1)&=\varrho_0, & I_{k+1}&=-(2k+1)!! \varrho_k \quad \text{for }
k\geq 1,
\\
\text{\cite{Eynard:2016yaa}}: &&
(2-t_3)&=\varrho_0, &
t_{2k+3} &= -\varrho_k, \quad \text{for } k\geq 1.
\end{align*}

It is clear that Theorem~\ref{thm:Deltarho} translates into the same
statement for the generating function of intersection numbers. We have
given this formulation in the very beginning in
Theorem~\ref{Thm:Zstable}. There we adopt the conventions
in \cite{Itzykson:1992ya} but rename $I_k\equiv t_k=-(2k-1)!!\varrho_{k-1}$ and
$T_0\equiv (1-I_1)$. We also redefined
$\mathcal{R}_m(t)=(2m-1)!!\tilde{\mathcal{R}}_m(\varrho)$ as
well as $N=\frac{V}{(2\lambda)^2}$.  The formula can easily
be implemented in computer
algebra\footnote{
A first implementation in Mathematica is provided
via the arXiv page of this paper or via 
\url{http://wwwmath.uni-muenster.de/u/raimar/files/IntersectionNumbers.nb}
\\
It takes less than 35 seconds on an office desktop to compute
all intersection numbers up to $g=10$.} and quickly computes the 
free energies $F_g(t)$ to moderately large $g$. Several other
implementations exist, for instance
the powerful Sage programme 
\cite{admcycles} which performs many more
natural operations in the tautological ring.
An early implementation in Maple (up to $g=10$) can be found
in \cite{Xu:2007??}. These implementations were an important
consistency check for us. 
Algorithms to compute $\kappa,\delta,\lambda$-classes from 
$\psi$-classes are given in \cite{Faber:1999??}.

\begin{example}\normalfont
  For convenience we list
\begin{align*}
F_3&=
\frac{1225}{144}{\cdot}\frac{t_2^6}{6!T_0^{10}} 
+ \frac{193}{288} {\cdot} \frac{t_2^4 t_3}{4! T_0^9} 
+ \frac{205}{3456} {\cdot} \frac{t_2^2 t_3^2}{2!2! T_0^8} 
+ \frac{53}{1152} {\cdot} \frac{t_2^3 t_4}{3! T_0^8} 
+ \frac{583}{96768} {\cdot} \frac{t_3^3}{3! T_0^7} 
\\
&+ \frac{1121}{241920} {\cdot} \frac{t_2 t_3 t_4}{T_0^7} 
+ \frac{17}{5760} {\cdot} \frac{t_2^2 t_5}{2! T_0^7}
+ \frac{607}{1451520} {\cdot} \frac{t_4^2}{2! T_0^6} 
+ \frac{503}{1451520} {\cdot} \frac{t_3 t_5}{T_0^6} 
\\
&+ \frac{77}{414720} {\cdot} \frac{t_2 t_6}{T_0^6} 
+ \frac{1}{82944} {\cdot} \frac{t_7}{T_0^5}
\end{align*}
(already given in \cite[eq.~(5.30)]{Itzykson:1992ya}) and 
\begin{align*}
F_4&= \frac{1816871}{48} {\cdot} \frac{t_2^9}{9! T_0^{15}} 
+\frac{3326267}{1728} {\cdot} \frac{t_2^7 t_3}{7! T_0^{14}} 
+\frac{728465}{6912} {\cdot} \frac{t_2^5 t_3^2}{5! 2! T_0^{13}}
+\frac{43201}{6912} {\cdot} \frac{t_2^3t_3^3}{3! 3! T_0^{12}}
\\
&
+\frac{134233}{331776} {\cdot} \frac{t_2 t_3^4}{4! T_0^{11}}
+ \frac{70735}{864} {\cdot} \frac{t_2^6 t_4}{6! T_0^{13}}
+\frac{83851}{17280} {\cdot}\frac{t_2^4 t_3 t_4}{4! T_0^{12}}
+\frac{26017}{82944} {\cdot} \frac{t_2^2t_3^2 t_4}{2!2!T_0^{11}}
\\
&+\frac{185251}{8294400} {\cdot} \frac{t_3^3 t_4}{3! T_0^{10}}
+\frac{5609}{23040} {\cdot} \frac{t_2^3 t_4^2}{3!2!T_0^{11}}
 + \frac{177}{10240} {\cdot} \frac{t_2 t_3 t_4^2}{2! T_0^{10}}
+\frac{175}{165888}{\cdot} \frac{t_4^3}{3! T_0^9}
\\
&
+\frac{21329}{6912} {\cdot} \frac{t_2^5 t_5}{5! T_0^{12}}
+\frac{13783}{69120} {\cdot} \frac{t_2^3 t_3 t_5}{3!T_0^{11}}
+\frac{1837}{129600} {\cdot} \frac{t_2 t_3^2 t_5}{2! T_0^{10}}
+\frac{7597}{691200} {\cdot} \frac{t_2^2 t_4 t_5}{2! T_0^{10}}
\\
&
+ \frac{719}{829440} {\cdot} \frac{t_3 t_4 t_5}{T_0^9} 
+\frac{533}{967680}{\cdot} \frac{t_2 t_5^2}{2! T_0^9}
+\frac{2471}{23040} {\cdot} \frac{t_2^4 t_6}{4! T_0^{11}}
+\frac{7897}{1036800} {\cdot} \frac{t_2^2 t_3 t_6}{2! T_0^{10}}
\\
&
+\frac{1997}{3317760} {\cdot} \frac{t_3^2 t_6}{2! T_0^9}
{+}\frac{1081}{2322432} {\cdot} \frac{t_2 t_4 t_6}{T_0^9} 
{+}\frac{487}{18579456} {\cdot} \frac{t_5 t_6}{T_0^8} 
{+}\frac{4907}{1382400} {\cdot} \frac{t_2^3 t_7}{3! T_0^{10}}
\\
&+\frac{16243}{58060800} {\cdot} \frac{t_2 t_3 t_7}{T_0^9} 
+\frac{1781}{92897280} {\cdot} \frac{t_4 t_7}{T_0^8} 
+ \frac{53}{460800} {\cdot} \frac{t_2^2 t_8}{2! T_0^9}
+ \frac{947}{92897280} {\cdot} \frac{t_3 t_8}{T_0^8} 
\\
&+ \frac{149}{39813120} {\cdot} \frac{t_2 t_9}{T_0^8} 
+ \frac{1}{7962624} {\cdot} \frac{t_{10}}{T_0^7}.
\end{align*}
The first line agrees with \cite[Table II]{Itzykson:1992ya}.
The notation is such that the intersection numbers are easily 
identified, e.g.\  $\langle \tau_2\tau_3^4\rangle
\equiv \int_{\overline{\mathcal{M}}_{4,5}} \psi_1^2 \psi_2^3
  \psi_3^3\psi_4^3\psi_5^3= 
\frac{134233}{331776}$ or $\langle\tau_2^2\tau_4\tau_5\rangle
\equiv \int_{\overline{\mathcal{M}}_{4,4}} \psi_1^2 \psi_2^2
  \psi_3^4\psi_4^5=\frac{7597}{691200}$.
The very last number is  $\langle\tau_{3g-2}\rangle 
\equiv \int_{\overline{\mathcal{M}}_{g,1}} \psi_1^{3g-2}
=\frac{1}{24^{g}\cdot g!}$ for $g=4$, in agreement with 
\cite[eq. (5.31)]{Itzykson:1992ya}. The arXiv version v1
of this paper also gave $F_5$ and $F_6$ in an appendix, but there is
not really a need for them. 
\end{example}

\begin{remark}\label{rmk:intersection}
  As explained in \cite[sec.~6.7.3.1]{Eynard:2016yaa}, it is
  more appropriate to view $F_g$, after (inverse) Schur
  transform of the $\varrho_l$,
\begin{align}
(-(2l+1)!!) \frac{\varrho_l}{\varrho_0} =:
  -[x^l] \exp\Big(-\sum_{i=1}^\infty s_ix^i\Big)
  \label{Schur}
\end{align}
as generating
  function of intersection numbers of $\kappa$-classes on
  $\overline{\mathcal{M}}_{g,0}$.
  When passing to $\mathcal{G}_g(z_1|\dots|z_n)$
  by application of the boundary creation operator,
equivalently to the $\omega_{g,n}(z_1,\dots,z_n)$ of topological recursion,
  all mixed
  intersection numbers of $\psi$- and $\kappa$-classes are generated.
A factor $\frac{-(2l+1)!!}{z_i^{2l+3}}$ translates to 
$\psi_i^l$ and differentiation with respect to $s_l$ gives a factor $\kappa_l$
under the $\overline{\mathcal{M} }_{g,n}$-integral.
 In particular, the subfamily of intersection numbers only of $\psi$-classes
  and with power $\geq 2$ arises by multiple application of
  $\prod_{i=1}^n \frac{2+3l_i}{z_i^{5+2l_i}} \frac{\partial}{
    \partial \varrho_{l_i}}\Big|_{l_i\geq 1}$ 
 on $F_g$ and projecting to the part without $\varrho_{l\geq 1}$. 
 Every term in $F_g$ with $n$ factors of $t_{i\geq 2}$ thus gives a
 unique intersection number of $\psi$-classes on
 $\overline{\mathcal{M}}_{g,n}$. This justifies to view $F_g$ as
 intersection numbers of $\psi$-classes.
\end{remark}
\begin{example}\normalfont
  Consider $g=2$ and $F_2$ given by (\ref{F2}).
Inserting 
  $(-7!!) \frac{\varrho_3}{\varrho_0}=s_3-s_2s_1+\frac{s_1^3}{6}$,
  $(-5!!) \frac{\varrho_2}{\varrho_0}=s_2-\frac{s_1^2}{2}$ and 
  $(-3!!) \frac{\varrho_1}{\varrho_0}=s_1$
from \eqref{Schur} we get
\begin{align*}
F_2&=\frac{7}{240} \cdot \frac{1}{6!} \Big((-3!!) \frac{\varrho_1}{\varrho_0}\Big)^3
+\frac{29}{5760} \Big((-5!!) \frac{\varrho_2}{\varrho_0}\Big)
\Big((-3!!) \frac{\varrho_1}{\varrho_0}\Big)
+\frac{1}{1152} \Big((-7!!) \frac{\varrho_3}{\varrho_0}\Big)
  \\
  &=
\frac{7}{240} \frac{s_1^3}{3!}
+\frac{29}{5760} (s_2-\tfrac{s_1^2}{2}) s_1
+\frac{1}{1152}(s_3-s_2s_1+\tfrac{s_1^3}{6})
\\
&=\frac{43}{2880} \frac{s_1^3}{3!}
+\frac{1}{240} s_2s_1
+\frac{1}{1152}s_3
\end{align*}
from which one extracts
$\int_{\overline{\mathcal{M}}_{2,0}} \kappa_3=\frac{1}{1152}$,
$\int_{\overline{\mathcal{M}}_{2,0}} \kappa_2\kappa_1=\frac{1}{240}$ and
$\int_{\overline{\mathcal{M}}_{2,0}} (\kappa_1)^3=\frac{43}{2880}$. 
\\
Applying $\hat{A}^{\dag}_{z_1}$ directly to (\ref{F2}) gives
\begin{align}  
&\mathcal{G}_2(z_1)
\label{G21}\\
&= (2\lambda)^3 \hat{A}^\dag_{z_1}F_2=
  \frac{(2\lambda)^7}{\varrho_0^3} \Big\{
\frac{7}{288}\Big(\frac{-3!!\varrho_1 }{\varrho_0}\Big)^4
\frac{(-1!!)}{z_1^3}
+\frac{5}{144}
\Big(
\frac{-3!!\varrho_1}{\varrho_0}
\Big)^2
\Big(\frac{-5!!\varrho_2}{\varrho_0}\Big)  
\frac{(-1!!)}{z_1^3}
\nonumber
\\
&+\frac{29}{5760}
\Big(\frac{-5!!\varrho_2}{\varrho_0}\Big)^2  
\frac{(-1!!)}{z_1^3}
+\frac{11}{1440}\frac{(-7!!\varrho_3)}{\varrho_0}
 \frac{(-3!!\varrho_1)}{\varrho_0}
\frac{(-1!!)}{z_1^3}
+\frac{1}{1152}
\Big(-9!! \frac{\varrho_4}{\varrho_0}\Big)
\frac{(-1!!)}{z_1^{3}}
\nonumber
\\
&+
\frac{7}{288}\Big(\frac{(-3!!\varrho_1) }{\varrho_0}\Big)^3
\frac{(-3!!)}{z_1^5}
+\frac{29}{1440} \frac{(-3!!\varrho_1)}{\varrho_0}
\frac{(-5!!\varrho_2)}{\varrho_0}
\frac{(-3!!)}{z_1^5}
+\frac{1}{384} 
\frac{(-7!!\varrho_3)}{\varrho_0}
\frac{(-3!!)}{z_1^5}
\nonumber
\\
&
+\frac{7}{480} 
\Big(
\frac{-3!!\varrho_1}{\varrho_0}
\Big)^2\frac{(-5!!)}{z_1^7}
+\frac{29}{5760}
\frac{(-5!!\varrho_2)}{\varrho_0}
\frac{(-5!!)}{z_1^7}
+\frac{29}{5760}
\frac{(-3!!\varrho_1)}{\varrho_0}
\frac{(-7!!)}{z_1^9}
\nonumber
\\
&+\frac{1}{1152}
\frac{(-9!!)}{z_1^{11}}\Big\}\;.
\nonumber
\end{align}
The very last term encodes the intersection number
$\int_{\overline{\mathcal{M}}_{2,1}} \psi_1^4=\frac{1}{1152}$.
All other involve $\kappa$-classes. To extract them we have to pass
via \eqref{Schur} to the
$s_l$-variables:
\begin{align}  
\mathcal{G}_2(z_1)
  &=
  \frac{(2\lambda)^7}{\varrho_0^3} \Big\{
\Big(\frac{1}{1152}s_4+\frac{13}{1920}s_1s_3
+\frac{53}{5760}\frac{s_2^2}{2}
+\frac{259}{5760}\cdot\frac{s_1^2}{2}s_2+\frac{29}{128}\cdot\frac{s_1^4}{24}\Big)
\frac{(-1!!)}{z_1^{3}}
\nonumber
\\
&
+
\Big(\frac{1}{384}s_3+\frac{101}{5760}s_2s_1
+\frac{169}{1920}\cdot\frac{s_1^3}{6}\Big)
\frac{(-3!!)}{z_1^5}
+\Big(\frac{29}{5760}s_2 +\frac{139}{5760}\frac{s_1^2}{2}
\Big)\frac{(-5!!)}{z_1^7}
\nonumber
\\
&
+\frac{29}{5760}s_1 \frac{(-7!!)}{z_1^9}
+\frac{1}{1152}\frac{(-9!!)}{z_1^{11}}\Big\}\;.
\label{G21s}
\end{align}
We extract e.g.\
$\int_{\overline{\mathcal{M}}_{2,1}} \kappa_2(\kappa_1)^2=\frac{259}{5760}$,
$\int_{\overline{\mathcal{M}}_{2,1}} \kappa_2\kappa_1 \psi_1=\frac{101}{5760}$,
$\int_{\overline{\mathcal{M}}_{2,1}} (\kappa_1)^2 (\psi_1)^2=\frac{139}{5760}$
and so on. These intersection numbers agree with the implementation
in \texttt{admcycles} \cite{admcycles}.
\\
On the other hand we can apply $\hat{A}_{z_1,z_2}$ to \eqref{G21}
in order to extract intersection numbers on
$\overline{\mathcal{M}_{2,2}}$. We only give the part without
$\varrho_{l\geq 1}$ which encodes intersection numbers of $\psi$-classes:
\begin{align*}  
\mathcal{G}_2(z_1|z_2)
&= 
  \frac{(2\lambda)^{10}}{\varrho_0^4} \Big\{
\frac{1}{384} \frac{(-9!!)}{z_2^{11}}
\frac{(-3!!)}{z_1^5}
+\frac{29}{5760}
\frac{(-7!!)}{z_2^9}
\frac{(-5!!)}{z_1^7}
+\frac{29}{5760}
\frac{(-5!!)}{z_2^7}
\frac{(-7!!)}{z_1^9}
\\
&
+\frac{1}{1152}
\frac{(-11!!)}{z_2^{13}}
\frac{(-1!!)}{z_1^{3}}
+\frac{1}{384}
\frac{(-3!!)}{z_2^5}
\frac{(-9!!)}{z_1^{11}}
+\frac{1}{1152}
\frac{(-1!!)(-11!!)}{z_2^3z_1^{13}}
\\
&+\mathcal{O}(\varrho_{l\geq 1})\Big\}\;.
\end{align*}
This gives the intersection numbers
$\int_{\overline{\mathcal{M}}_{2,2}} \psi_1^0\psi_2^5=\frac{1}{1152}$,
$\int_{\overline{\mathcal{M}}_{2,2}} \psi_1\psi_2^4=\frac{1}{384}$ and
$\int_{\overline{\mathcal{M}}_{2,2}} \psi_1^2\psi_2^3=\frac{29}{5760}$
as well as the symmetric copies $\psi_1\leftrightarrow \psi_2$.
The last case arises by application of 
$\prod_{i=1}^2(2\lambda)^3\frac{2+3l_i}{z_i^{5+2l_i}} \frac{\partial}{
  \partial \varrho_{l_i}}\Big|_{l_1=1,l_2=2}$
to the term 
$(2\lambda)^4\frac{29}{128} \frac{\varrho_1\varrho_2}{\varrho_0^4}$
of $F_2$. In this sense $F_2$ itself can be viewed as generating function
of the subset of intersection numbers of
$\psi$-classes on $\overline{\mathcal{M}}_{2,n}$ where the power
of every $\psi_i$ is $\geq 2$.
\end{example}

\subsection{A deformed Virasoro algebra}
\label{sec:Virasoro}

We return to \eqref{ident-Z} with $O_V\equiv 0$, but instead of
applying the inverse of
$\hat{K}_z$ we directly take the residue to define a family of operators
on rational functions $f(\varrho)$ of $\varrho_l$:
\begin{align*}
\tilde{L}_n f(\rho):= \Res_{z\to 0} \Big[z^{3+2n}
\Big(\frac{V^2}{8\lambda^4} \hat{K}_z \hat{\mathrm{A}}^\dag_z + 
(\hat{\mathrm{A}}^\dag_z )^2 +\frac{1}{\varrho_0 z^4} \frac{\partial 
\hat{\mathrm{A}}^\dag_z }{\partial z} 
+\frac{V^2}{64\lambda^4 z^4} \Big) f(\varrho) dz\Big].
\nonumber
\end{align*}
By construction, $\tilde{L}_n \mathcal{Z}^{np}_V=0$. Recall that in
the Kontsevich model one has $L_n\mathcal{Z}$ for the \emph{full}
partition function and generators $L_n$ of a Virasoro algebra (or
rather a Witt algebra). As explained below, these $\tilde{L}_n$ do not satisfy
the commutation relations of the Virasoro algebra exactly. 
An explicit expression is obtained from 
\begin{align*}
\hat{K}_z \hat{\mathrm{A}}^{\dag}_z &=
\sum_{l= 0}^{\infty} \sum_{j=0}^{l} \frac{(3+2l)\varrho_{l-j} }{z^{4+2j}}
\frac{\partial}{\partial \varrho_l},
\\
\frac{1}{\varrho_0 z^4} \frac{\partial}{\partial z}
\hat{\mathrm{A}}^{\dag}_z
&=
\sum_{l= 0}^\infty \Big(\frac{3(3+2l) \varrho_{l+1}}{
\varrho_0^2 z^8}-\frac{(3+2l)(5+2l)}{\varrho_0 z^{10+2l}}\Big)
\frac{\partial}{\partial \varrho_l},
\\
\hat{\mathrm{A}}^{\dag}_z\hat{\mathrm{A}}^{\dag}_z 
&=\sum_{k= 0}^{\infty} \Big(\frac{(5+2k) \varrho_{k+2}}{
\varrho_0 z^3}-\frac{5+2k}{z^{7+2k}}\Big)\frac{(3+2k)}{
\varrho_0 z^3}
\frac{\partial}{\partial \varrho_k}
\\
&+
\sum_{k= 0}^{\infty} \Big({-}\frac{3\varrho_{1}}{
\varrho_0 z^3}+\frac{3}{z^{5}}\Big)
\Big(\frac{(3{+}2k) \varrho_{k+1}}{
\varrho_0^2 z^3}\Big)
\frac{\partial}{\partial \varrho_k}
\\
&+
\sum_{l,k= 0}^{\infty} \frac{(3{+}2l) (3{+}2k) \varrho_{k+1} \varrho_{l+1}}{
\varrho_0^2 z^6}
\frac{\partial^2}{\partial \varrho_l \partial \varrho_k}
\\
& -
\sum_{l,k= 0}^{\infty} \frac{2(3{+}2l)(3{+}2k) \varrho_{l+1}}{
\varrho_0 z^{8+2k}}
\frac{\partial^2}{\partial \varrho_l \partial \varrho_k}
+ 
\sum_{l,k= 0}^{\infty}
\frac{(3{+}2l)(3{+}2k)}{z^{10+2l+2k}}
\frac{\partial^2}{\partial \varrho_l \partial \varrho_k}.
\end{align*}
Evaluating the residues and defining $A=\frac{(2\lambda)^4}{4V^2}$ and rescaling
$L_n:=A \tilde{L}_n$ gives 
\begin{align*}
L_0&=\frac{1}{16}+ \frac{1}{2}\sum_{l=0}^\infty (3+2l)\varrho_l 
\frac{\partial} {\partial \varrho_l},
\\
L_1&=\frac{1}{2}\sum_{l=0}^\infty 
(5{+}2l)\varrho_{l}\frac{\partial}{\partial \varrho_{l+1}}
+A\Big(\sum_{k=0}^\infty \frac{(3{+}2k)}{\varrho_0^2}
\varrho_{k+1}\frac{\partial}{\partial \varrho_k}
{-} \frac{3\varrho_1}{\varrho_0^2}\Big)
\sum_{l=0}^\infty (3{+}2l)\varrho_{l+1}
\frac{\partial}{\partial \varrho_l}
\nonumber
\\
& \hspace*{-0.5cm} \text{and for $n\geq 2$: } \nonumber
\\
\nonumber
L_{n}&=\frac{1}{2}
\sum_{l=0}^\infty 
(3{+}2n{+}2l)\varrho_{l}\frac{\partial}{\partial \varrho_{n+l}}
+A
 \delta_{n,2} \sum_{l=0}^\infty \frac{6(3+2l) \varrho_{l+1}}{
\varrho_0^2}\frac{\partial}{\partial \varrho_{l}}
\nonumber
\\
&-A\frac{2(2n{-}3)(2n{-}1)}{\varrho_0}
\frac{\partial}{\partial \varrho_{n-3}}
+A 
\sum_{l= 0}^{n-3}(3+2l)(2n-2l-3)
\frac{\partial^2}{\partial \varrho_l \partial \varrho_{n-3-l}}
\nonumber
\\
&-A\sum_{l= 0}^{\infty} \frac{2(3+2l)(2n-1) \varrho_{l+1}}{
\varrho_0}
\frac{\partial^2}{\partial \varrho_{n-2} \partial \varrho_l}.
\nonumber
\end{align*}
To write it in a more compact way, it is convenient to introduce the
differential operator
\begin{align}\label{roh-1}
 \hat{D}:=
 \sum_{l=0}^\infty \frac{(3+2l)\varrho_{l+1}}{\varrho_0}
\frac{\partial}{\partial \varrho_l}\;.
\end{align}
Note that $\frac{\partial}{\partial \varrho_{l}}
\hat{D}\neq\hat{D}\frac{\partial}
{\partial \varrho_{l}}$. The result is:
\begin{lemma} \label{LemmaVira}
The nonplanar partition $\mathcal{Z}_V^{np}
:= \exp\Big(\sum_{g=1}^\infty V^{2-2g} F_g\Big)$ satisfies 
the constraints $L_n\mathcal{Z}^{np}_V=0$ for all $n\in
\mathbb{N}$, where 
\begin{align}\label{Virasoro}
L_0&=\frac{1}{16}+ \frac{1}{2}\sum_{l=0}^\infty (3+2l)\varrho_l 
\frac{\partial} {\partial \varrho_l},
\\
L_1&=\frac{1}{2}\sum_{l=0}^\infty 
(5+2l)\varrho_{l}\frac{\partial}{\partial \varrho_{l+1}}
+A\hat{D}^2
\nonumber
\\
& \hspace*{-0.5cm} \text{and for $n\geq 2$: } \nonumber
\\
\nonumber
L_{n}&=\frac{1}{2}
\sum_{l=0}^\infty 
(3{+}2n{+}2l)\varrho_{l}\frac{\partial}{\partial \varrho_{n+l}}
+A 
\sum_{l= 0}^{n-3}(3+2l)(2n-2l-3)
\frac{\partial^2}{\partial \varrho_l 
\partial \varrho_{n-3-l}}
\nonumber
\\
&-2A(2n-1)\frac{\partial}
{\partial \varrho_{n-2}}\hat{D},
\nonumber
\end{align}
where $\hat{D}$ is the differential operator defined by \eqref{roh-1} 
and $A=\frac{(2\lambda)^4}{4V^2}$.
\end{lemma}

With the commutation rules
\begin{align*}
\Big[\hat{D},\varrho_l\Big]
&= \frac{3+2l}{\varrho_0} \varrho_{l+1}\;, & l\geq 0\;,
\\
\Big[\hat{D},\frac{\partial}{\partial
  \varrho_l} \Big]
&= -\frac{1+2l}{\varrho_0} \frac{\partial}{\partial \varrho_{l-1}}
& l\geq 1\;, && 
\Big[\hat{D},\frac{\partial}{\partial
  \varrho_0} \Big]
&= \frac{1}{\varrho_0} \hat{D}\;
\nonumber,
\end{align*}
we end up after long but straightforward computation:
\begin{lemma}
 The generators $L_n$ of Lemma \ref{LemmaVira} obey the commutation relation
 \begin{align*}
  [L_0,L_n]=-nL_n
 \end{align*}
and for any $m,n\geq 1$,
\begin{align*}
[L_m,L_n]
&= (m-n) L_{m+n}
-4A (m+1)B_{n-2}L_{m-1}
+4A (n+1)B_{m-2}L_{n-1}
\\
&- 4A \delta_{m,1} \frac{n(n+1)}{\varrho_0^2} L_{n-2}
+ 4A \delta_{n,1} \frac{m(m+1)}{\varrho_0^2} L_{m-2}
\nonumber
\end{align*}
where 
\[
B_m:= (2m+3)\frac{\partial}{\partial \varrho_{m}}
\frac{1}{\varrho_0} \quad\text{for }m\geq 0\;,\qquad
B_{-1}:= -\frac{1}{2}\Big\{\hat{D},
\frac{1}{\varrho_0} \Big\} \;.
\]
\end{lemma}
\begin{remark}
  The differential operator $\hat{D}$ has its origin in the implicit
  definition of the constant $c$ \eqref{implicit} and the dependence
  of $\varrho_l$ on $c$. Since the expression
 \begin{align*}
  \varrho_{-1}:=-\int_1^{\Lambda^2} dY\;\frac{\varrho(Y)}{\sqrt{Y+c}}, 
\qquad  (\varrho_{-1}=c\;\text{for}\;D=0)
 \end{align*}
diverges for any $D>0$ in the limit $\Lambda\to\infty$, it is necessary
to reconstruct the analogue of the 
derivative $\frac{\partial}{\partial c}$ 
through the differential  operator $\hat{D}$. 
Replacing the differential operator by 
$\hat{D}\mapsto \frac{\partial}{\partial \varrho_{-1}}$ and the generators by
 \begin{align*}
  L_n \mapsto L_n+\frac{1}{2}\varrho_{-1}\frac{\partial}{\partial \varrho_{-1}},
 \end{align*}
 recovers the original undeformed Virasoro algebra. As explained
 above, $\varrho_{-1}$ and consequently the standard Virasoro
 generators do not exist in dimension $D>0$.  The renormalisation 
 necessary for $D>0$ alters the definition of $c$ and prevents the
 construction of $L_{-1}$ and $F_0$ which in $D=0$ depend on $\varrho_{-1}$.
 Higher topologies ($\chi\leq 0$) are not affected because any 
 explicit $\varrho_{-1}$-dependence drops out.
\end{remark}

\section{Summary}
\label{sec:summary}

The construction of the renormalised $\Phi^3_D$-QFT model on 
noncommutative geometries of dimension $D\leq 6$ is now complete. 
After the previous solution of the planar sector in
\cite{Grosse:2016pob, Grosse:2016qmk} we established
in this paper an algorithm to compute any correlation function 
$
G^{(g)}_{|\narrowunderline{p}_1^1...\narrowunderline{p}^1_{N_1}|...|
\narrowunderline{p}_1^B...\narrowunderline{p}^B_{N_B}|}$ of genus $g\geq 1$:

\begin{enumerate}\itemsep 0pt
\item Compute the free energy $F_g(t)$ via Theorem~\ref{Thm:Zstable} and 
the note thereafter. It encodes the $p(3g-3)$ intersection numbers 
of $\psi$-classes (all with a power $\geq 2$) on the moduli spaces
$\overline{\mathcal{M}}_{g,n}$ of stable complex curves of genus $g$. 
Take $F_1=-\frac{1}{24} \log T_0$ for $g=1$. Alternatively, start
from intersection numbers obtained by other methods (e.g.\
\cite{admcycles}). 

\item Change variables to $\varrho_0=1-t_0$ and $\varrho_l
=-\frac{t_{l+1}}{(2l+1)!!}$, where $\varrho_l$ are given by 
\eqref{moments} for the measure \eqref{measure0} and with $c$ 
implicitly defined by (\ref{implicit}).

\item Apply to the resulting $F_g(\varrho)$ according to 
Proposition~\ref{prop:F} and Theorem~\ref{finaltheorem} 
the boundary creation operators $
\hat{\mathrm{A}}^{\dag g}_{z_1,\dots,z_B}\circ 
\dots \hat{\mathrm{A}}^{\dag g}_{z_1,z_2} \circ 
\hat{\mathrm{A}}^{\dag g}_{z_1}$ defined in Definition \ref{DefOp}.
Multiply by $(2\lambda)^{4g+3B-4}$ to obtain
$\G_g(z_1|\cdots|z_B)$. 

\item Pass to 
$\G_g(z_1^1...z^1_{N_1}|...|z_1^B...z^B_{N_B})$ via difference quotients 
\eqref{recursiveGBg}, where $X^\beta_{k_\beta}=(z^\beta_{k_\beta})^2-c$.

\item Specify to $z^\beta_{k_\beta}\mapsto (
4F_{\narrowunderline{p}^\beta_{k_\beta}}^2+c)^{1/2}$ to obtain 
$G^{(g)}_{|\narrowunderline{p}_1^1...\narrowunderline{p}^1_{N_1}|...|
\narrowunderline{p}_1^B...\narrowunderline{p}^B_{N_B}|}$, where 
$F_{\narrowunderline{p}}$ arises by mass-renormalisation from the 
$E_{\narrowunderline{p}}$ in the initial action \eqref{action3} of the model.
\label{algo-last}
\end{enumerate}
Our work was essentially a reverse engineering in opposite order. The
last step \ref{algo-last}.\ was given to us by the formal partition
function of the model. From there we had to climb up to the formula
for the intersection numbers. 

We remark that, in spite of the relation to the integrable Kontsevich
model \cite{Kontsevich:1992ti}, this $\Phi^3_D$-model provides a fascinating toy model for a
quantum field theory which shows many facets of renormalisation. Our
exact formulae can be expanded about $\lambda=0$ via \eqref{implicit} and
agree with the usual perturbative renormalisation which in $D=6$ needs
Zimmermann's forest formula \cite{Zimmermann:1969jj} (see
\cite{Grosse:2016qmk}). Also note that 
at fixed genus $g$ one expects $\mathcal{O}(n!)$ graphs with $n$
vertices so that a convergent summation at fixed $g$ cannot be
expected a priori. Moreover, in $D=6$ the $\beta$-function of the
coupling constant is positive for real $\lambda$, which in this
particular case poses not the slightest problem for summation. 

What remains to understand is the resummation in the genus, i.e.\
$\sum_{g=1}^\infty V^{2-2g} \G_g(z)$ or $\sum_{g=2}^\infty N^{2-2g}
F_g(t)$. All intersection numbers are positive for $t_l>0$, which
corresponds to $\varrho_l<0$ for $l\geq 1$. Because of the
$\lambda^2$-prefactor in front of \eqref{measure0} and the definition
\eqref{moments} of the $\varrho_l$, we have $t_l>0$ for real
$\lambda$. Therefore, the sum over the genus must diverge for $\lambda
\in \mathbb{R}$, which is not surprising because in this case the
action \eqref{action3} is unbounded from below. In contrast, it was
observed in \cite{Grosse:2016qmk} that for the planar sector it is
better to take $\lambda\in \mathbb{R}$. The final challenge of this
model is to establish that $\sum_{g=2}^\infty N^{2-2g} F_g(t)$ is
Borel summable for $t_l<0$, which would achieve convergence of the
genus expansion in two disks in the complex $\lambda$-plane tangent
from above and below the real axis at $\lambda=0$.
Important progress in this direction was recently achieved in 
\cite{Eynard:2019mps} with the proof of factorial bounds 
$|F_g| \leq r^{-g}\Gamma(\beta g) $ for some $r>0$ and $\beta\leq 5$.

\subsection*{Acknowledgements}

We thank Ga\"etan Borot, Johannes Schmitt and Alexander Alexandrov 
for valuable feedback. We are grateful to an anonymous referee for
a thorough verification and numerous helpful comments.
This feedback allowed us to eliminate many typos and
to fill a larger gap in the previous proof of Proposition \ref{prop:F}.
Our work was partially supported by the Deutsche
Forschungsgemeinschaft (DFG) under the coordinated programmes SFB 878
and RTG 2149 and the Cluster of Excellence\footnote{``Gef\"ordert
  durch die Deutsche Forschungsgemeinschaft (DFG) im Rahmen der
  Exzellenzstrategie des Bundes und der L\"ander EXC 2044 –390685587,
  Mathematik M\"unster: Dynamik--Geometrie--Struktur"} ``Mathematics
M\"unster'', as well as by the Erwin Schr\"odinger Institute in
Vienna.  A.H. is greatful to Akifumi Sako for fruitful discussions and
for the hospitality at the Tokyo University of Science, where a major
part of this work was developed.

\appendix 

\section{Lemmas relevant for 
Theorem \ref{finaltheorem}}
\label{appendix1}

\begin{assumption}\label{conj1}
We assume that $\G_g(z)$ is, for $g\geq 1$, a function of $z$ and of 
$\varrho_0,\dots,\varrho_{3g-2}$ (true for $g=1$). We take eq.\
\eqref{eqfinalthm} and in particular   
$\G_g(z|z\triangleleft_J):=(2\lambda)^3 \hat{\mathrm{A}}^{\dag g}_{z_J,z} 
\G_g(z\triangleleft_J)$
as a {\sffamily definition} of a family of functions
$\G_g(z_{1}|z_J)$ and derive equations for that family. 
\end{assumption}

\begin{lemma}\label{lemma5}
Let $J=\{2,...,B\}$. Then under Assumption~\ref{conj1}
and with Definition~\ref{defint} of 
the operator $\hat{K}_{z_1}$ one has
\begin{align*}
&\hat{K}_{z_1} \G_g(z_1|z\triangleleft_J)
\\
&=\frac{8\lambda^3}{z_1^2}
\bigg(\sum_{l= 0}^{3g-3+|J|} 
(3+2l)\frac{\partial \G_g(z\triangleleft_J)}{\partial \varrho_l}
\sum_{k=0}^{l}\frac{\varrho_k}{z_1^{2+2l-2k}}
+\sum_{\beta\in J}\frac{1}{z_\beta} \frac{\partial}{\partial z_\beta}
\G_g(z\triangleleft_J)\bigg).
\end{align*}
\end{lemma}
\begin{proof}
Take Definition \ref{DefOp} for $\hat{\mathrm{A}}^{\dag g}_{z_J,z} 
\G_g(z\triangleleft_J)$ 
and apply Lemma \ref{lemmaopK}.
\end{proof}

\begin{lemma}\label{lemma6} Let $J=\{2,...,B\}$.
Then under Assumption~\ref{conj1} one has 
\begin{align*}
&\frac{8\lambda^3}{z_\beta} \frac{\partial}{\partial z_\beta}
\frac{\G_g(z_1|z\triangleleft_{J\backslash\{\beta\}})
-\G_g(z_\beta|z\triangleleft_{J\backslash\{\beta\}})}{z_1^2-z_\beta^2}
+2\lambda \G_0(z_1|z_\beta )\G_{g}(z_1|z\triangleleft_{J\backslash \{\beta \}})
\\
&=
(2\lambda)^6
\Bigg[
\sum_{l= 0}^{3g-4+|J|} 
\Bigg(
-\sum_{n=0}^1 
\frac{(3+2l)(1+2n) \varrho_{l+1}}{\varrho_0 
z_1^{4-2n}z_\beta^{3+2n}}
\\
&\hspace*{5cm}
+\sum_{n=0}^{l+2}
\frac{(3+2l)(1+2n)}{z_1^{6+2l-2n}z_\beta^{3+2n}}
\Bigg)
\frac{\partial \G_g(z\triangleleft_{J\backslash\{\beta \}})}{\partial \varrho_l}
\\
&+\sum_{i\in J\backslash \{\beta \}}
\sum_{n=0}^1\frac{1+2n}{\varrho_0 z_i z_1^{4-2n}z_\beta^{3+2n}}
\frac{\partial \G_g(z\triangleleft_{J\backslash\{\beta \}})}{\partial z_i}
\Bigg].
\end{align*}
\end{lemma}
\begin{proof}
Definition \ref{DefOp} gives with 
$\frac{\frac{1}{z_1^{3+2j}}-\frac{1}{y^{3+2j}}}{z_1^2-y^2}=
-\sum_{l=0}^{2j+2}\frac{z_1^ly^{2j+2-l}}{z_1^{3+2j}y^{3+2j}(z_1+y)}$
for the first term
\begin{align*}
&\frac{\G_g(z_1|z\triangleleft_{J\backslash\{\beta\}})
-\G_g(z_\beta|z\triangleleft_{J\backslash\{\beta\}})}{
(2\lambda)^3(z_1^2-z_\beta^2)}
\\
&= \sum_{l=0}^{3g-4+|J|}
\left(
-\frac{(3+2l) \varrho_{l+1}}{\varrho_0}
\Bigg(\frac{\frac{1}{z_1^3}-\frac{1}{z_\beta^3}}{z_1^2-z_\beta^2}\Bigg)
+(3+2l) 
\Bigg(\frac{\frac{1}{z_1^{5+2l}}-\frac{1}{z_\beta^{5+2l}}}{z_1^2-z_\beta^2}\Bigg)
\right)
\frac{\partial \G_g(z\triangleleft_{J\backslash\{\beta \}})}{\partial \varrho_l}
\\
&+\sum_{i\in J\backslash \{\beta \}}\frac{1}{\varrho_0 z_i} 
\Bigg(\frac{\frac{1}{z_1^3}-\frac{1}{z_\beta^3}}{z_1^2-z_\beta^2}\Bigg)
\frac{\partial \G_g(z\triangleleft_{J\backslash\{\beta \}})}{\partial z_i}
\\
&=\sum_{l= 0}^{3g-4+|J|} 
\Bigg(
\frac{(3+2l) \varrho_{l+1}}{\varrho_0}
\frac{\sum_{n=0}^{2}z_1^n z_\beta^{2-n}}{z_1^{3}z_\beta^3(z_1+z_\beta)}
\\
&\qquad\qquad -(3+2l)
\frac{\sum_{n=0}^{2l+4}z_1^n z_\beta^{2l+4-n}}{z_1^{5+2l}z_\beta^{5+2l}
(z_1+z_\beta)}
\Bigg)
\frac{\partial \G_g(z\triangleleft_{J\backslash\{\beta \}})}{\partial \varrho_l}
\\
&-\sum_{i\in J\backslash \{\beta \}}\frac{1}{\varrho_0 z_i}
\frac{\sum_{n=0}^{2}z_1^n z_\beta^{2-n}}{z_1^{3}z_\beta^3(z_1+z_\beta)}
\frac{\partial \G_g(z\triangleleft_{J\backslash\{\beta \}})}{\partial z_i}.
\end{align*}
The second term (after dividing by $(2\lambda)^3$) reads
\begin{align*}
&\frac{(2\lambda)}{(2\lambda)^3} 
\G_0(z_1|z_\beta )\G_{g}(z_1|z\triangleleft_{J\backslash \{\beta \}})
\\
&=-\frac{2\lambda}{z_1z_\beta} \frac{\partial}{\partial z_\beta} 
\frac{1}{(z_1+z_\beta)}\Bigg[
\sum_{l=0}^{3g-4+|J|} 
\Big(\frac{-(3+2l) \varrho_{l+1}}{\varrho_0 z_1^3}
  \\
  &\hspace*{3cm} +\frac{(3+2l)}{z_1^{5+2l}} \Big)
\frac{\partial \G_g(z\triangleleft_{J\backslash\{\beta \}})}{\partial \varrho_l}
+\sum_{i\in J\backslash \{\beta \}}\frac{1}{\varrho_0z_i z_1^3}
\frac{\partial \G_g(z\triangleleft_{J\backslash\{\beta \}})}{\partial z_i}
\Bigg].
\end{align*}
The denominator $(z_1+z_\beta)$ cancels in the combination 
of interest:
\begin{align*}
&\frac{8\lambda^3}{z_\beta} \frac{\partial}{\partial z_\beta}
\frac{\G_g(z_1|z\triangleleft_{J\backslash\{\beta\}})
-\G_g(z_\beta|z\triangleleft_{J\backslash\{\beta\}})}{z_1^2-z_\beta^2}
+2\lambda \G_0(z_1|z_\beta )\G_{g}(z_1|z\triangleleft_{J\backslash \{\beta \}})
\\
&= 
\frac{(2\lambda)^6}{z_\beta} \frac{\partial}{\partial z_\beta}
\Bigg[
\sum_{l= 0}^{3g-4+|J|} 
\Bigg(
\frac{(3+2l) \varrho_{l+1}}{\varrho_0}
\frac{z_1^2+z_\beta^2}{z_1^4z_\beta^3}
\\
&
-(3{+}2l)
\frac{\sum_{n=0}^{l+2}z_1^{2n} z_\beta^{2l+4-2n}}{z_1^{6+2l}z_\beta^{5+2l}}
\Bigg)
\frac{\partial \G_g(z\triangleleft_{J\backslash\{\beta \}})}{\partial \varrho_l}
- \hspace*{-3mm} \sum_{i\in J\backslash \{\beta \}}\frac{1}{\varrho_0 z_i}
\frac{z_1^2+z_\beta^2}{z_1^4z_\beta^3}
\frac{\partial \G_g(z\triangleleft_{J\backslash\{\beta \}})}{\partial z_i}
\Bigg].
\end{align*}
The remaining $z_\beta$-derivative confirms the assertion.
\end{proof}

\begin{lemma}\label{lemma7}
Let $J=\{2,...,B\}$.
Then under Assumption~\ref{conj1} one has 
\begin{align*}
&2\lambda \sum_{h=1}^{g-1} \G_h(z_1|z\triangleleft_J) \G_{g-h}(z_1) 
+\lambda\sum_{h=1}^{g-1} \sum_{\substack{I\subset J\\ 1\leq |I|< |J|}}
\G_h(z_1|z\triangleleft_I)\G_{g-h}(z_1|z\triangleleft_{J\backslash I})
\\
&+\lambda \G_{g-1}(z_1|z_1|z\triangleleft_J)
\\
&=-(2\lambda)^{3B-3}\hat{\mathrm{A}}^{\dag g}_{z_1,\dots,z_B}...
\hat{\mathrm{A}}^{\dag g}_{z_1,z_2}\hat{K}_{z_1} \G_g(z_1).
\end{align*}
\end{lemma}
\begin{proof}
Apply $(2\lambda)^{3B-3}\hat{\mathrm{A}}^{\dag g}_{z_1,\dots,z_B}...
\hat{\mathrm{A}}^{\dag g}_{z_1,z_2}$
to equation \eqref{1punktneuV}, use the Leibniz rule
and take \eqref{eqfinalthm} (as definition) into account.
\end{proof}

\begin{lemma}\label{lemma8}
Let $J=\{2,...,B\}$. Then under Assumption~\ref{conj1} one has 
\begin{align*}
&(2\lambda)^3[\hat{K}_ {z_1},\hat{\mathrm{A}}^{\dag g}_{z_1,\dots,z_B}] 
\G_g(z_1|z\triangleleft_{J\backslash {B}})
\\
&=(2\lambda)^6\Bigg[\sum_{l=0}^{3g-4+|J|}\frac{3+2l}{z_1^2z_B^3}
\Bigg(\frac{\varrho_{l+1}}{\varrho_0 z_1^2}
+\frac{3 \varrho_{l+1}}{\varrho_0 z_B^2}
- \frac{1}{z_1^{4+2l}}
-\frac{(5+2l)}{z_B^{4+2l}}\Bigg)\frac{\partial}{\partial \varrho_{l}}
\\
&\qquad
-\sum_{l=0}^{3g-4+|J|}\sum_{k=0}^l\frac{(3+2l)(3+2k)}{z_1^{4+2l-2k}
z_B^{5+2k}}\frac{\partial}{\partial \varrho_{l}}
\\
&\hspace*{2cm}
-\sum_{i\in J\backslash \{B \}}
\frac{1}{\varrho_0 z_1^2z_i z_B^3}
\Big(\frac{1}{z_1^2}+\frac{3}{z_B^2}\Big)
\frac{\partial}{\partial z_i}
	\Bigg]\G_g(z\triangleleft_{J\backslash \{B\}}).
	\end{align*}
\end{lemma}
\begin{proof}
The first term of the lhs, 
$\hat{K}_{z_1}\hat{\mathrm{A}}^{\dag g}_{z_1,\dots,z_B} 
\G_g(z_1|z\triangleleft_{J\backslash B})$, gives by Lemma \ref{lemma5} 
and $\G_g(z\triangleleft_J)=(2\lambda)^3\hat{\mathrm{A}}^{\dag g}_{z_2,\dots,z_B}
\G_g(z\triangleleft_{J\backslash \{B\}})$ the following
\begin{align*}
&\hat{K}_{z_1} (\G_g(z_1|z\triangleleft_J))
\\
&=\frac{(2\lambda)^6}{z_1^2}\Bigg[
\sum_{l=0}^{3g-3+|J|} (3+2l)\sum_{k=0}^{l}\frac{\varrho_k}{z_1^{2+2l-2k}}
\\
&\quad\times \frac{\partial}{\partial \varrho_{l}}\Bigg(
\sum_{l'=0}^{3g-4+|J|} \!\!
\Big({-}\frac{(3{+}2l') \varrho_{l'+1}}{\varrho_0 z_B^3}
+\frac{3{+}2l'}{z_B^{5+2l'}}\Big)
\frac{\partial}{\partial \varrho_{l'}}
\\
&\hspace*{5cm}
+\!\!\! \sum_{i\in J\backslash \{B \}} \frac{1}{\varrho_0 z_i z_B^3}
\frac{\partial}{\partial z_i}\Bigg)
\big(\G_g(z\triangleleft_{J\backslash \{B\}})\big)
\\
&+\sum_{\beta\in J}\frac{1}{z_\beta}\frac{\partial}{\partial z_\beta}    
\Bigg(
\sum_{l'=0}^{3g-4+|J|} \!\!
\Big({-}\frac{(3{+}2l') \varrho_{l'+1}}{\varrho_0 z_B^3}
+\frac{3{+}2l'}{z_B^{5+2l'}}\Big)
\frac{\partial}{\partial \varrho_{l'}}
\\
&\hspace*{5cm}
+ \!\!\! \sum_{i\in J\backslash \{B \}} \frac{1}{\varrho_0 z_i z_B^3}
\frac{\partial}{\partial z_i}\Bigg)
\big(\G_g(z\triangleleft_{J\backslash \{B\}})\big)
\Bigg]
\\
&=\frac{(2\lambda)^6}{z_1^2}\Bigg[
\sum_{l=0}^{3g-4+|J|} \sum_{k=0}^{l+1}
\frac{-(5+2l)(3+2l)\varrho_k}{\varrho_0 z_1^{4+2l-2k}z_B^3}
\frac{\partial}{\partial \varrho_{l}}
+\sum_{l'= 0}^{3g-4+|J|}
\frac{3(3+2l')\varrho_0\varrho_{l'+1}}{\varrho_0^2 
z_1^2 z_B^3}\frac{\partial}{\partial \varrho_{l'}}
\\
&-\!\!\! \sum\limits_{i\in J\backslash \{B \}}
\frac{3\varrho_0}{\varrho_0^2z_1^2 z_i z_B^3}\frac{\partial}{\partial z_i}
+\sum_{l,l'=0}^{3g-4+|J|}
\sum\limits_{k=0}^{l}\frac{(3{+}2l) \varrho_k}{z_1^{2+2l-2k}}
\Big({-}\frac{(3{+}2l') \varrho_{l'+1}}{\varrho_0 z_B^3}
+\frac{3{+}2l'}{z_B^{5+2l'}}\Big)
\frac{\partial^2}{\partial\varrho_l\partial\varrho_{l'}}
\\
&+\sum_{l=0}^{3g-4+|J|}\sum_{k=0}^l\sum_{i\in J\backslash \{B \}}
\frac{(3+2l) \varrho_k}{\varrho_0 z_1^{2+2l-2k} z_iz_B^3 }
\frac{\partial^2}{\partial \varrho_l\partial z_i}
\\
&+\sum_{\beta \in J\backslash \{B \}}\Bigg(
\sum_{l'=0}^{3g-4+|J|}\Big(-\frac{(3+2l') \varrho_{l'+1}}{\varrho_0 
z_\beta z_B^3}+\frac{3+2l'}{z_\beta z_B^{5+2l'}}\Big)
\frac{\partial^2}{\partial z_\beta\partial \varrho_{l'}}
\\
&\hspace*{5cm}
+
\sum_{i\in J\backslash \{B \}} \frac{1}{\varrho_0 z_\beta z_B^3}
\frac{\partial}{\partial z_\beta} \frac{1}{z_i} \frac{\partial}{\partial z_i}
\Bigg)
\\
&-\!\!\!\sum_{l'=0}^{3g-4+|J|}
\Big(-\frac{3(3+2l') \varrho_{l'+1}}{\varrho_0 z_B^5}
+\frac{(3+2l')(5+2l')}{z_B^{7+2l'}}\Big)
\frac{\partial}{\partial \varrho_{l'}}
\\
&\hspace*{5cm}
-\!\!\! \sum_{i\in J\backslash \{B \}} \frac{3}{\varrho_0 z_i z_B^5}
\frac{\partial}{\partial z_i}\Bigg]\G_g(z\triangleleft_{J\backslash \{B\}}).
\end{align*}
We have used that $\G_g(z\triangleleft_{J\backslash \{B\}})$ can 
only depend on $\varrho_l$ for $l\leq 3g-4+|J|$.
For the second term of the lhs, 
$\hat{\mathrm{A}}^{\dag g}_{z_1,\dots,z_B}\hat{K}_{z_1}\G_g(z_1|z\triangleleft_{J\backslash B})$, 
Lemma \ref{lemma5} can also be used with $B-1$ instead of $B$:
\begin{align*}
  &(2\lambda)^3\hat{\mathrm{A}}^{\dag g}_{z_1,\dots,z_B}
  \hat{K}_{z_1}\G_g(z_1|z\triangleleft_{J\backslash B})
\\
&=(2\lambda)^6 
\Bigg(\sum_{l'=0}^{3g-3+|J|}\!\!\! \Big(
{-}\frac{(3+2l') \varrho_{l'+1}}{\varrho_0 z_B^3}+\frac{3+2l'}{z_B^{5+2l'}}
\Big)\frac{\partial}{\partial \varrho_{l'}}
\\
&\hspace*{5cm}
+\sum_{i\in J\backslash \{B \}} \frac{1}{\varrho_0 z_i z_B^3}
\frac{\partial}{\partial z_i}
\frac{1}{\varrho_0 z_1 z_B^3}\frac{\partial}{\partial z_1}\Bigg)
\\
&\times \frac{1}{z_1^2}\Bigg[
\sum_{l= 0}^{3g-4-|J|} (3+2l)\sum_{k=0}^{l}\frac{\varrho_k}{z_1^{2+2l-2k}}
\frac{\partial}{\partial \varrho_l}
+\sum_{\beta\in J\backslash \{B \}}\frac {1  }{z_\beta} 
\frac{\partial}{\partial z_\beta}\Bigg]
\G_g(z\triangleleft_{J\backslash \{B\}})
\\
&=(2\lambda)^6\Bigg[
\sum_{l=0}^{3g-4+|J|}\sum_{k=0}^l\frac{(3+2l)}{z_1^{4+2l-2k}}
\Big( {-}\frac{(3+2k) \varrho_{k+1}}{\varrho_0 z_B^3}
+\frac{3+2k}{z_B^{5+2k}}\Big)\frac{\partial}{\partial \varrho_l}
\\
&+\sum_{l,l'= 0}^{3g-4-|J|}
\sum_{k=0}^l\frac{(3+2l)\varrho_k}{z_1^{4+2l-2k}}
\Big({-}\frac{(3+2l') \varrho_{l'+1}}{\varrho_0 z_B^3}
+\frac{3+2l'}{z_B^{5+2l'}}\Big)
\frac{\partial^2}{\partial\varrho_l\partial\varrho_{l'}}
\\
&+\sum_{i\in J\backslash \{B \}}
\bigg(\sum_{l=0}^{3g-4+|J|}\sum_{k=0}^{l} 
\frac{(3+2l)\varrho_k}{
\varrho_0 z_1^{4+2l-2k} z_i z_B^3}
\frac{\partial^2}{\partial \varrho_l\partial z_i}
+\sum_{\beta\in J\backslash \{B \}}
\frac{1}{\varrho_0 z_1^2 z_i z_B^3}
\frac{\partial}{\partial z_i} \frac{1}{z_\beta} 
\frac{\partial}{\partial z_\beta} \bigg)
\\
&+\sum_{l'=0}^{3g-4+|J|}\sum_{\beta\in J\backslash \{B \}}
\Big( {-}\frac{(3+2l') \varrho_{l'+1}}{\varrho_0 z_1^2 z_\beta z_B^3}
+\frac{3+2l'}{z_1^2 z_\beta z_B^{5+2l'}}\Big)
\frac{\partial^2}{\partial \varrho_{l'}\partial z_\beta}
\\
&-\frac{2}{\varrho_0z_1^4 z_B^3}\bigg(\sum_{l=0}^{3g-4+|J|}\sum_{k=0}^{l} 
\frac{(3+2l)\varrho_k}{z_1^{2+2l-2k}}
\frac{\partial}{\partial \varrho_l}+ \sum_{\beta\in J\backslash \{B \}}
\frac{1}{z_\beta} \frac{\partial}{\partial z_\beta}\bigg)
\\
&-\frac{1}{\varrho_0z_1^2 z_B^3}\sum_{l=0}^{3g-4+|J|}\sum_{k=0}^{l} 
\frac{(3+2l)(2+2l-2k)\varrho_k}{z_1^{4+2l-2k}}
\frac{\partial}{\partial \varrho_l}
\Bigg]\G_g(z\triangleleft_{J\backslash \{B\}}).
\end{align*}
Subtracting the second from the first expression proves the Lemma.
\end{proof}

\begin{lemma}\label{lemma9}
Let $J=\{2,...,B\}$.The linear integral equation \eqref{B>2} is
under Assumption \ref{conj1} and with Definition \ref{DefOp} 
equivalent to the expression
\begin{align*}
0=&(2\lambda)^{3}[\hat{K}_{z_1},\hat{\mathrm{A}}^{\dag g}_{z_1,\dots,z_B}] 
\G_g(z_1|z\triangleleft_{J\backslash \{B\}})
+2\lambda \G_0(z_1|z_B)
\G_g(z_1|z\triangleleft_{J\backslash\{B\}})
\\
&+(2\lambda)^3\frac{1}{z_B}\frac{\partial}{\partial z_B}
\frac{\G_g(z_1|z\triangleleft_{J\backslash\{B\}})
-\G_g(z_B|z\triangleleft_{J\backslash\{B\}})}{z_1^2-z_B^2}.
\end{align*}
\end{lemma}
\begin{proof}
With Lemma \ref{lemma7} we can rewrite the linear 
integral equation \eqref{B>2} in the form
\begin{align}\label{genuseqop}\nonumber
  0&=(2\lambda)^{3B-3}\hat{K}_{z_1}
  \hat{\mathrm{A}}^{\dag g}_{z_1,\dots,z_B} \dots 
\hat{\mathrm{A}}^{\dag g}_{z_1,z_2} \G_g(z_1)
-(2\lambda)^{3B-3} \hat{\mathrm{A}}^{\dag g}_{z_1,\dots,z_B}
\dots \hat{\mathrm{A}}^{\dag g}_{z_1,z_2} \hat{K}_{z_1} \G_g(z_1)
\\
&+2\lambda \G_g(z_1) \G_0(z_1|z\triangleleft_J)
+2\lambda \!\!\! \sum\limits_{\substack{I\subset J\\1\leq|I|<|J|}} \!\!\! 
\G_0(z_1|z\triangleleft_{I})\G_g(z_1|z\triangleleft_{J\backslash I})
\nonumber
\\
&+(2\lambda)^3
\sum\limits_{\beta \in J} 
\frac{1}{z_\beta} 
\frac{\partial}{\partial z_\beta}
\frac{\G_g(z_1|z\triangleleft_{J\backslash\{\beta\}})
-\G_g(z_\beta|z\triangleleft_{J\backslash\{\beta\}})}{z_1^2-z_\beta^2}.
\end{align}
By using this formula for $
\hat{\mathrm{A}}^{\dag g}_{z_1,\dots,z_{B-1}} \dots
\hat{\mathrm{A}}^{\dag g}_{z_1,z_2} \hat{K}_{z_1} \G_g(z_1)$ 
and inserting it back into \eqref{genuseqop} gives
\begin{align*}
0=&
(2\lambda)^{3B-3}[\hat{K}_{z_1}, \hat{\mathrm{A}}^{\dag g}_{z_1,\dots,z_B}]
\hat{\mathrm{A}}^{\dag g}_{z_1,\dots,z_{B-1}} \dots 
\hat{\mathrm{A}}^{\dag g}_{z_1,z_2} \G_g(z_1)
\\
&+2\lambda\G_g(z_1)\G_0(z_1|z\triangleleft_J)
-(2\lambda)^4\hat{\mathrm{A}}^{\dag g}_{z_1,\dots,z_B}
(\G_g(z_1)\G_0(z_1|z\triangleleft_{J\backslash B}))
\\
&+2\lambda \!\!\! \sum\limits_{\substack{I\subset J\\1\leq|I|<|J|}} \!\!\!
\G_0(z_1|z\triangleleft_{I})\G_g(z_1|z\triangleleft_{J\backslash I})
\\
&-(2\lambda)^4\hat{\mathrm{A}}^{\dag g}_{z_1,\dots,z_B} \!\!\! 
\sum\limits_{\substack{I\subset J\backslash\{B\}\\1\leq|I|<|J|-1}}
\!\!\! \G_0(z_1|z\triangleleft_{I})\G_g(z_1|z\triangleleft_{J\backslash \{I,B\}})
\\
&+(2\lambda)^3
\sum\limits_{\beta \in J} 
\frac{1}{z_\beta} 
\frac{\partial}{\partial z_\beta}
\frac{\G_g(z_1|z\triangleleft_{J\backslash\{\beta\}})
-\G_g(z_\beta|z\triangleleft_{J\backslash\{\beta\}})}{z_1^2-z_\beta^2}
\\
&-(2\lambda)^6\hat{\mathrm{A}}^{\dag g}_{z_1,\dots,z_B}
\sum\limits_{\beta \in J\setminus\{B\}} 
\frac{1}{z_\beta} 
\frac{\partial}{\partial z_\beta}
\frac{\G_g(z_1|z\triangleleft_{J\backslash\{\beta,B\}})
-\G_g(z_\beta|z\triangleleft_{J\backslash\{\beta,B\}})}{z_1^2-z_\beta^2}.
\end{align*}
The second and third line break down to
$2\lambda \G_0(z_1|z_B)\G_g(z_1|z\triangleleft_{J\backslash\{B\}}).$
Therefore, the assertion follows if we can show that, 
in the fourth line, the part of the sum which excludes 
$\beta=B$ cancels with the fifth line.
This is true because of  
\[
\Big[\hat{\mathrm{A}}^{\dag g}_{z_1,\dots,z_B},\frac{1}{z_\beta}\frac{\partial}{\partial z_\beta}
\Big]=0 \quad\text{and} \quad 
\hat{\mathrm{A}}^{\dag g}_{z_1,\dots,z_B} \frac{1}{z_1^2-z_\beta^2}=0.
\]
Consequently, the linear integral equation can be written 
by operators of the form given in this Lemma.
\end{proof}


\begin{thebibliography}{46}
\ifx \bisbn   \undefined \def \bisbn  #1{ISBN #1}\fi
\ifx \binits  \undefined \def \binits#1{#1}\fi
\ifx \bauthor  \undefined \def \bauthor#1{#1}\fi
\ifx \batitle  \undefined \def \batitle#1{#1}\fi
\ifx \bjtitle  \undefined \def \bjtitle#1{#1}\fi
\ifx \bvolume  \undefined \def \bvolume#1{\textbf{#1}}\fi
\ifx \byear  \undefined \def \byear#1{#1}\fi
\ifx \bissue  \undefined \def \bissue#1{#1}\fi
\ifx \bfpage  \undefined \def \bfpage#1{#1}\fi
\ifx \blpage  \undefined \def \blpage #1{#1}\fi
\ifx \burl  \undefined \def \burl#1{\textsf{#1}}\fi
\ifx \doiurl  \undefined \def \doiurl#1{\url{https://doi.org/#1}}\fi
\ifx \betal  \undefined \def \betal{\textit{et al.}}\fi
\ifx \binstitute  \undefined \def \binstitute#1{#1}\fi
\ifx \binstitutionaled  \undefined \def \binstitutionaled#1{#1}\fi
\ifx \bctitle  \undefined \def \bctitle#1{#1}\fi
\ifx \beditor  \undefined \def \beditor#1{#1}\fi
\ifx \bpublisher  \undefined \def \bpublisher#1{#1}\fi
\ifx \bbtitle  \undefined \def \bbtitle#1{#1}\fi
\ifx \bedition  \undefined \def \bedition#1{#1}\fi
\ifx \bseriesno  \undefined \def \bseriesno#1{#1}\fi
\ifx \blocation  \undefined \def \blocation#1{#1}\fi
\ifx \bsertitle  \undefined \def \bsertitle#1{#1}\fi
\ifx \bsnm \undefined \def \bsnm#1{#1}\fi
\ifx \bsuffix \undefined \def \bsuffix#1{#1}\fi
\ifx \bparticle \undefined \def \bparticle#1{#1}\fi
\ifx \barticle \undefined \def \barticle#1{#1}\fi
\bibcommenthead
\ifx \bconfdate \undefined \def \bconfdate #1{#1}\fi
\ifx \botherref \undefined \def \botherref #1{#1}\fi
\ifx \url \undefined \def \url#1{\textsf{#1}}\fi
\ifx \bchapter \undefined \def \bchapter#1{#1}\fi
\ifx \bbook \undefined \def \bbook#1{#1}\fi
\ifx \bcomment \undefined \def \bcomment#1{#1}\fi
\ifx \oauthor \undefined \def \oauthor#1{#1}\fi
\ifx \citeauthoryear \undefined \def \citeauthoryear#1{#1}\fi
\ifx \endbibitem  \undefined \def \endbibitem {}\fi
\ifx \bconflocation  \undefined \def \bconflocation#1{#1}\fi
\ifx \arxivurl  \undefined \def \arxivurl#1{\textsf{#1}}\fi
\csname PreBibitemsHook\endcsname

\bibitem{Witten:1990hr}
\begin{barticle}
\bauthor{\bsnm{Witten}, \binits{E.}}:
\batitle{{Two-dimensional gravity and intersection theory on moduli space}}.
\bjtitle{Surveys Diff. Geom.}
\bvolume{1},
\bfpage{243}--\blpage{310}
(\byear{1991}).
\doiurl{10.4310/SDG.1990.v1.n1.a5}
\end{barticle}
\endbibitem

\bibitem{Kontsevich:1992ti}
\begin{barticle}
\bauthor{\bsnm{Kontsevich}, \binits{M.}}:
\batitle{{Intersection theory on the moduli space of curves and the matrix Airy
  function}}.
\bjtitle{Commun. Math. Phys.}
\bvolume{147},
\bfpage{1}--\blpage{23}
(\byear{1992}).
\doiurl{10.1007/BF02099526}
\end{barticle}
\endbibitem

\bibitem{Alexandrov:2010bn}
\begin{barticle}
\bauthor{\bsnm{Alexandrov}, \binits{A.}}:
\batitle{{Cut-and-Join operator representation for Kontsewich-Witten
  tau-function}}.
\bjtitle{Mod. Phys. Lett.}
\bvolume{A26},
\bfpage{2193}--\blpage{2199}
(\byear{2011})
{\href{https://arxiv.org/abs/1009.4887}{{arXiv:1009.4887}}}
{[hep-th]}.
\doiurl{10.1142/S0217732311036607}
\end{barticle}
\endbibitem

\bibitem{Brezin:1977sv}
\begin{barticle}
\bauthor{\bsnm{Brezin}, \binits{E.}},
\bauthor{\bsnm{Itzykson}, \binits{C.}},
\bauthor{\bsnm{Parisi}, \binits{G.}},
\bauthor{\bsnm{Zuber}, \binits{J.B.}}:
\batitle{{Planar diagrams}}.
\bjtitle{Commun. Math. Phys.}
\bvolume{59},
\bfpage{35}
(\byear{1978}).
\doiurl{10.1007/BF01614153}
\end{barticle}
\endbibitem

\bibitem{DiFrancesco:2004qj}
\begin{bchapter}
\bauthor{\bsnm{Di~Francesco}, \binits{P.}}:
\bctitle{{2D quantum gravity, matrix models and graph combinatorics}}.
In: \bbtitle{Application of Random Matrices in Physics. Proc. Les Houches},
pp. \bfpage{33}--\blpage{88}
(\byear{2004})
\end{bchapter}
\endbibitem

\bibitem{Gross:1989vs}
\begin{barticle}
\bauthor{\bsnm{Gross}, \binits{D.J.}},
\bauthor{\bsnm{Migdal}, \binits{A.A.}}:
\batitle{{Nonperturbative two-dimensional quantum gravity}}.
\bjtitle{Phys. Rev. Lett.}
\bvolume{64},
\bfpage{127}
(\byear{1990}).
\doiurl{10.1103/PhysRevLett.64.127}
\end{barticle}
\endbibitem

\bibitem{DiFrancesco:1993cyw}
\begin{barticle}
\bauthor{\bsnm{Di~Francesco}, \binits{P.}},
\bauthor{\bsnm{Ginsparg}, \binits{P.H.}},
\bauthor{\bsnm{Zinn-Justin}, \binits{J.}}:
\batitle{{2-D Gravity and random matrices}}.
\bjtitle{Phys. Rept.}
\bvolume{254},
\bfpage{1}--\blpage{133}
(\byear{1995})
{\href{https://arxiv.org/abs/hep-th/9306153}{{arXiv:hep-th/9306153}}}
{[hep-th]}.
\doiurl{10.1016/0370-1573(94)00084-G}
\end{barticle}
\endbibitem

\bibitem{Langmann:2002cc}
\begin{barticle}
\bauthor{\bsnm{Langmann}, \binits{E.}},
\bauthor{\bsnm{Szabo}, \binits{R.J.}}:
\batitle{{Duality in scalar field theory on noncommutative phase spaces}}.
\bjtitle{Phys. Lett.}
\bvolume{B533},
\bfpage{168}--\blpage{177}
(\byear{2002})
{\href{https://arxiv.org/abs/hep-th/0202039}{{arXiv:hep-th/0202039}}}
{[hep-th]}.
\doiurl{10.1016/S0370-2693(02)01650-7}
\end{barticle}
\endbibitem

\bibitem{Langmann:2003if}
\begin{barticle}
\bauthor{\bsnm{Langmann}, \binits{E.}},
\bauthor{\bsnm{Szabo}, \binits{R.J.}},
\bauthor{\bsnm{Zarembo}, \binits{K.}}:
\batitle{{Exact solution of quantum field theory on noncommutative phase
  spaces}}.
\bjtitle{JHEP}
\bvolume{01},
\bfpage{017}
(\byear{2004})
{\href{https://arxiv.org/abs/hep-th/0308043}{{arXiv:hep-th/0308043}}}
{[hep-th]}.
\doiurl{10.1088/1126-6708/2004/01/017}
\end{barticle}
\endbibitem

\bibitem{Grosse:2004yu}
\begin{barticle}
\bauthor{\bsnm{Grosse}, \binits{H.}},
\bauthor{\bsnm{Wulkenhaar}, \binits{R.}}:
\batitle{{Renormalisation of $\phi^4$-theory on noncommutative $\mathbb{R}^4$
  in the matrix base}}.
\bjtitle{Commun. Math. Phys.}
\bvolume{256},
\bfpage{305}--\blpage{374}
(\byear{2005})
{\href{https://arxiv.org/abs/hep-th/0401128}{{arXiv:hep-th/0401128}}}
{[hep-th]}.
\doiurl{10.1007/s00220-004-1285-2}
\end{barticle}
\endbibitem

\bibitem{Grosse:2005ig}
\begin{barticle}
\bauthor{\bsnm{Grosse}, \binits{H.}},
\bauthor{\bsnm{Steinacker}, \binits{H.}}:
\batitle{{Renormalization of the noncommutative $\phi^3$-model through the
  Kontsevich model}}.
\bjtitle{Nucl. Phys.}
\bvolume{B746},
\bfpage{202}--\blpage{226}
(\byear{2006})
{\href{https://arxiv.org/abs/hep-th/0512203}{{arXiv:hep-th/0512203}}}
{[hep-th]}.
\doiurl{10.1016/j.nuclphysb.2006.04.007}
\end{barticle}
\endbibitem

\bibitem{Grosse:2006qv}
\begin{barticle}
\bauthor{\bsnm{Grosse}, \binits{H.}},
\bauthor{\bsnm{Steinacker}, \binits{H.}}:
\batitle{{A Nontrivial solvable noncommutative $\phi^3$-model in 4
  dimensions}}.
\bjtitle{JHEP}
\bvolume{08},
\bfpage{008}
(\byear{2006})
{\href{https://arxiv.org/abs/hep-th/0603052}{{arXiv:hep-th/0603052}}}
{[hep-th]}.
\doiurl{10.1088/1126-6708/2006/08/008}
\end{barticle}
\endbibitem

\bibitem{Grosse:2006tc}
\begin{barticle}
\bauthor{\bsnm{Grosse}, \binits{H.}},
\bauthor{\bsnm{Steinacker}, \binits{H.}}:
\batitle{{Exact renormalization of a noncommutative $\phi^3$-model in 6
  dimensions}}.
\bjtitle{Adv. Theor. Math. Phys.}
\bvolume{12}(\bissue{3}),
\bfpage{605}--\blpage{639}
(\byear{2008})
{\href{https://arxiv.org/abs/hep-th/0607235}{{arXiv:hep-th/0607235}}}
{[hep-th]}.
\doiurl{10.4310/ATMP.2008.v12.n3.a4}
\end{barticle}
\endbibitem

\bibitem{Disertori:2006nq}
\begin{barticle}
\bauthor{\bsnm{Disertori}, \binits{M.}},
\bauthor{\bsnm{Gurau}, \binits{R.}},
\bauthor{\bsnm{Magnen}, \binits{J.}},
\bauthor{\bsnm{Rivasseau}, \binits{V.}}:
\batitle{{Vanishing of beta function of non commutative $\Phi^4_4$ theory to
  all orders}}.
\bjtitle{Phys. Lett.}
\bvolume{B649},
\bfpage{95}--\blpage{102}
(\byear{2007})
{\href{https://arxiv.org/abs/hep-th/0612251}{{arXiv:hep-th/0612251}}}
{[hep-th]}.
\doiurl{10.1016/j.physletb.2007.04.007}
\end{barticle}
\endbibitem

\bibitem{Grosse:2012uv}
\begin{barticle}
\bauthor{\bsnm{Grosse}, \binits{H.}},
\bauthor{\bsnm{Wulkenhaar}, \binits{R.}}:
\batitle{{Self-dual noncommutative $\phi^4$-theory in four dimensions is a
  non-perturbatively solvable and non-trivial quantum field theory}}.
\bjtitle{Commun. Math. Phys.}
\bvolume{329},
\bfpage{1069}--\blpage{1130}
(\byear{2014})
{\href{https://arxiv.org/abs/1205.0465}{{arXiv:1205.0465}}}
{[math-ph]}.
\doiurl{10.1007/s00220-014-1906-3}
\end{barticle}
\endbibitem

\bibitem{Witten:1991mn}
\begin{bchapter}
\bauthor{\bsnm{Witten}, \binits{E.}}:
\bctitle{{On the Kontsevich model and other models of two-dimensional
  gravity}}.
In: \bbtitle{Differential Geometric Methods in Theoretical Physics. Proc. 20th
  Int. Conf., New York, 1991},
pp. \bfpage{176}--\blpage{216}.
\bpublisher{World Sci. Publ.},
\blocation{River Edge, NJ}
(\byear{1992})
\end{bchapter}
\endbibitem

\bibitem{Lando:2004??}
\begin{bbook}
\bauthor{\bsnm{Lando}, \binits{S.K.}},
\bauthor{\bsnm{Zvonkin}, \binits{A.K.}}:
\bbtitle{Graphs on Surfaces and Their Applications}.
\bsertitle{Encyclopaedia of Mathematical Sciences},
vol. \bseriesno{141},
p. \bfpage{455}.
\bpublisher{Springer},
\blocation{Berlin Heidelberg}
(\byear{2004}).
\doiurl{10.1007/978-3-540-38361-1}.
\bcomment{With an appendix by Don B. Zagier, Low-Dimensional Topology, II}
\end{bbook}
\endbibitem

\bibitem{Eynard:2016yaa}
\begin{bbook}
\bauthor{\bsnm{Eynard}, \binits{B.}}:
\bbtitle{Counting Surfaces}.
\bsertitle{Progress in Mathematical Physics},
vol. \bseriesno{70}.
\bpublisher{Springer},
\blocation{Switzerland}
(\byear{2016}).
\doiurl{10.1007/978-3-7643-8797-6}
\end{bbook}
\endbibitem

\bibitem{Chekhov:2006vd}
\begin{barticle}
\bauthor{\bsnm{Chekhov}, \binits{L.}},
\bauthor{\bsnm{Eynard}, \binits{B.}},
\bauthor{\bsnm{Orantin}, \binits{N.}}:
\batitle{{Free energy topological expansion for the 2-matrix model}}.
\bjtitle{JHEP}
\bvolume{12},
\bfpage{053}
(\byear{2006})
{\href{https://arxiv.org/abs/math-ph/0603003}{{arXiv:math-ph/0603003}}}
{[math-ph]}.
\doiurl{10.1088/1126-6708/2006/12/053}
\end{barticle}
\endbibitem

\bibitem{Eynard:2007kz}
\begin{barticle}
\bauthor{\bsnm{Eynard}, \binits{B.}},
\bauthor{\bsnm{Orantin}, \binits{N.}}:
\batitle{{Invariants of algebraic curves and topological expansion}}.
\bjtitle{Commun. Num. Theor. Phys.}
\bvolume{1},
\bfpage{347}--\blpage{452}
(\byear{2007})
{\href{https://arxiv.org/abs/math-ph/0702045}{{arXiv:math-ph/0702045}}}
{[math-ph]}.
\doiurl{10.4310/CNTP.2007.v1.n2.a4}
\end{barticle}
\endbibitem

\bibitem{Eynard:2014zxa}
\begin{botherref}
\oauthor{\bsnm{Eynard}, \binits{B.}}:
{A short overview of the "Topological recursion"}
(2014)
{\href{https://arxiv.org/abs/1412.3286}{{arXiv:1412.3286}}}
{[math-ph]}
\end{botherref}
\endbibitem

\bibitem{Ambjorn:1992gw}
\begin{barticle}
\bauthor{\bsnm{Ambj{{\o}}rn}, \binits{J.}},
\bauthor{\bsnm{Chekhov}, \binits{L.}},
\bauthor{\bsnm{Kristjansen}, \binits{C.F.}},
\bauthor{\bsnm{Makeenko}, \binits{{\relax Yu}.}}:
\batitle{{Matrix model calculations beyond the spherical limit}}.
\bjtitle{Nucl. Phys.}
\bvolume{B404},
\bfpage{127}--\blpage{172}
(\byear{1993})
{\href{https://arxiv.org/abs/hep-th/9302014}{{arXiv:hep-th/9302014}}}
{[hep-th]}.
\doiurl{10.1016/0550-3213(93)90476-6}.
\bcomment{{\href{http://dx.doi.org/10.1016/0550-3213(95)00391-5}{doi:10.1016/0550-3213(95)00391-5}}.[Erratum:
  Nucl. Phys.B449,681(1995)]}
\end{barticle}
\endbibitem

\bibitem{Grosse:2019jnv}
\begin{botherref}
\oauthor{\bsnm{Grosse}, \binits{H.}},
\oauthor{\bsnm{Hock}, \binits{A.}},
\oauthor{\bsnm{Wulkenhaar}, \binits{R.}}:
{Solution of all quartic matrix models}
(2019)
{\href{https://arxiv.org/abs/1906.04600}{{arXiv:1906.04600}}}
{[math-ph]}
\end{botherref}
\endbibitem

\bibitem{Panzer:2018tvy}
\begin{barticle}
\bauthor{\bsnm{Panzer}, \binits{E.}},
\bauthor{\bsnm{Wulkenhaar}, \binits{R.}}:
\batitle{{Lambert-W Solves the Noncommutative $\varPhi ^4$-Model}}.
\bjtitle{Commun. Math. Phys.}
\bvolume{374}(\bissue{3}),
\bfpage{1935}--\blpage{1961}
(\byear{2020})
{\href{https://arxiv.org/abs/1807.02945}{{arXiv:1807.02945}}}
{[math-ph]}.
\doiurl{10.1007/s00220-019-03592-4}
\end{barticle}
\endbibitem

\bibitem{Grosse:2019qps}
\begin{barticle}
\bauthor{\bsnm{Grosse}, \binits{H.}},
\bauthor{\bsnm{Hock}, \binits{A.}},
\bauthor{\bsnm{Wulkenhaar}, \binits{R.}}:
\batitle{{Solution of the self-dual $\Phi^4$ QFT-model on four-dimensional
  Moyal space}}.
\bjtitle{JHEP}
\bvolume{01},
\bfpage{081}
(\byear{2020})
{\href{https://arxiv.org/abs/1908.04543}{{arXiv:1908.04543}}}
{[math-ph]}.
\doiurl{10.1007/JHEP01(2020)081}
\end{barticle}
\endbibitem

\bibitem{deJong:2019oez}
\begin{barticle}
\bauthor{\bparticle{de} \bsnm{Jong}, \binits{J.}},
\bauthor{\bsnm{Hock}, \binits{A.}},
\bauthor{\bsnm{Wulkenhaar}, \binits{R.}}:
\batitle{{Nested Catalan tables and a recurrence relation in noncommutative
  quantum field theory}}.
\bjtitle{Ann. Inst. H. Poincar\'e D Comb. Phys. Interact.}
\bvolume{9},
\bfpage{47}--\blpage{72}
(\byear{2022})
{\href{https://arxiv.org/abs/1904.11231}{{arXiv:1904.11231}}}
{[math-ph]}.
\doiurl{10.4171/AIHPD/113}
\end{barticle}
\endbibitem

\bibitem{Grosse:2016pob}
\begin{barticle}
\bauthor{\bsnm{Grosse}, \binits{H.}},
\bauthor{\bsnm{Sako}, \binits{A.}},
\bauthor{\bsnm{Wulkenhaar}, \binits{R.}}:
\batitle{{Exact solution of matricial $\Phi^3_2$ quantum field theory}}.
\bjtitle{Nucl. Phys.}
\bvolume{B925},
\bfpage{319}--\blpage{347}
(\byear{2017})
{\href{https://arxiv.org/abs/1610.00526}{{arXiv:1610.00526}}}
{[math-ph]}.
\doiurl{10.1016/j.nuclphysb.2017.10.010}
\end{barticle}
\endbibitem

\bibitem{Grosse:2016qmk}
\begin{barticle}
\bauthor{\bsnm{Grosse}, \binits{H.}},
\bauthor{\bsnm{Sako}, \binits{A.}},
\bauthor{\bsnm{Wulkenhaar}, \binits{R.}}:
\batitle{{The $\Phi^3_4$ and $\Phi^3_6$ matricial QFT models have reflection
  positive two-point function}}.
\bjtitle{Nucl. Phys.}
\bvolume{B926},
\bfpage{20}--\blpage{48}
(\byear{2018})
{\href{https://arxiv.org/abs/1612.07584}{{arXiv:1612.07584}}}
{[math-ph]}.
\doiurl{10.1016/j.nuclphysb.2017.10.022}
\end{barticle}
\endbibitem

\bibitem{Makeenko:1991ec}
\begin{barticle}
\bauthor{\bsnm{Makeenko}, \binits{{\relax Yu}.}},
\bauthor{\bsnm{Semenoff}, \binits{G.W.}}:
\batitle{{Properties of Hermitean matrix models in an external field}}.
\bjtitle{Mod. Phys. Lett.}
\bvolume{A6},
\bfpage{3455}--\blpage{3466}
(\byear{1991}).
\doiurl{10.1142/S0217732391003985}
\end{barticle}
\endbibitem

\bibitem{Hock:2020rje}
\begin{botherref}
\oauthor{\bsnm{Hock}, \binits{A.}}:
{Matrix Field Theory}
(2020)
{\href{https://arxiv.org/abs/2005.07525}{{arXiv:2005.07525}}}
{[math-ph]}.
(PhD thesis, WWU M{\"u}nster)
\end{botherref}
\endbibitem

\bibitem{Branahl:2020yru}
\begin{barticle}
\bauthor{\bsnm{Branahl}, \binits{J.}},
\bauthor{\bsnm{Hock}, \binits{A.}},
\bauthor{\bsnm{Wulkenhaar}, \binits{R.}}:
\batitle{{Blobbed topological recursion of the quartic Kontsevich model I: Loop
  equations and conjectures}}.
\bjtitle{Commun. Math. Phys.}
\bvolume{393}(\bissue{3}),
\bfpage{1529}--\blpage{1582}
(\byear{2022})
{\href{https://arxiv.org/abs/2008.12201}{{arXiv:2008.12201}}}
{[math-ph]}.
\doiurl{10.1007/s00220-022-04392-z}
\end{barticle}
\endbibitem

\bibitem{Borot:2015hna}
\begin{barticle}
\bauthor{\bsnm{Borot}, \binits{G.}},
\bauthor{\bsnm{Shadrin}, \binits{S.}}:
\batitle{{Blobbed topological recursion: properties and applications}}.
\bjtitle{Math. Proc. Cambridge Phil. Soc.}
\bvolume{162}(\bissue{1}),
\bfpage{39}--\blpage{87}
(\byear{2017})
{\href{https://arxiv.org/abs/1502.00981}{{arXiv:1502.00981}}}
{[math-ph]}.
\doiurl{10.1017/S0305004116000323}
\end{barticle}
\endbibitem

\bibitem{Hock:2021tbl}
\begin{botherref}
\oauthor{\bsnm{Hock}, \binits{A.}},
\oauthor{\bsnm{Wulkenhaar}, \binits{R.}}:
{Blobbed topological recursion of the quartic Kontsevich model II: Genus=0}
(2021)
{\href{https://arxiv.org/abs/2103.13271}{{arXiv:2103.13271}}}
{[math-ph]}
\end{botherref}
\endbibitem

\bibitem{Eynard:2013csa}
\begin{botherref}
\oauthor{\bsnm{Eynard}, \binits{B.}},
\oauthor{\bsnm{Orantin}, \binits{N.}}:
{About the x-y symmetry of the $F_g$ algebraic invariants}
(2013)
{\href{https://arxiv.org/abs/1311.4993}{{arXiv:1311.4993}}}
{[math-ph]}
\end{botherref}
\endbibitem

\bibitem{Hock:2022wer}
\begin{botherref}
\oauthor{\bsnm{Hock}, \binits{A.}}:
{On the $x$-$y$ Symmetry of Correlators in Topological Recursion via Loop
  Insertion Operator}
(2022)
{\href{https://arxiv.org/abs/2201.05357}{{arXiv:2201.05357}}}
{[math-ph]}
\end{botherref}
\endbibitem

\bibitem{Branahl:2020uxs}
\begin{barticle}
\bauthor{\bsnm{Branahl}, \binits{J.}},
\bauthor{\bsnm{Hock}, \binits{A.}},
\bauthor{\bsnm{Wulkenhaar}, \binits{R.}}:
\batitle{{Perturbative and geometric analysis of the quartic Kontsevich
  model}}.
\bjtitle{SIGMA}
\bvolume{17},
\bfpage{085}
(\byear{2021})
{\href{https://arxiv.org/abs/2012.02622}{{arXiv:2012.02622}}}
{[math-ph]}.
\doiurl{10.3842/SIGMA.2021.085}
\end{barticle}
\endbibitem

\bibitem{Branahl:2021uea}
\begin{botherref}
\oauthor{\bsnm{Branahl}, \binits{J.}},
\oauthor{\bsnm{Hock}, \binits{A.}}:
{Genus one free energy contribution to the quartic Kontsevich model}
(2021)
{\href{https://arxiv.org/abs/2111.05411}{{arXiv:2111.05411}}}
{[math-ph]}
\end{botherref}
\endbibitem

\bibitem{Branahl:2021slr}
\begin{barticle}
\bauthor{\bsnm{Branahl}, \binits{J.}},
\bauthor{\bsnm{Hock}, \binits{A.}},
\bauthor{\bsnm{Grosse}, \binits{H.}},
\bauthor{\bsnm{Wulkenhaar}, \binits{R.}}:
\batitle{{From scalar fields on quantum spaces to blobbed topological
  recursion}}.
\bjtitle{J. Phys. A}
\bvolume{55}(\bissue{42}),
\bfpage{423001}
(\byear{2022})
{\href{https://arxiv.org/abs/2110.11789}{{arXiv:2110.11789}}}
{[hep-th]}.
\doiurl{10.1088/1751-8121/ac9260}
\end{barticle}
\endbibitem

\bibitem{Branahl:2022uge}
\begin{botherref}
\oauthor{\bsnm{Branahl}, \binits{J.}},
\oauthor{\bsnm{Hock}, \binits{A.}}:
{Complete solution of the LSZ Model via Topological Recursion}
(2022)
{\href{https://arxiv.org/abs/2205.12166}{{arXiv:2205.12166}}}
{[math-ph]}
\end{botherref}
\endbibitem

\bibitem{Hock:2018wup}
\begin{barticle}
\bauthor{\bsnm{Hock}, \binits{A.}},
\bauthor{\bsnm{Wulkenhaar}, \binits{R.}}:
\batitle{{Noncommutative 3-colour scalar quantum field theory model in 2D}}.
\bjtitle{Eur. Phys. J.}
\bvolume{C78}(\bissue{7}),
\bfpage{580}
(\byear{2018})
{\href{https://arxiv.org/abs/1804.06075}{{arXiv:1804.06075}}}
{[math-ph]}.
\doiurl{10.1140/epjc/s10052-018-6042-3}
\end{barticle}
\endbibitem

\bibitem{Itzykson:1992ya}
\begin{barticle}
\bauthor{\bsnm{Itzykson}, \binits{C.}},
\bauthor{\bsnm{Zuber}, \binits{J.B.}}:
\batitle{{Combinatorics of the modular group. 2. The Kontsevich integrals}}.
\bjtitle{Int. J. Mod. Phys.}
\bvolume{A7},
\bfpage{5661}--\blpage{5705}
(\byear{1992})
{\href{https://arxiv.org/abs/hep-th/9201001}{{arXiv:hep-th/9201001}}}
{[hep-th]}.
\doiurl{10.1142/S0217751X92002581}
\end{barticle}
\endbibitem

\bibitem{admcycles}
\begin{botherref}
\oauthor{\bsnm{Delecroix}, \binits{V.}},
\oauthor{\bsnm{Schmitt}, \binits{J.}},
\oauthor{\bparticle{van} \bsnm{Zelm}, \binits{J.}}:
admcycles - a sage package for calculations in the tautological ring of the
  moduli space of stable curves.
J. Soft. Alg. Geom.,
89--112
(2021)
{\href{https://arxiv.org/abs/2002.01709}{{2002.01709}}}.
\doiurl{10.2140/jsag.2021.11.89}
\end{botherref}
\endbibitem

\bibitem{Xu:2007??}
\begin{botherref}
\oauthor{\bsnm{Xu}, \binits{H.}}:
\verb|psrecursion.mw| -- a {M}aple program to compute intersection indices
  \url{http://www.cms.zju.edu.cn/news.asp?id=1275&ColumnName=pdfbook&Version=english}
(2007)
\end{botherref}
\endbibitem

\bibitem{Faber:1999??}
\begin{bchapter}
\bauthor{\bsnm{Faber}, \binits{C.}}:
\bctitle{Algorithms for computing intersection numbers on moduli spaces of
  curves, with an application to the class of the locus of {J}acobians}.
In: \bbtitle{New Trends in Algebraic Geometry}.
\bsertitle{London Math. Soc. Lecture Note Ser.},
vol. \bseriesno{264},
pp. \bfpage{93}--\blpage{109}.
\bpublisher{Cambridge Univ. Press},
\blocation{Cambridge}
(\byear{1999}).
\doiurl{10.1017/CBO9780511721540.006}
\end{bchapter}
\endbibitem

\bibitem{Zimmermann:1969jj}
\begin{barticle}
\bauthor{\bsnm{Zimmermann}, \binits{W.}}:
\batitle{{Convergence of Bogolyubov's method of renormalization in momentum
  space}}.
\bjtitle{Commun. Math. Phys.}
\bvolume{15},
\bfpage{208}--\blpage{234}
(\byear{1969}).
\doiurl{10.1007/BF01645676}
\end{barticle}
\endbibitem

\bibitem{Eynard:2019mps}
\begin{botherref}
\oauthor{\bsnm{Eynard}, \binits{B.}}:
{Large genus behavior of topological recursion}
(2019)
{\href{https://arxiv.org/abs/1905.11270}{{arXiv:1905.11270}}}
{[math-ph]}
\end{botherref}
\endbibitem

\end{thebibliography}

\end{document}